\documentclass[dvipsnames,11pt]{article}
\usepackage{amsmath,amssymb,amsthm,color,bm,mathrsfs,extarrows,tikz,graphicx,mathtools,enumitem,fancybox,makecell}
\usepackage[square,sort,comma,numbers]{natbib}
\usepackage{subcaption}
\usepackage[ruled]{algorithm2e}
\usepackage[margin=1in]{geometry}
\usepackage{xcolor,bbm}
\usepackage[T1]{fontenc}
\usepackage[pagebackref]{hyperref}
\usepackage{comment}
\usepackage{soul}
\usepackage{ifthen}
\usepackage{mathdots}
\usepackage{enumitem}
\usepackage{mathpazo}

\allowdisplaybreaks

\setlist[itemize]{itemsep=0pt}
\setlist[enumerate]{itemsep=0pt}
\hypersetup{colorlinks=true,urlcolor=blue,linkcolor=blue,citecolor=[rgb]{.42,.56,.14},}
\usetikzlibrary{decorations.pathreplacing,decorations.markings,decorations.pathmorphing,decorations.shapes,arrows.meta,positioning,patterns}

\usepackage[capitalise,nameinlink]{cleveref}
\Crefname{lemma}{Lemma}{Lemmas}
\Crefname{fact}{Fact}{Facts}
\Crefname{theorem}{Theorem}{Theorems}
\Crefname{corollary}{Corollary}{Corollaries}
\Crefname{claim}{Claim}{Claims}
\Crefname{example}{Example}{Examples}
\Crefname{problem}{Problem}{Problems}
\Crefname{definition}{Definition}{Definitions}
\Crefname{notation}{Notation}{Notations}
\Crefname{assumption}{Assumption}{Assumptions}
\Crefname{subsection}{Subsection}{Subsections}
\Crefname{section}{Section}{Sections}
\Crefformat{equation}{(#2#1#3)}
\Crefname{figure}{Figure}{Figures}

\newtheorem{theorem}{Theorem}[section]
\newtheorem*{theorem*}{Theorem}

\newtheorem{proposition}[theorem]{Proposition}
\newtheorem*{proposition*}{Proposition}

\newtheorem*{property*}{Property}
\newtheorem{lemma}[theorem]{Lemma}
\newtheorem*{lemma*}{Lemma}
\newtheorem{corollary}[theorem]{Corollary}
\newtheorem*{corollary*}{Corollary}
\newtheorem*{conjecture*}{Conjecture}
\newtheorem{fact}[theorem]{Fact}
\newtheorem*{fact*}{Fact}

\newtheorem*{exercise*}{Exercise}
\newtheorem{hypothesis}[theorem]{Hypothesis}
\newtheorem*{hypothesis*}{Hypothesis}

\theoremstyle{definition}
\newtheorem{definition}[theorem]{Definition}

\newtheorem{exercise-easy}[theorem]{Exercise}
\newtheorem{exercise-med}[theorem]{Exercise}
\newtheorem{exercise-hard}[theorem]{Exercise$^\star$}
\newtheorem{claim}[theorem]{Claim}
\newtheorem*{claim*}{Claim}

\newtheorem{remark}[theorem]{Remark}
\newtheorem*{remark*}{Remark}

\newtheorem*{observation*}{Observation}

\DeclareSymbolFont{extraup}{U}{zavm}{m}{n}
\DeclareMathSymbol{\varheart}{\mathalpha}{extraup}{86}
\DeclareMathSymbol{\vardiamond}{\mathalpha}{extraup}{87}

\DeclareMathOperator*{\E}{\mathbb E}

\renewcommand{\Pr}{\operatorname*{\mathbf{Pr}}}

\newcommand{\eps}{\varepsilon}
\newcommand{\abs}[1]{\left| #1 \right|}

\newcommand{\abra}[1]{\left\langle #1 \right\rangle}
\newcommand{\pbra}[1]{\left( #1 \right)}
\newcommand{\sbra}[1]{\left[ #1 \right]}
\newcommand{\cbra}[1]{\left\{ #1 \right\}}
\newcommand{\floorbra}[1]{\left\lfloor #1 \right\rfloor}
\newcommand{\ceilbra}[1]{\left\lceil #1 \right\rceil}

\renewcommand{\mid}{\,\middle\vert\,}

\newcommand{\bin}{\{0,1\}}

\newcommand{\poly}{\mathsf{poly}}
\newcommand{\polylog}{\mathsf{polylog}}

\newcommand{\indicator}{\mathbbm{1}}

\newcommand{\LDTest}{\mathsf{LDTest}}
\newcommand{\AugLDTest}{\mathsf{AugLDTest}}
\newcommand{\vcsp}{VecCSP}
\newcommand{\svcsp}{SVecCSP}
\newcommand{\linearc}{\mathsf{l}}
\newcommand{\parallelc}{\mathsf{p}}

\newcommand{\pairx}{\mathsf{x}}
\newcommand{\pairy}{\mathsf{y}}
\newcommand{\pcppdelta}{\delta_\star}
\newcommand{\Enc}{\mathrm{Enc}}
\newcommand{\Dec}{\mathrm{Dec}}
\newcommand{\ckt}{\mathsf{ckt}}
\newcommand{\ldt}{\mathsf{ldt}}
\newcommand{\pcppckt}{\Pcal_{\sf ckt}}
\newcommand{\pcpptabletest}{\Pcal_{\sf mt}}
\newcommand{\pcppldt}{\Pcal_{\sf ldt}}
\newcommand{\eccim}{\mathrm{Im}}
\newcommand{\charF}{\mathsf{char}}
\newcommand{\pcppdoubletest}{\Pcal_{\sf dt}}
\newcommand{\pcppsingletest}{\Pcal_{\sf st}}

\newcommand{\Fbb}{\mathbb{F}}
\newcommand{\Hbb}{\mathbb{H}}
\newcommand{\Lbb}{\mathbb{L}}
\newcommand{\Nbb}{\mathbb{N}}

\newcommand{\Acal}{\mathcal{A}}

\newcommand{\Ecal}{\mathcal{E}}

\newcommand{\Pcal}{\mathcal{P}}

\newcommand{\Scal}{\mathcal{S}}

\newcommand{\val}{{\rm val}}
\newcommand{\nrm}{{\sf RM}}

\newcommand{\prm}{{\nrm}^{\mathbb F,m,d,t}}

\newcommand{\pcircuit}{{\sc CktVal}}

\newcommand{\probttt}{{\sc MultiTest}}
\newcommand{\singletest}{{\sc SingleTest}}
\newcommand{\doubletest}{{\sc DoubleTest}}
\newcommand{\svsat}{{\sc SVSat}}
\newcommand{\cktsize}{|\mathbb F|^m\poly|\mathbb F|}

\newcommand{\gapkcsp}{{$\varepsilon$-Gap $k$-Variable CSP}}

\renewcommand{\paragraph}[1]{\medskip\noindent\textbf{#1.}}

\renewcommand{\tilde}{\widetilde}
\renewcommand{\bar}{\overline}
\renewcommand{\hat}{\widehat}
\newcommand{\lowerbound}[1]{#1/2^{\omega({\sqrt{\log #1}\log \log #1})}}

\newcommand{\lowerboundK}{\lowerbound{K}}

\title{Almost Optimal Time Lower Bound for Approximating Parameterized Clique, CSP, and More, under ETH}
\author{
Venkatesan Guruswami\thanks{Simons Institute for the Theory of Computing, and Departments of EECS and Mathematics, UC Berkeley. Email: {\tt venkatg@berkeley.edu}. Research supported in part by NSF grants CCF-2228287 and CCF-2211972 and a Simons Investigator award. }
\and 
Bingkai Lin\thanks{State Key Laboratory for Novel Software Technology, Nanjing University. Email: \texttt{lin@nju.edu.cn}}
\and 
Xuandi Ren\thanks{Department of EECS, UC Berkeley. Email: \texttt{xuandi\_ren@berkeley.edu}. Supported in part by NSF grant CCF-2228287.} 
\and
Yican Sun\thanks{School of Computer Science, Peking University. Email: \texttt{sycpku@pku.edu.cn}}
\and
Kewen Wu\thanks{Department of EECS, UC Berkeley. Email: \texttt{shlw\_kevin@hotmail.com}. Supported by a Sloan Research Fellowship and NSF CAREER Award CCF-2145474.}
}
\date{}

\begin{document}

\maketitle

\begin{abstract}
The Parameterized Inapproximability Hypothesis (PIH), which is an analog of the PCP theorem in parameterized complexity, asserts that, there is a constant $\varepsilon> 0$ such that for any computable function $f:\mathbb{N}\to\mathbb{N}$, no $f(k)\cdot n^{O(1)}$-time algorithm can, on input a $k$-variable  CSP instance with domain size $n$, find an assignment satisfying $1-\varepsilon$ fraction of the constraints. A recent work by Guruswami, Lin, Ren, Sun, and Wu (STOC'24) established PIH under the Exponential Time Hypothesis (ETH). 

In this work, we improve the quantitative aspects of PIH and prove (under ETH) that approximating sparse parameterized CSPs within a constant factor requires $n^{k^{1-o(1)}}$ time. This immediately implies that, assuming ETH, finding a $(k/2)$-clique in an $n$-vertex graph with a $k$-clique requires $n^{k^{1-o(1)}}$ time. We also prove almost optimal time lower bounds for approximating $k$-ExactCover and Max $k$-Coverage.

Our proof follows the blueprint of the previous work to identify a "vector-structured" ETH-hard CSP whose satisfiability can be checked via an appropriate form of "parallel" PCP. Using further ideas in the reduction, we guarantee additional structures for constraints in the CSP. We then leverage this to design a parallel PCP of almost linear size based on Reed-Muller codes and derandomized low degree testing. 
\end{abstract}

\clearpage

\tableofcontents

\clearpage

\section{Introduction}\label{sec:intro}

 One of the goals of complexity theory is to pinpoint the asymptotically optimal time (or other resource) needed to solve basic computational problems or classes of problems. The theory of NP-completeness attacks this at a coarse level, but modern complexity theory also has tools to give more fine-grained information on computational complexity. 

 A common setting for fine-grained understanding of computational hardness is considering \emph{parameterized problems}. Under this setting, each instance is attached with an additional parameter $k$ indicating some specific quantities (e.g., the optimum or the treewidth). We treat $k$ as some super constant that is much smaller than the instance size $n$ and consider the existence or absence of algorithms with running time depends both on $n$ and $k$ (e.g., with running time $2^{2^k}n^{O(1)}$, or $n^{\sqrt{k}}$). The hardness of parameterized problems is studied under the realm of  \emph{parameterized complexity theory}~\cite{FG06}. It is a central challenge to figure out the minimal time (depending on both $n$ and $k$) to solve prototypical parameterized problems.

 A representative example is the {\sc $k$-Clique} problem parameterized by the optimum $k$, which is one of the most fundamental problems in parameterized complexity theory: given an $n$-vertex graph as input, determine if it has a clique of size $k$. The naive brute force algorithm takes roughly $n^k$ time.
 Using fast matrix multiplication, there are better algorithms that take $n^{\omega k/3}$ time, where $\omega$ is the matrix multiplication exponent. 
 On the hardness side, it is known that no algorithm can decide {\sc $k$-Clique} within running time $f(k)n^{o(k)}$, for any computable function $f(k)$, under the widely-considered Exponential Time Hypothesis (ETH) \cite{CHKX06}, which states that algorithms cannot solve {\sc 3SAT} formulas on $n$ variables within running time $2^{o(n)}$. The optimal running time for $k$-clique is therefore pinpointed to be $n^{\Omega(k)}$, assuming ETH.

 Can one design a faster algorithm if one settles for \emph{approximating} {\sc $k$-Clique}?
 For example, what if we only want to find a clique of size $k/2$ in a graph that is promised to have a $k$-clique.\footnote{This can also be formulated as a ``gap" decision problem of distinguishing graphs with a $k$-clique from those which do not even have a $(k/2)$-clique.} 

 It was shown that such an approximation to {\sc $k$-Clique} still requires the tightest $f(k)n^{\Omega(k)}$ time~\cite{CCK+17}, under the very strong assumption Gap-ETH. 
 The Gap-ETH postulates an exponential time lower bound of approximating {\sc Max 3-SAT} within a constant ratio.
 The constant gap baked into the assumption is then transformed into a constant gap for approximating {\sc $k$-Clique}.
 Though Gap-ETH has been proved under particular strengthening of ETH~(smooth version of ETH)~\cite{App17},
 a theoretically more satisfactory result is to obtain the hardness of approximating {\sc $k$-Clique} under the original ETH.  Under ETH, weaker time lower bounds were known for constant-factor approximations of {\sc $k$-Clique}: $f(k)n^{\Omega\pbra{\sqrt[6]{\log k}}}$ in \cite{Lin21}, which was later improved to $f(k)n^{\Omega(\log k)}$ in \cite{LRSW22,CFLL23} and $f(k)n^{k^{\Omega(1/\log \log k)}}$ in \cite{LRSW23}.\footnote{There is another line of work on improving the inapproximability factor of {\sc $k$-Clique} under the minimal hypothesis {\sf W[1]$\neq$FPT}: constant factor in \cite{Lin21}, and the improved $k^{o(1)}$ in \cite{KK22,CFLL23}.}
 However, all these approaches cannot obtain lower bounds better than $f(k)n^{\Omega(\sqrt{k})}$ due to the coding theoretic barriers~\cite{KT00}.

 However, this paper significantly improves this lower bound, assuming only the original ETH.
 \begin{theorem}
 \label{thm:kclique-harndess}
     Assume ETH. For any constant $\varepsilon >0$ and any computable function $f(k)$, any algorithm that approximates {\sc $k$-Clique} within an $\varepsilon$ ratio must take runtime $f(k)n^{k^{1-o(1)}}$.
 \end{theorem}

To prove the theorem above, we follow a similar idea of proving the NP-hardness of approximating cliques, which relies on the NP-hardness of constant approximating CSP (a.k.a., the PCP theorem) and a subsequent FGLSS reduction~\cite{FGLSS96}. 
In the parameterized setting, we can apply an analogous reduction.
\begin{itemize}
    \item We first establish a near-optimal lower bound for approximating (sparse) ``parameterized CSPs'' with $k$ variables, $O(k)$ constraints, an alphabet of size $n$, and some constant inapproximability factor.
    \item Then \Cref{thm:kclique-harndess} follows immediately by the FGLSS reduction~\cite{FGLSS96} and an expander-based gap-amplification procedure \cite{AFWZ95,LRSW22}.
\end{itemize}

The first step is equivalent to establishing a quantitative version of the Parameterized Inapproximability Hypothesis (PIH)~\cite{LRSZ20}, which plays the role of the PCP theorem in the parameterized complexity theory. 
The first quantitative version of PIH was established under Gap-ETH~\cite{CCK+17}, with the lower bound $f(k)n^{\Omega(\log k)}$. A recent improvement proved PIH under ETH~\cite{GLRSW24}, with a weaker lower bound $f(k)n^{\Omega\pbra{\sqrt{\log \log k}}}$. However, both results are too weak to establish~\Cref{thm:kclique-harndess}. In this paper, we make a significant quantitative improvement to the reduction in~\cite{GLRSW24}, obtaining our main theorem stated below.

\begin{theorem}[Informal version of ~\Cref{thm:main}]
\label{thm:main-intro}
Assume ETH. For some constant $\eps \in (0,1)$, for any computable function $f(k)$, every algorithm that takes as input a satisfiable parameterized 2CSP instance with $k$ variables, $O(k)$ constraints and size-$n$ alphabet finds an assignment satisfying $1-\eps$ fraction of the constraints, must take $f(k)n^{k^{1-o(1)}}$ time.
\end{theorem}

Combining with the parallel repetition in projection games~\cite{rao08}, we can immediately boost the soundness to any constant $\eps>0$ with lower bound $f(k)n^{k^{\Omega\left(1/\log(1/\varepsilon)\right)}}$.

\begin{corollary}
\label{cor:main-intro}
Assume ETH. For any computable function $f(k)$ and any constant $\eps\in (0,1)$, every algorithm that takes as input a satisfiable parameterized 2CSP instance with $k$ variables, $O(k)$ constraints and size-$n$ alphabet finds an assignment satisfying $1-\eps$ fraction of the constraints, must take $f(k)n^{k^{\Omega\left(1/\log(1/\varepsilon)\right)}}$ time.
\end{corollary}

By the discussion above, \Cref{thm:kclique-harndess} follows immediately by combining \Cref{thm:main-intro}, the FGLSS reduction, and the expander-based gap-amplification procedure~\cite{AFWZ95,LRSW22}. Moreover, using \Cref{thm:main-intro} as the foundation, we can obtain, via reductions, strong inapproximability results for other fundamental parameterized problems. Below, we list two application highlights and refer interested readers to \cite{GLRSW24} for more detailed discussions.

\paragraph{Application Highlight: {\sc $k$-ExactCover}} 
{\sc $k$-ExactCover} (a.k.a., {\sc $k$-UniqueSetCover}) is one of the canonical problems in the parameterized world. It is a weaker version of the {\sc $k$-SetCover} problem.
For the $\rho$-approximation version of {\sc $k$-ExactCover} with $\rho \ge 1$, denoted by $(k, \rho\cdot k)$-{\sc ExactCover}, the instance consists of a universe $U$ and a collection $\mathcal S$ of subsets of $U$, with a goal to distinguish the following two cases.
\begin{itemize}
    \item There exists $k$ \emph{disjoint} subsets from $\mathcal S$ whose union equals the whole universe $U$.
    \item The union of any $\rho\cdot k$ subsets of $\Scal$ is a proper subset of $U$. 
\end{itemize}
Here, the parameter is the optimum $k$. We remark that the additional disjointness requirement in the completeness part makes {\sc $(k, \rho\cdot k)$-ExactCover} an excellent intermediate problem for proving the hardness of other problems \cite{ABSS97, manurangsi2020tight}. 

On the algorithmic side, the $(k, \rho\cdot k)$-{\sc ExactCover} has a brute-force $|\Scal|^{\Omega(k)}$-time algorithm. However, no $|\Scal|^{o(k)}$-time algorithms are known. 
Thus, it is natural to consider whether we can establish a matching $|\Scal|^{\Omega(k)}$ lower bound.
Our work almost establishes a lower bound for $(k, \rho\cdot k)$-{\sc ExactCover}, under ETH, for some constant $\rho$. Previously, this was only known under the Gap-ETH assumption~\cite{manurangsi2020tight}.

\begin{theorem}\label{thm:hardness-exactcover}
    Assume ETH. There exists some constant $\rho\ge 1$, such that for any computable function $f(k)$, any algorithm deciding $(k, \rho\cdot k)$-{\sc ExactCover} must take runtime $f(k)|\Scal|^{k^{1-o(1)}}$.
\end{theorem}

To prove the theorem above, We note that the previous work~\cite{guruswami2023baby} achieves a parameter-preserving reduction from PIH to $(k, \rho\cdot k)$-{\sc ExactCover} for any constant $\rho$, by imitating the beautiful reduction of Feige~\cite{feige1998threshold}.\footnote{In fact, a weaker variant of PIH, named Average Baby PIH over rectangular constraints, suffices.}

Therefore, \Cref{thm:hardness-exactcover} follows by combining the reduction of \cite{guruswami2023baby} with our \Cref{thm:main-intro}. Applying \Cref{cor:main-intro}, we can boost the ratio $\rho$ to any constant with lower bound $f(k)|\Scal|^{k^{\Omega(1/{\log(\rho)})}}$.% lower bound.

\begin{proposition}\label{thm:hardness-exactcover-boost}
    Assume ETH. For any constant $\rho\ge 1$ and computable function $f(k)$, any algorithm deciding $(k, \rho\cdot k)$-{\sc ExactCover} must take runtime $f(k)|\Scal|^{k^{\Omega(1/{\log(\rho)})}}$.
\end{proposition}

\paragraph{Application Highlight: {\sc Max $k$-Coverage}} 
The {\sc Max $k$-Coverage} is the maximization variant of the $k$-{\sc SetCover} problem. For the $\rho$-approximation version of {\sc Max $k$-Coverage} with $\rho < 1$, denoted by {\sc Max $(\rho, k)$-Coverage}, the instances are the same as {\sc $k$-ExactCover} above, but the goal changes to distinguish the following two case:
\begin{itemize}
    \item There exists $k$ subsets from $\mathcal S$ whose union equals $\mathcal U$.
    \item Any $k$ subsets from $\mathcal S$ has the union size at most $\rho\cdot |\mathcal U|$.
\end{itemize}
{\sc Max $(\rho, k)$-Coverage} has been widely studied in previous literature. There exists a simple greedy algorithm solving {\sc Max $(1-\frac{1}{e}, k)$-Coverage} within polynomial runtime~\cite{Hoc97}. 

On the hardness side, a celebrated result of Feige~\cite{feige1998threshold} showed the NP-hardness of {\sc Max $(1-\frac{1}{e}+\varepsilon, k)$-Coverage} for any $\varepsilon>0$, thus proving a tight inapproximability result.

In the parameterized world, one can solve {\sc Max $k$-Coverage} in $|\Scal|^{k}$ time by brute force enumeration. On the other hand, Cohen-Addad, Gupta, Kumar, Lee, and Li~\cite{CGK+19} showed that assuming Gap-ETH, {\sc Max $(1-\frac{1}{e}+\varepsilon, k)$-Coverage} requires $f(k)|\Scal|^{k^{{\rm poly}(1/\varepsilon)}}$ runtime. 
Manurangsi~\cite{manurangsi2020tight} further improved this lower bound to the tightest $f(k)|\Scal|^{\Omega(k)}$ under Gap-ETH. 

Our work implies an almost-optimal time lower bound for {\sc Max $(\rho, k)$-Coverage} under ETH for some constant $\rho$.

\begin{theorem}
\label{hardness:maxkcoverage}
    Assume ETH. There exists some constant $\rho\in (0,1)$, such that for any computable function $f(k)$, any algorithm deciding {\sc Max $(\rho, k)$-Coverage} must take runtime $f(k)|\Scal|^{k^{1-o(1)}}$.
\end{theorem}

\Cref{hardness:maxkcoverage} follows from our \Cref{thm:main-intro} and the analysis in \cite[Sections 9.1 and 9.2]{manurangsi2020tight} that accomplishes a gap-preserving reduction from {\sc $k$-CSP} to {\sc Max $k$-Coverage}. Applying \Cref{cor:main-intro}, we can boost the ratio to the tightest $1-\frac1e + \eps$ with lower bound $f(k)|\Scal|^{k^{\eps'}}$.

\begin{proposition}
    Assume ETH. for any constant $\eps\in (0,1)$ and computable function $f(k)$, there exists a constant $\eps' = \eps'(\eps)$, any algorithm deciding {\sc Max $(1-\frac1e+\eps, k)$-Coverage} must take runtime $f(k)|\Scal|^{k^{\eps'}}$.
\end{proposition}

\paragraph{New PCP Characterizations} 
An interesting byproduct of \Cref{thm:main-intro} is a new PCP theorem for 3SAT as follows.

\begin{theorem}[Informal Version of \Cref{thm:new-pcp-restate}]\label{thm:new-pcp}
For any parameter $k\ll n$, 
 {\sc 3SAT} has a constant-query PCP verifier with alphabet size $|\Sigma| = 2^{n/k^{1-o(1)}}$, runtime ${\rm poly}(|\Sigma|, n)$, and $
\log k+O\pbra{\sqrt{\log k}\log \log k}
$ random coins, which has perfect completeness and soundness $\frac12$.
\end{theorem}

\Cref{thm:new-pcp} generalizes the classic PCP theorem and gives a smooth trade-off between the proof length and the alphabet size, connecting parameterized complexity and the classical complexity theory. 

\paragraph{Paper Organization}
In \Cref{sec:tech-overview}, we provide an overview for our proof of \Cref{thm:main-intro} and discuss future works.
In \Cref{sec:prelim}, we formalize necessary notation and concepts.
The formal structure for proving \Cref{thm:main-intro} is presented in \Cref{sec:pf_main}, with most technical statements deferred to other sections: the reduction from ETH-hard problems to special vector-valued CSPs is provided in \Cref{sec:prepro}, and the PCPP verifier for a helper language is constructed in \Cref{sec:pcpp-probttt}.
In \Cref{sec:lowdeg}, we give details of derandomized parallel low degree tests which is a key component for the PCPP construction.
\section{Technical Overview}\label{sec:tech-overview}

In this part, we provide general ideas of proving \Cref{thm:main-intro}. Due to the equivalence between the existence of PCP systems and the inapproximability of constraint satisfaction problems~\cite{Din07,AB09}, \Cref{thm:new-pcp} directly follows from \Cref{thm:main-intro}.

We follow the spirit of \cite{GLRSW24} to prove \Cref{thm:main-intro}. The proof framework is as follows.
\begin{itemize}
    \item First, we reduce an ETH-hard problem to a CSP problem with specific structures. 
    \item Then, leveraging the special structures, we construct a probabilistic checkable proof of proximity (PCPP) verifier to check whether the encoding of some given solution satisfies all constraints. 
    \Cref{thm:main-intro} follows by converting the PCP verifier into CSP instances~\cite{AB09,Din07}.
\end{itemize}
For the first step, \cite{GLRSW24} reduces {\sc 3SAT} to vector-valued CSPs (\vcsp{}s for short), whose variables take values in a vector space $\Fbb_4^t$. In addition, each constraint is either a coordinate-wise parallel constraint or a linear constraint.
\begin{itemize}
    \item A parallel constraint (over variables $x$ and $y$) is defined by a sub-constraint $\Pi^{sub}: \Fbb_4\times \Fbb_4\to \bin$ and a subset of coordinate $Q\subseteq[t]$. It checks whether $\Pi^{sub}(x_i,y_i) = 1$ for every coordinate $i\in Q$.
    \item A linear constraint enforces that two vector-valued variables satisfy a linear equation specified by a matrix $M\in\Fbb_4^{t\times t}$, i.e., $y=Mx$.
\end{itemize}
Then, in the second step, \cite{GLRSW24} encodes the solution by parallel Walsh-Hadamard code and constructs a PCPP verifier with double-exponential proof length, resulting in an $f(k)n^{\Omega(\sqrt{\log \log k})}$ lower bound for \gapkcsp{}. 

\subsection{More Refined Vector Structure}

Unfortunately, the vector structure used in~\cite{GLRSW24} is far from being enough for obtaining an almost-optimal lower bound due to the following reasons.
\begin{itemize}
    \item For parallel constraints, \vcsp{} sets up the sub-constraints over \emph{a subset of the coordinates}. 
    There might be $2^{|E_{\sf p}|}$ (where $|E_{\sf p}|$ is the number of parallel constraints) different sub-CSP instances over all coordinates, each of which requires an individual PCPP verifier. To simultaneously check the satisfiability of all these sub-instances, they tuple these verifiers into a giant verifier, resulting in an exponential blowup of the proof length. Hence, an almost linear proof size is impossible.
    \item For linear constraints, \vcsp{} defined in ~\cite{GLRSW24} allow them to be over arbitrary pairs of variables. They introduce auxiliary variables for all pairs of variables and their corresponding linear constraints to check such unstructured constraints. 
    The number of auxiliary variables is $|V|\cdot|E_{\sf l}|$, where $|V|, |E_{\sf l}|$ are the numbers of variables and linear constraints, respectively. This means that the proof length is at least quadratic, making an almost linear proof size impossible.
    \item Furthermore, the \vcsp{} instance in \cite{GLRSW24} has parameters $|V| = O(k^2), |E| = O(k^2)$, which is a starting point with already a quantitative loss for any subsequent constructions. 
\end{itemize}

The analysis above urges us to mine more vector structures and devise new reductions with smaller parameter blowups. 
In this work, we further engineer \vcsp{} and obtain special vector-valued CSPs~(\svcsp{} for short) with three more features.
\begin{itemize}
    \item First, \svcsp{} partitions the variables into two disjoint parts $\{x_1,\ldots x_k\}\dot\cup \{y_1,\ldots y_k\}$.
    \item Second, for parallel constraints, \svcsp{} sets up the sub-constraint on \emph{all} coordinates, which implies a unified sub-CSP instance among all coordinates. As a result, we avoid the tuple procedure, enabling a highly succinct proof.
    \item Third, for linear constraints, \svcsp{} only sets up linear constraints over $x_i$ and $y_i$, with the same index $i$. After encoding $x$ and $y$ by the parallel Reed-Muller code, we can leverage this alignment and introduce auxiliary proofs, which is also a codeword of the parallel Reed-Mulle code, to check the validity of linear constraints efficiently, with an almost-linear blowup.
\end{itemize}

In addition, to decrease the parameter blowup, we apply the reduction in~\cite{Mar10, KMPS23} to obtain ETH-hard \emph{sparse} parameterized CSP instances. Then, we imitate the reduction in~\cite{GLRSW24} to obtain a sparse \vcsp{} instance with $|V| = O(k), |E| = O(k)$, which ultimately reduces to \svcsp{}, with an almost optimal lower bound, by properly duplicating variables and relocating constraints.

\subsection{Applying The Parallel Reed-Muller Code}
After establishing the almost optimal runtime lower bound for \svcsp{}, we need to design a PCPP verifier certifying the encoding of a given assignment satisfies all constraints. However, the previous work~\cite{GLRSW24} designs a PCPP verifier that requires a proof with a double-exponential length, which is far from obtaining \Cref{thm:main-intro}.

This paper applies the parallel Reed-Muller (RM) encoding with an almost-linear codeword length. Recall that the variables in the \svcsp{} instance are divided into two disjoint parts $\{x_1,\ldots x_k\}\dot\cup \{y_1,\ldots y_k\}$, the input proof of the PCPP verifier consists of $\hat{x}$ and $\hat{y}$, which are supposed to be the parallel RM encoding of the assignment over $x$ and $y$ respectively. The verification procedure is demonstrated as follows.

First, to ensure that $\hat{x}$ and $\hat{y}$ are indeed codewords of the parallel RM code, we apply parallel low degree testing. The standard low degree testing causes a quadratic blowup in the proof length. To ensure an almost-linear proof length, we use the {parallel version} of the \emph{derandomized low degree testing}~\cite{ben2003randomness} (\Cref{sec:lowdeg}). 

Second, we check whether $\hat{x}\circ \hat{y}$ satisfies all parallel constraints. Since a unified sub-CSP instance exists among all coordinates, as long as we have constructed a PCPP verifier for this sub-CSP, we can simulate it in parallel to check all coordinates at the same time, with no blowup in the proof length (as opposed to an exponential blowup in~\cite{GLRSW24}). 
Our sub-CSP has an alphabet of constant size. Thus, we can apply the existing approach~(see, e.g., \cite[Theorem 3.3]{BGH06}) in a black box way to construct an almost-linear PCPP verifier for all parallel constraints in our sub-CSP.

Finally, we check whether $\hat{x}\circ \hat{y}$ satisfies all linear constraints. Recall that \svcsp{} only sets up linear constraints over $x_i$ and $y_i$, with the same index $i$. 
Denote $M_i$ as the matrix for the index $i$, i.e., the $i$-th linear constraint is $y_i = M_ix_i$.
We introduce an auxiliary proof $\hat{z}$, satisfying 
$$
\hat{z} = \hat{y} - \hat{M}\hat{x}
$$
where $\hat{M}$ is the parallel RM encoding for matrices $\{M_1, M_2, \dots, M_k\}$.
Then, all systematic parts of $\hat{z}$, i.e., the codeword entries corresponding to $z_i:=y_i-M_ix_i$ for $i\in[k]$, should be $\vec{0}$.
The key observation is that if $\hat{x}$ and $\hat{y}$ are codewords of parallel RM code, then so does $\hat{z}$. 
Hence, we can apply parallel derandomized low degree testing for $\hat{z}$ and apply another PCPP, as in the parallel part, to check whether all systematic parts of $\hat{z}$ are $\vec{0}$. 
Finally, we check whether $\hat{z}$ satisfies the equation above by simply querying a random index, for which the soundness and completeness are guaranteed by Schwartz-Zippel lemma~\cite{Sch80}. 
In this way, we have a highly efficient PCPP verifier that checks all linear constraints with an almost linear proof length.

The overall PCPP verifier is the combination of the verifiers for parallel constraints and for linear constraints.
A more detailed framework of our proof is in \Cref{sec:pf_main}.

\subsection{Future Works}

Our work gives almost optimal time lower bounds for the approximation version of many canonical parameterized problems under ETH, including {\sc $k$-Clique}, {\sc $k$-ExactCover}, {\sc $k$-Variable CSPs}, and {\sc Max $k$-Coverage}. 

Technically, we prove the almost optimal time lower bound for constant-gap $k$-Variable CSPs. Using this result as the cornerstone, we obtain the inapproximability for other problems by existing reductions. 
One open question is whether almost optimal time lower bounds for other problems also follow from this result, e.g., $k$-{\sc Balanced Biclique}~\cite{CCK+17,feldman2020survey}.

The second question is to obtain the (actual) optimal $f(k)n^{\Omega(k)}$ time lower bounds for constant-gap $k$-Variable CSPs. 
This problem can be seen as the parameterized extension of the long-standing linear-size PCP conjecture~\cite{BEK+12}. 
In the non-parameterized world, the state-of-the-art PCP theorem with the shortest proof length is due to~\cite{Din07} with quasilinear proof length.
It is also interesting if this optimal bound (i.e., the parameterized extension of the linear-size PCP conjecture) can be established assuming the existence of linear-size PCP.

For more interesting directions, please refer to~\cite{GLRSW24}.

\section{Preliminaries}\label{sec:prelim}

For a positive integer $n$, we use $[n]$ to denote the set $\cbra{1,2,\ldots,n}$.
We use $\log$ to denote the logarithm with base $2$.
For a prime power $q=p^c$ where $p$ is a prime and $c\ge1$ is an integer, we use $\Fbb_q$ to denote the finite field of order $p^c$ and characteristic $\charF(\Fbb)=p$.

For an event $\Ecal$, $\indicator_\Ecal$ is defined the indicator function, which equals $1$ if $\Ecal$ happens and $0$ otherwise.
For a finite set $S\ne\emptyset$, we use $x\sim S$ to denote a uniformly random element from $S$.

For disjoint sets $S$ and $T$, we use $S\dot\cup T$ to denote their union while emphasizing $S\cap T=\emptyset$.

\paragraph{Asymptotics}
Throughout the paper, we use $O(\cdot),\Theta(\cdot),\Omega(\cdot)$ to hide absolute constants that do not depend on any other parameter.
We also use $\poly(\cdot)$ to denote some implicit polynomial in terms of the parameters within, e.g., $\poly(f,g)$ is upper bounded by $(f^2+g^2+C)^C$ for some absolute constant $C\ge0$.

\subsection{Constraint Satisfaction Problem}

In this paper, we only focus on constraint satisfaction problems (CSPs) of arity two. Formally, a CSP instance $G$ is a quadruple $(V, E, \Sigma, \{\Pi_{e}\}_{e\in E})$, where:
\begin{itemize}
    \item $V$ is for the set of variables.
    \item $E$ is for the set of constraints.
    Each constraint $e=\cbra{u_e,v_e} \in E$ connects two distinct variables $u_e,v_e\in V$.

    The \emph{constraint graph} is the undirected graph on vertices $V$ with edges $E$. Note that we allow multiple constraints between the same pair of variables; thus, the constraint graph may have parallel edges.
    \item $\Sigma$ is for the alphabet of each variable in $V$. For convenience, we sometimes have different alphabets for different variables, and we will view them as a subset of a grand alphabet with some natural embedding.
    \item $\{\Pi_{e}\}_{e\in E}$ is the set of constraint validity functions. 
    Given a constraint $e\in E$, the function $\Pi_e\colon\Sigma\times\Sigma\to\bin$ checks whether the constraint $e$ between $u_e$ and $v_e$ is satisfied.
\end{itemize}
We use $|G|=(|V|+|E|)\cdot|\Sigma|$ to denote the \emph{size} of a CSP instance $G$.

\paragraph{Assignment and Satisfiability Value}
An \emph{assignment} is a function $\sigma\colon V\to \Sigma$ that assigns each variable a value in the alphabet. 
We use $$\val(G, \sigma)=\frac1{|E|}\sum_{e\in E} \Pi_e(\sigma(u_e),\sigma(v_e))$$ to denote the \emph{satisfiability value} for an assignment $\sigma$. 
The satisfiability value for $G$ is $\val(G) = \max_{\sigma\colon V\to \Sigma} \val(G, \sigma)$. 
We say that an assignment $\sigma$ is a \emph{solution} if $\val(G,\sigma) = 1$, and $G$ is \emph{satisfiable} iff $G$ has a solution. When the context is clear, we omit $\sigma$ in the description of a constraint, i.e., $\Pi_e(u_e,v_e)$ stands for $\Pi(\sigma(u_e),\sigma(v_e))$. 

\paragraph{Boolean Circuits}
Throughout the paper, we consider standard Boolean circuits with AND/OR gate of fan-in two and fan-out two, and NOT gate with fan-out two.

\paragraph{Exponential Time Hypothesis (ETH)}
\emph{Exponential Time Hypothesis (ETH)}, first proposed by Impagliazzo and Paturi \cite{IP01}, is a famous strengthening of the $\mathsf{P}\neq\mathsf{NP}$ hypothesis and that has since found numerous applications in  modern complexity theory, especially in fine-grained complexity.

The ETH postulates that the general {\sc 3SAT} problem has no sub-exponential algorithm.
In this paper, we use the ETH-based hardness of the {\sc 4-Regular 3-Coloring} problem.

\begin{definition}[{\sc 4-Regular 3-Coloring}]\label{def:4regular3coloring}
A 2CSP instance $G=(V,E,\Sigma,\{\Pi_e\}_{e\in E})$ is an instance of {\sc 4-Regular 3-Coloring} if (1) $\Sigma=\Fbb_4$, (2) the constraint graph is $4$-regular, and (3) each $\Pi_e$ checks whether the two endpoints of $e$ are assigned with different colors from $\Lambda$, where $\Lambda\subset\Fbb_4$ has size three and is fixed in advance.
\end{definition}

We remark that usually {\sc 3-Coloring} is defined directly with a ternary alphabet $\Lambda$. Here for simplicity of later reductions, we assume the alphabet is $\Fbb_4$. This is without loss of generality since the coloring constraint is also upgraded to check additionally whether the colors are from $\Lambda$.

\begin{theorem}[ETH Lower Bound for {\sc 4-Regular 3-Coloring} \cite{CFG+16}]\label{thm:eth_3color}
    Assuming ETH, no algorithm can decide {\sc 4-Regular 3-Coloring} in $2^{o(|V|)}$ time.
\end{theorem}

\subsection{Parameterized Complexity Theory}

In parameterized complexity theory, we consider a promise language $L_{\rm yes}\dot\cup L_{\rm no}$ equipped with a computable function $\kappa$, which returns a parameter $\kappa(x)\in\Nbb$ for every input instance $x$. We use $(L_{\rm yes}\dot\cup L_{\rm no}, \kappa)$ to denote a parameterized language. We think of $\kappa(x)$ as a growing parameter that is much smaller than the instance size $|x|$. 

A parameterized promise language $(L_{\rm yes}\dot\cup L_{\rm no},\kappa)$ is \emph{fixed parameter tractable} (FPT) if there is an algorithm such that for every input $(x, \kappa(x))\in (L_{\rm yes}\dot\cup L_{\rm no}, \kappa)$, it decides whether $x\in L_{\rm yes}$ in $f(\kappa(x))\cdot |x|^{O(1)}$ time for some computable function $f$.

An \emph{FPT reduction} from $(L_{\sf yes}\cup L_{\sf no}, \kappa)$ to $(L'_{\sf yes}\cup L'_{\sf no}, \kappa')$ is an algorithm $\mathcal A$ which, on every input $(x, \kappa(x))$ outputs another instance $(x', \kappa'(x'))$ such that:
\begin{itemize}
    \item \textsc{Completeness.} If $x\in L_{\sf yes}$, then $x'\in L'_{\sf yes}$.
    \item \textsc{Soundness.} If $x\in L_{\sf no}$, then $x'\in L'_{\sf no}$.
    \item \textsc{FPT.} There exist universal computable functions $f$ and $g$ such that  $|\kappa'(x')|\le g(\kappa(x))$ and the runtime of $\mathcal A$ is bounded by $f(\kappa(x))\cdot |x|^{O(1)}$.
\end{itemize}

We refer to \cite{FG06} for the backgrounds on fixed parameter tractability and FPT reductions. 

\paragraph{$\eps$-Gap $k$-Variable CSP} 
We mainly focus on the gap version of the parameterized CSP problem. Formally, an $\eps$-Gap $k$-Variable CSP problem is the following parameterized promise language $(L_{\rm yes}\dot\cup L_{\rm no}, \kappa)$.
\begin{itemize}
    \item $L_{\rm yes}$ consists of all CSPs $G$ with $\val(G) = 1$.
    \item $L_{\rm no}$ consists of all CSPs $G$ with $\val(G) < 1-\eps$.
    \item $\kappa(G)$ equals the number of variables in $G$.
\end{itemize}
In other words, we need to decide whether a given CSP instance $(G,|V|)$ with $k$ variables satisfies $\val(G)=1$ or $\val(G)<1-\varepsilon$. 

\paragraph{Parameterized Inapproximability Hypothesis (PIH)} 
\emph{Parameterized Inapproximability Hypothesis (PIH)} is a folklore conjecture generalizing the celebrated PCP theorem to the parameterized complexity.
It was first rigorously formulated in \cite{LRSZ20}. Below, we present a slight reformulation, asserting fixed parameter intractability (rather than $W[1]$-hardness specifically) of gap CSP.

\begin{hypothesis}[PIH]\label{hypo:pih}
For an absolute constant $0<\varepsilon<1$, no FPT algorithm can decide \gapkcsp{}.
\end{hypothesis}

\subsection{Parallel Reed-Muller Code}

\paragraph{Word}
We say $x$ is a word (a.k.a., vector) with alphabet $\Sigma$ if $x$ is a string of finite length and each entry is an element from $\Sigma$; and $\Sigma^*$ contains all words with alphabet $\Sigma$.
Assume $x$ has length $m$. For each $I\subseteq[m]$, we use $x_I$ to denote the sub-string of $x$ on entries in $I$.
When $I=\cbra{i}$ is a singleton set, we simply use $x_i$ to denote $x_{\cbra{i}}$. 
For a word $x$ over a vector alphabet $\Sigma^t$, for each entry $x_i$ and $j\in [t]$, we define $x_i[j]$ as the $j$-th coordinate of $x_i$. We define $x[j]$ as $x_1[j] \circ x_2[j]\circ \cdots \circ x_{m}[j]$, which is a word over $\Sigma$.

Let $y$ be another word. We use $x\circ y$ to denote the concatenation of $x$ and $y$.
If $y$ is also of length $m$, we define $\Delta(x,y) := \Pr_{i\in [m]}[x_i\ne y_i]$ as their relative Hamming distance.
For a set $S$ of words, if $\Delta(x,z)\ge\delta$ holds for every $z\in S$, we say $x$ is $\delta$-far from $S$; otherwise we say $x$ is $\delta$-close to $S$. In particular, if $S=\emptyset$, then $x$ is $1$-far from $S$.
Below, we recall the notion of error correcting codes. 

\begin{definition}[Error Correcting Code (ECC)]\label{def:ecc}
An error correcting code is the image of the encoding map $C\colon\Sigma_1^k\to \Sigma_2^K$ with message length $k$ and codeword length $K$. We say that the ECC has a relative distance $\delta$ if $\Delta(C(x),C(y))\ge \delta$ holds for any distinct $x,y\in\Sigma_1^k$. We use $\delta(C)$ to denote the relative distance of the image map of $C$ and use $\rm{Im}(C)$ to denote the codewords of $C$.
\end{definition}

\paragraph{Reed-Muller Code} 
We use the parallel Reed-Muller (RM) code to construct PCPPs. 
This parallel operation is called interleaving in coding theory.
Below, we present the formal definition.

For an $m$-variate parallel-output function $f\colon\Fbb^m\to\Fbb^{t}$, we denote $f[1],\ldots,f[t]\colon\Fbb^m\to\Fbb$ as its single-output components, i.e., $f(x)=(f[1](x),\ldots,f[t](x))$
We say $f$ is of \emph{parallel degree-$d$} if $f[1],\dots,f[t]$ are degree-$d$ polynomials, where a polynomial is degree-$d$ if all monomials with (total) degree larger than $d$ have zero coefficients.

If $|\Fbb|>d$ and $|\Fbb|^m\ge \binom{m+d}{d}$, by a dimension argument, there exist $\binom{m+d}{d}$ distinct points (a.k.a., interpolation set) $\{\xi_1, \dots, \xi_{\binom{m+d}{d}}\}\subseteq \Fbb^m$ whose values $a=(a_1,\ldots,a_{\binom{m+d}d})\in(\Fbb^{t})^{\binom{m+d}d}$ can uniquely determine the polynomial $f_a$ of parallel degree $d$.

\begin{definition}[Parallel RM Code]\label{def:prm}
    Assume $|\Fbb|>d$ and $|\Fbb|^m\ge \binom{m+d}{d}$. Let $\{\xi_1,\dots, \xi_{\binom{m+d}{d}}\}$ be the set above. The $(\Fbb, m, d, t)$-parallel RM code is the image of the following encoding map:
    $$
    \prm: \left(\Fbb^{t}\right)^{\binom{m+d}{d}}\to \left(\Fbb^{t}\right)^{|\Fbb|^m},
    $$
    where for each $a = (a_1,\dots, a_{\binom{m+d}{d}})\in (\Fbb^{t})^{\binom{m+d}{d}}$, the encoding $\prm(a)$ is the truth table of $f_a$ over the whole space $\Fbb^m$. 

    In addition, $a$ is the systematic part of the parallel RM encoding, which can be read off directly: for each $j\in \binom{m+d}{d}$, $a_j$ equals the entry indexed by $\xi_j$ in the codeword.
\end{definition}

By Schwartz-Zippel lemma\footnote{Schwartz-Zippel lemma is usually stated only for single-output polynomials, but it naturally generalizes to the parallel case.}, the relative distance of $(\Fbb, m, d, t)$-parallel RM code is 
$$
\delta(\prm)=1-\frac{d}{|\Fbb|}.
$$
Furthermore, there is an efficient codeword testing procedure (\Cref{thm:ldt-maintext}) for $\prm$, the proof of which is in \Cref{sec:lowdeg}.

\begin{theorem}[Codeword Testing]\label{thm:ldt-maintext}
    Assume $\charF(\Fbb)=2$ and $|\Fbb|\ge\max\cbra{6md,2^{100}m\log|\Fbb|}$.
    Let $\Sigma=\Fbb^{d+1}$ be the set of univariate degree-$d$ polynomials over $\Fbb$.
    There exists an efficient verifier $\pcppldt$ with the following properties.
    \begin{itemize}
        \item The input of $\pcppldt$ is $T\circ \pi$, where $T\in (\Fbb^t)^{|\Fbb|^m}$ is supposed to be a codeword of $\prm$ and $\pi\in (\Sigma^t)^{|\Fbb|^m\cdot (m\log |\Fbb|)^{O(1)}}$ is the auxiliary proof.
        \item $\pcppldt$ tosses $m\log |\Fbb|+O\pbra{\log \log |\Fbb| +\log m}$ unbiased coins and makes $2$ queries on $T\circ \pi$.
        \item If $T\in \eccim(\prm)$, then there exists some $\pi$ such that $\pcppldt(T\circ \pi)$ always accepts.
        \item If $T$ is $\delta$-far from $\eccim(\prm)$, then $\Pr[\pcppldt(T\circ \pi)\ \text{rejects}]\ge 2^{-40}\delta$ for any $\pi$.
    \end{itemize}
\end{theorem}

Given a Boolean circuit $C:\bin^k\to \{0,1\}$, we say a word $x\in (\{0,1\}^t)^k$ \emph{parallel satisfies} $C$ iff $C(x[j]) = 1$ holds for every coordinate $j\in[t]$. 

Extracting from the verifier in \Cref{thm:ldt-maintext}, we obtain a circuit of small size that describes the codeword testing procedure, the proof of which is also in \Cref{sec:lowdeg}.

\begin{theorem}\label{thm:pldt-circuit}
Assume $\charF(\Fbb)=2$ and $|\Fbb|\ge\max\cbra{6md,2^{100}m\log|\Fbb|}$.
There exists a Boolean circuit $C_\ldt$ of size $\cktsize$ for $T\in (\Fbb^t)^{|\Fbb|^m}$, where we encode $\Fbb$ as $\bin^{\log|\Fbb|}$, such that $T$ is codeword of $\prm$ iff $T$ parallel satisfies $C_\ldt$.    
\end{theorem}

\subsection{Pair Language and Probabilistic Checkable Proof of Proximity}

We will perform various satisfiability testing for assignments of Boolean circuits where both are given as input. 
This motivates the notion of pair language, where the input is naturally divided into two parts corresponding to circuits and assignments respectively.

\paragraph{Pair Language}
Formally, $L$ is a pair language over the alphabet $\Sigma_\pairx$ and $\Sigma_\pairy$ if all words of $L$ are in the form $(x,y)$ where $x\in \Sigma_\pairx^*,y\in \Sigma_\pairy^*$. For each $x\in \Sigma_\pairx^*$, we define $L(x)=\{y\in \Sigma_\pairy^*| (x,y)\in L\}$ as the restriction of $L$ on $x$.

\paragraph{Probabilistic checkable proof of proximity (PCPP)}
PCPP provides robust testing of satisfiability for pair languages.

\begin{definition}[PCPP]\label{def:pcpp}
Given a pair language $L$ with alphabet $\Sigma$, a $(q, r, \delta, \varepsilon, \Sigma_2)$-PCPP verifier $\mathcal P$ takes as input a pair of words $(x,y)$ and an auxiliary proof $\pi$ with alphabet $\Sigma_2$, such that:
\begin{enumerate}[label=\textbf{(T\arabic*)}]
    \item \label{def-pcpp-procedure} The verifier $\mathcal P$ reads all bits of $x$, tosses at most $r(|x|)$ many unbiased coins, makes at most $q(|x|)$ queries on $y$ and $\pi$, and then decides to accept or reject within runtime $(|x|+|\Sigma| + |\Sigma_2|)^{O(1)}$. 
    
    We use $I_r$ to denote the query positions on $y\circ \pi$ under randomness $r$, and use $\mathcal P(x,y,r,a)$ to indicate the behavior (i.e., accept or reject) of $\mathcal P$ under randomness $r$ and query answer $a$.
    \item \label{def-pcpp-completeness} If $(x,y)\in L$, then there exists some $\pi$ such that $\mathcal P$ always accepts.
    \item \label{def-pcpp-soundness} If $y$ is $\delta$-far from $L(x)$, then $\mathcal P$ rejects with probability at least $\varepsilon$ for every $\pi$.
\end{enumerate}
We say $\mathcal P$ is a $(q,r,\varepsilon, \Sigma_2)$-PCP verifier if it is a $(q,r,1,\varepsilon, \Sigma_2)$-PCPP verifier.
\end{definition}

\begin{remark}[Proof Length]\label{rmk:pcpp}
We can always upper bound $|\pi|$ by $2^{r(|x|)}q(|x|)$, for which reason we do not put an additional parameter in \Cref{def:pcpp}.
\end{remark}

\paragraph{From PCPP to CSP}
PCPPs are tightly connected with CSPs. The following standard reduction establishes the connection.
\begin{definition}[From PCPP to CSP]\label{def:csp-view-pcpp}
    Given a $(q,r,\delta,\varepsilon,\Sigma_2)$-PCPP verifier $\mathcal P$ for the pair language $L = \{(x,y) | x, y\in \Sigma^*\}$. We define a CSP instance $G' = (V', E', \Sigma', \{\Pi_{e}'\}_{e\in E'})$, where $V'=V_\pairy\dot\cup V_{\pi}\dot\cup V_{{\rm pcpp}}$ and $\Sigma'=\Sigma\cup \Sigma_2$, by the following steps:
    \begin{itemize}
        \item First, we treat each position of $y$ (resp., $\pi$) as a single variable in $V_\pairy$ and (resp., $V_{\pi}$) with alphabet $\Sigma$ (resp., $\Sigma_2$). Note that $|V_{y}| = |y|$ and $|V_{\pi}|\le 2^{r(|x|)}q(|x|)$ by \Cref{rmk:pcpp}.
        \item Then, for each choice of random coins $r\in \bin^{r(|x|)}$, let $S_r$ be the set of query positions over $y\circ \pi$ under randomness $r$; and we add a variable $z_r$ to $V_{\rm pcpp}$ whose alphabet is $(\Sigma\cup \Sigma_2)^{|S_r|}$, i.e., all possible configurations of the query result.
        Note that $|V_{\rm pcpp}|\le 2^{r(|x|)}$.
        
        \item Finally, we add constraints between $z_r$ and every query position $i\in S_r$. The constraint checks whether $z_r$ is an accepting configuration, and the assignment of the position $i$ is consistent with the assignment of $z_r$.
    \end{itemize}
\end{definition}

By construction, the completeness and soundness are preserved up to a factor of $q$ under this reduction, where the loss comes from splitting $q$ queries into $q$ consistency checks. In addition, the reduction from $\mathcal P$ to $G'$ is an FPT reduction.

\begin{fact}\label{fct:pcpp}
The reduction described in \Cref{def:csp-view-pcpp} is an FPT reduction. Recall that $\mathcal P$ is a $(q,r,\delta,\varepsilon,\Sigma_2)$-PCPP verifier for the pair language $L$ over the alphabet $\Sigma$.
We have the following properties for $G'$: 
    \begin{itemize}
        \item \textsc{Alphabet.}
        The alphabet of $G'$ is $\left(\Sigma\cup \Sigma_2\right)^q$.
        \item \textsc{Parameter Blowup.}
        The number of varialbes $|V'|\le |y|+2^{r(|x|)}q(|x|)+2^{r(|x|)}$. The number of constraints $|E'| = O(V'\cdot q(|x|))$.
        \item \textsc{Completeness.} If $(x,y)\in L$, then there exists a solution $\sigma'$ of $G'$ assigning $y$ to $V_\pairy$.
        \item \textsc{Soundness.} If $y$ is $\delta$-far from $L(x)$, then any assignment $\sigma'$ assigning $y$ to $V_\pairy$ satisfies at most $1-\frac{\varepsilon}{q}$ fraction of the constraints in $G'$. Note that if $L(x) = \emptyset$, then $\val(G')\le 1-\frac{\varepsilon}{q}$.
    \end{itemize}    
\end{fact}

\paragraph{Circuit Value Problem (\pcircuit{})}
\pcircuit{} is a standard pair language widely used in the classical PCP theorems (see e.g., \cite{BGH06}).
It directly checks if a binary string is a solution to a Boolean circuit.

\begin{definition}[\pcircuit{}]
\label{def:cktval}
    \pcircuit{} is a pair language over $\bin$ consisting of all words in the form of $w = (C,z)$, where $C$ is a Boolean predicate with $|z|$ input bits and $z$ is a binary string that satisfies $C$. 
    We define the input length $|w|$ to be the size of $C$ plus $|z|$.
\end{definition}

We quote the following almost-linear PCPP result for \pcircuit{}, which follows from \cite[Theorem~3.3]{BGH06} by setting $t$ to be a constant sufficiently large.

\begin{theorem}[PCPP for \pcircuit{}, \cite{BGH06}]\label{thm:pcpp-cktsat}
There exists an absolute constant $0<\pcppdelta\le2^{-100}$ such that \pcircuit{} has 
$$
\text{an $\pbra{O(1), \log|w| + O(\log^{0.1}|w|), \pcppdelta, \frac12, \bin}$-PCPP verifier $\pcppckt$,}
$$
where $|w|$ is the input length of \pcircuit{}.
\end{theorem}
\begin{remark}
In the parallel setting, we want to check if a circuit $C$ is satisfied on $t$ different inputs $z_1,z_2,\ldots,z_t$. 
Note that according to the \cref{def:pcpp}, given $(C,z_i)$ and an auxiliary proof $\pi_i$ for every $i\in[t]$, the PCPP reads the whole $C$ and then queries $z_i$ and $\pi_i$. In other words, the query locations depend only on $C$.  
So, when applied in parallel, the PCPP  queries the same locations for different $i\in[t]$.
\end{remark}

\section{Proof of the Main Theorems}\label{sec:pf_main}

This section is devoted to giving a landscape of the proof of \Cref{thm:main-intro} and \Cref{thm:new-pcp}. To depict a clear picture, we relegate the proof of technical lemmas in subsequent sections. Below, we first present the formal statement of \Cref{thm:main-intro} and \Cref{thm:new-pcp}.

\begin{theorem}[Formal Version of \Cref{thm:main-intro}]
\label{thm:main}
    Assume ETH. Then for some constant $\eps\in(0,1)$ and any computable function $f(K)$, no algorithm can take as input a 2CSP instance $\Lambda$ with $K$ variables, $O(K)$ constraints and alphabet $[N]$, distinguish between the following two cases in $f(K)N^{\lowerboundK}$ time:
    \begin{itemize}
        \item $\Lambda$ is satisfiable;
        \item any assignment satisfies at most $1-\varepsilon$ fraction of the constraints in $\Lambda$.
    \end{itemize}
\end{theorem}

\begin{theorem}[Formal Version of \Cref{thm:new-pcp}]\label{thm:new-pcp-restate}
For any integer $k\ll n$ sufficiently large, {\sc 3SAT} has a constant-query PCP verifier of alphabet size $|\Sigma| = \exp\!\pbra{\frac nk \cdot 2^{O(\sqrt {\log k})}}
$, runtime ${\rm poly}(|\Sigma|, n)$, and $
\log k+O\pbra{\sqrt{\log k}\log \log k}
$ random coins, which has perfect completeness and soundness $\frac12$.
\end{theorem}

The proof of \Cref{thm:main} follows a similar idea in the previous work~\cite{GLRSW24}.
\begin{itemize}
    \item First, we reduce from {\sc 4-Regular 3-Coloring}, which has the lower bound $2^{\Omega(|V|)}$ under ETH, to a CSP with specific vector structures, termed as the special vector-valued CSP (SVecCSP for short). This step is somewhat similar to the previous approach~\cite{GLRSW24}, in the sense that both approaches engineer non-parameterized CSP problems into parameterized CSPs with vector structure. However, the vector structure in~\cite{GLRSW24} is not enough for obtaining the almost-optimal lower bound for $k$-Variable CSP. This paper obtains a much more refined structure. 
    \item Second, we design a PCPP verifier for SVecCSPs. To obtain an almost-optimal lower bound, we encode the solution of SVecCSPs via Reed-Muller code with an almost-linear blowup. Below, we present the details of our proof.
\end{itemize}

\subsection{Reduction from {\sc 4-Regular 3-Coloring} to \svcsp{}}

In this section, we define \textit{special vector-valued CSP} (\svcsp{} for short). The notion of vector-valued CSPs was first defined in \cite{GLRSW24}, and here we consider a better-structured variant.

\begin{definition}[\svcsp{}]\label{def:svcsp}
A CSP instance $G = (V,E=E_\parallelc\dot\cup E_\linearc,\Sigma,\{\Pi_{e}\}_{e\in E})$ is an \svcsp{} if the following properties hold.
\begin{enumerate}[label=\textbf{(S\arabic*)}]
    \item $V = \{x_1,\ldots,x_k\}\dot\cup \{y_1,\ldots,y_k\}$.
    \item $\Sigma=\Fbb^t$ is a $t$-dimensional vector space over a finite field $\Fbb$ with $\charF(\Fbb)=2$.
    \item \label{itm:def:svcsp_2} For each constraint $e=\{u,v\}\in E_\parallelc$  where $u = (u_1,\dots,u_t)$ 
       and $v = (v_1,\ldots,v_t)$ are two variables in $V$, there is a sub-constraint $\Pi_{e}^{sub}: \Fbb\times \Fbb\to \{0,1\}$ such that $\Pi_{e}(u,v)$ checks $\Pi_{e}^{sub}$ on all coordinates, i.e., 
        $$
        \Pi_e(u,v) = \bigwedge_{i\in [t]} \Pi_e^{sub}(u_i,v_i).
        $$
    \item \label{itm:def:svcsp_1} $E_\linearc=\{\{x_i,y_i\}_{i\in[k]}\}$. For each $i\in [k]$, there exists a matrix $M_i\in \Fbb^{t\times t}$ and a constraint $\Pi_i(x_i,y_i)$ check if $y_i=M_ix_i$, i.e.,
\[
\Pi_i(x_i,y_i)=\indicator_{y_i = M_ix_i}.
\]
\end{enumerate}
\end{definition}

The following theorem establishes the hardness of \svcsp{}.

\begin{theorem}\label{thm:reduction_svcsp}
    There is a reduction algorithm that holds the following. Given as input an integer $6\le k\le n$ and an $n$-variable {\sc 4-Regular 3-Coloring} instance $\Gamma$, it produces an \svcsp{} instance $G=(V,E,\Sigma,\{\Pi_e\}_{e \in E})$ where:
    \begin{enumerate}[label=\textbf{(R\arabic*)}]
        \item\label{itm:thm:3sat2csp_1} \textsc{Variables and Constraints.} $|V|=O(k)$ and $|E|=O(k)$.
            \item\label{itm:thm:3sat2csp_2} \textsc{Runtime.} The reduction runs in time $\poly(n,2^{n \log k/k})$.
        \item\label{itm:thm:3sat2csp_3} \textsc{Alphabet.} $\Sigma = \Fbb_4^{t}$ where $t=O\pbra{\frac{n\log k}k}$.
        \item\label{itm:thm:3sat2csp_5} \textsc{Satisfiability.} $G$ is satisfiable iff $\Gamma$ is satisfiable.
    \end{enumerate}
\end{theorem}

We defer the proof of \Cref{thm:reduction_svcsp} to \Cref{sec:prepro}. 
Compared with vector-valued CSPs (\vcsp{}) used in the previous work~\cite[Definition 3.3]{GLRSW24}, \svcsp{} offers more structural properties.
\begin{itemize}
    \item For parallel constraints, the previous work~\cite{GLRSW24} sets up the sub-constraint over \emph{a subset of the coordinates}. 
    By construction, there might be $2^{|E_{\sf p}|}$ different sub-CSP instances over all coordinates. Each of the $2^{|E_{\sf p}|}$ sub-CSPs requires an individual PCPP verifier. To simultaneously check the satisfiability of all these sub-instances, they tuple these verifiers into a giant verifier, resulting in an exponential blowup of the proof length. Hence, an almost-linear size of proof is impossible.

    In contrast, \svcsp{} sets up the sub-constraint on \emph{all} coordinates, which implies a unified sub-CSP instance among all coordinates. As a result, we avoid the tuple procedure, enabling a highly succinct proof.
    \item For linear constraints,
    \vcsp{} defined in ~\cite{GLRSW24} allow them to be over arbitrary pairs of variables. They introduce auxiliary variables for all pairs of variables and their corresponding linear constraints to check such unstructured constraints. 
    The number of auxiliary variables is $|V|\cdot|E_{\sf l}|$, which means that the proof length is at least quadratic, making an almost linear size of proof impossible.
    
    However, \svcsp{} only sets up linear constraints over $x_i$ and $y_i$, with the same index $i$. After encoding $x$ and $y$ by the parallel Reed-Muller code, we can leverage this alignment and introduce auxiliary proofs, which is also a codeword of the parallel Reed-Mulle code, to check the validity of linear constraints efficiently, with an almost-linear blowup.
\end{itemize}

Though \svcsp{} has a more refined structure than \vcsp{}, we can obtain it from \vcsp{} by properly duplicating variables and relocating constraints. 
However, the \vcsp{} instance in \cite{GLRSW24} has parameters $|V| = O(k^2), |E| = O(k^2)$, which are too large to obtain~\Cref{thm:reduction_svcsp}.
To get around this, we combine results from \cite{Mar10,KMPS23} and the reduction in \cite{GLRSW24} to obtain a sparse \vcsp{} instance with $|V| = O(k), |E| = O(k)$, which ultimately leads to \Cref{thm:reduction_svcsp}.

\subsection{Reduction from SVecCSPs to \texorpdfstring{$\varepsilon$}{ε}-{\sc Gap} \texorpdfstring{$k$}{k}-{\sc Variable CSPs}}

With the ETH-hardness of \svcsp{} by \Cref{thm:reduction_svcsp}, the remaining work is to reduce \svcsp{} to \gapkcsp{}. 
To this end, we follow the same idea of previous works~\cite{ALM+98, GLRSW24}.
\begin{enumerate}[label=\textbf{(L\arabic*)}]
    \item \label{red:1} First, we encode the solution using an error correcting code.
    \item \label{red:2} Then, we design a constant-query PCPP verifier to check whether the given proof is the codeword of some solution with the aid of auxiliary proofs. 
    \item \label{red:3} Finally, we obtain \Cref{thm:main} by converting the PCPP verifier into an instance of \gapkcsp{} by \Cref{fct:pcpp}.
\end{enumerate}

\begin{remark}
In our actual construction of \Cref{red:2}, parallel constraints and linear constraints will be processed separately, building on a (supposedly) same solution.
Hence we need to encode the solution using error correcting codes in \Cref{red:1} and design PCPP verifier (instead of PCP verifier) for \Cref{red:2} on top of the shared encoding of the solution.
\end{remark}

For the first step~(\Cref{red:1}), we need error correcting codes with almost-linear length blowup. In detail, we choose the parallel Reed-Muller code~(see \Cref{def:prm}) with suitable choice of parameters. This motivates us to define the following pair language \svcsp{} Satisfiability (\svsat{}).

\begin{definition}[\svsat{}]\label{def:svsat}
$(\Fbb, m, d, t)$-\svsat{} is a pair language consisting of $w=(G, (\hat{x}, \hat{y}))$ where:
\begin{itemize}
    \item $G = (V=\{x_1,\ldots,x_k\} \dot\cup \{y_1,\ldots,y_k\},E = E_\parallelc\dot\cup E_\linearc,\Sigma = \Hbb^t, \{\Pi_e\}_{e\in E})$ is an \svcsp{} instance. We require that $k\le \binom{m+d}{d}$ and $\Hbb$ is a subfield of $\Fbb$. 
    \item $\hat{x}, \hat{y}$ are codewords of $\prm$. Suppose $\hat{x} = \prm(\sigma_x)$ and $\hat{y} = \prm(\sigma_y)$ for some $\sigma_x, \sigma_y\in (\Fbb^t)^{\binom{m+d}{d}}$. 
    Define assignment $\sigma: V\to \Fbb^t$ by
    $$
    \sigma(v) := \begin{cases}
        \sigma_{x}(i) & v = x_i,\\
        \sigma_{y}(i) & v = y_i.
    \end{cases}
    $$
    We further require that $\sigma$ is a solution of $G$, which implicitly demands $\sigma(v)\in\Hbb^t$ for $v\in V$.
\end{itemize}
\end{definition}

\begin{remark}
Since $\binom{m+d}{d}\ge k$, the index $i$ in defining $\sigma$ is at most the length of $\sigma_x,\sigma_y$ and thus $\sigma$ is well defined in \Cref{def:svsat}. 
The parameters $\Fbb,m,d,t$ come from the parameters of the parallel Reed-Muller code. 
\end{remark}

In this pair language, $G$ is the starting \svcsp{} instance in our reduction, and $(\hat{x}, \hat{y})$ serves as the encoding of a solution. 
We have the following connection between $(\Fbb, m, d, t)$-\svsat{} and the satisfiability of \svcsp{}.

\begin{fact}
\label{fct:svsat}
    Let $G = (V=\{x_1,\ldots,x_k\} \dot\cup \{y_1,\ldots,y_k\},E = E_\parallelc\dot\cup E_\linearc,\Sigma = \Hbb^t, \{\Pi_e\}_{e\in E})$ be an \svcsp{} instance.
    Assume $k\le \binom{m+d}{d}$ and $\Hbb$ is a subfield of $\Fbb$. Then, $G$ is satisfiable iff there exists $(\hat{x}, \hat{y})$ such that $(G, (\hat{x}, \hat{y}))\in(\Fbb, m, d, t)$-\svsat{}.
\end{fact}

\begin{remark}[Encoding Choice]
One may wonder the necessity for encoding $\hat x$ and $\hat y$ separately, as opposed to encoding them jointly as $\hat{x\circ y}$.
This is due to the bipartite structure of the linear constraints $E_\linearc$ (see \Cref{itm:def:svcsp_1}) that we will need to ensure proximity for the encoding on both $x$ and $y$. This will be clear in \Cref{itm:pcpp-3} and \Cref{lem:pcpp-s}.
\end{remark}

For the second step~(\Cref{red:2}), we need to construct a PCPP verifier for $(\Fbb, m, d, t)$-\svsat{}. Formally, we construct a PCPP verifier with the following parameters.

\begin{theorem}\label{thm:pcp-svcsp}
 Assume $\charF(\Fbb) = 2$ and $|\Fbb|\ge\max\{12md, 2^{101} m \log |\Fbb| \}$. Then for any $\delta \in[0,1]$, $(\Fbb, m, d, t)$-\svsat{} has 
 $$
 \text{ and $\pbra{O(1), \log(|E|+|\Fbb|^m)+O\pbra{\log^{0.1}(|E|+|\Fbb|^m)+\log|\Fbb|} , \delta, \Omega(\delta), \Sigma^t}$-PCPP verifier $\Pcal$,}
 $$ 
 where $|E|$ is the number of constraints in $(\Fbb,m,d,t)$-\svsat{} and $\Sigma=\Fbb^{d+1}$.
\end{theorem}

Fix some $(G, (\hat{x}, \hat{y}))$ supposed to be in $(\Fbb, m, d, t)$-\svsat{}. We use $\sigma$ to denote the assignment recovered (if succeeded) from $\hat{x}$ and $\hat{y}$. Our goal is to verify whether $\sigma$ is the solution of $G$.
To this end, we will handle parallel and linear constraints in \Cref{sec:svsat_parallel} and \Cref{sec:svsat_whole} separately, then combine them together in \Cref{sec:svsat_whole} to prove \Cref{thm:pcp-svcsp}.

\subsubsection{Verification of Parallel Constraints}\label{sec:svsat_parallel}

We first consider verifying whether $\sigma$ satisfies all parallel constraints by probing the indices of $(\hat{x}, \hat{y})$. Since $\prm$ is a systematic code~(\Cref{def:prm}), we can recover the value of an index in the assignment $\sigma$ directly from probing an index of $\hat{x}\circ \hat{y}$. In addition, recall that parallel constraints set up the same sub-constraints over all coordinates~(\Cref{itm:def:svcsp_2}), we can build up a unified Boolean circuit $C_\parallelc: \bin^{2|\Fbb|^m\log |\Fbb|}\to \bin$ that checks whether the systematic part of $\hat{x}$ and $\hat{y}$ satisfies the sub-constraints in a single coordinate. In detail, the circuit $C_\parallelc$ is defined as follows.

\begin{definition}[The Circuit $C_\parallelc$]
\label{def:cp}
    Let $\tilde x,\tilde y$ be words of length $|\Fbb|^m$ over alphabet\footnote{For ease of presentation, the circuit's input is described as alphabet $\Fbb$. We take the trivial conversion from an element of $\Fbb$ into a binary string of length $\log |\Fbb|$.} $\Fbb$.
    Assume $\tilde x,\tilde y$ are codewords of $\prm$.
    The circuit $C_\parallelc$ executes the following.
\begin{itemize}
    \item $C_\parallelc$ recovers the messages $\sigma_{\tilde x}, \sigma_{\tilde y}$ from the systematic part of $\tilde x,\tilde y$ respectively.
    \item After that, $C_\parallelc$ checks whether the assignment $\tilde\sigma$ specified by $\sigma_{\tilde x},\sigma_{\tilde y}$ has the correct subfield entries and satisfies all constraints in the $E_{\parallelc}$ part of $G$, at the single-coordinate level. 
    Specifically, $\tilde\sigma$ is the assignment of $V$ defined by 
    \begin{equation*}
    \sigma(v) = \begin{cases}
        \sigma_{\tilde x}(i) & v = x_i,\\
        \sigma_{\tilde y}(i) & v = y_i.
    \end{cases}
    \end{equation*}
    For every $v \in V$, $C_\parallelc$ checks whether $\sigma(v) \in \Hbb$; and for every constraint $e=\{u,v\} \in E_{\textsf p}$, $C_\parallelc$ checks whether $\Pi_e^{sub}(\sigma(u),\sigma(v))=1$. 
\end{itemize}
The size of circuit $C_{\sf p}$ is bounded by $(|\Fbb|^m + k + |E_\parallelc|)\cdot\poly |\Fbb|=(|\Fbb|^m + |E|)\cdot\poly |\Fbb|$.
\end{definition}

\paragraph{Double Test Problem (\doubletest{})}
In light of \Cref{def:cp}, $(\hat{x}, \hat{y})$ satisfies all parallel constraints iff $\hat x,\hat y$ are correct codewords and $C_\parallelc(\hat{x}[i]\circ \hat{y}[i]) = 1$ for each $i\in [t]$. This motivates us to consider the following pair language \doubletest{} related to \pcircuit{}.

\begin{definition}[\doubletest{}]\label{def:doubletest}
Assume $\charF(\Fbb)=2$. $(\Fbb,m,d,t)$-\doubletest{} is a pair language over $\Sigma_\pairx=\bin,\Sigma_\pairy=\Fbb^t$ consisting of $w=(C,T_1\circ T_2)$ where
\begin{itemize}
\item $C$ is a Boolean circuit with $2|\Fbb|^m\log|\Fbb|$ input bits and $T_1,T_2\in(\Fbb^t)^{|\Fbb|^m}$ are codewords of $\prm$;
\item if we view $\Fbb$ as $\bin^{\log|\Fbb|}$, $T_1\circ T_2$ parallel satisfies $C$.
\end{itemize}
We define the input length $|w|$ to be the size of $C$ plus $2|\Fbb|^m\log|\Fbb|$. Note that the dimension $t$ is reflected on the alphabet, not on the length.
\end{definition}

In short, \doubletest{} extends \pcircuit{} by allowing the assignment to have more dimensions (i.e., $t$), allowing the assignment to be partitioned into two parts (i.e., $T_1,T_2$), and assuming each part is encoded by parallel RM code.

Given \Cref{def:doubletest} and \Cref{def:cp}, we have the following statement.
Note that the assumption that $\tilde x,\tilde y$ are codewords of $\prm$ in \Cref{def:cp} is guaranteed in \Cref{def:doubletest}.
\begin{fact}\label{fct:parallel-checking}
Parallel constraints $E_\parallelc$ in $G$ are satisfied iff $(C_\parallelc, \hat{x}\circ \hat{y})\in (\Fbb, m, d, t)$-\doubletest{}.  
\end{fact}

Thus, to verify that $\sigma$ satisfies all parallel constraints, if suffices to construct a PCPP verifier for \doubletest{}. The verifier is formally stated as follows and will be proved in \Cref{sec:pcpp-probttt}.

\begin{theorem}[PCPP for \doubletest{}{}]\label{thm:pcpp-doubletest}    
Assume $\charF(\Fbb)=2$ and $|\Fbb|\ge\max\cbra{6md,2^{100}m\log|\Fbb|}$.
For any $\delta \in[0,1]$, $(\Fbb, d, m, t)$-\doubletest{} has 
$$
\text{an $\pbra{O(1), \log|w| + O\pbra{\log^{0.1}|w|+\log|\Fbb|}, \delta,\Omega(\delta), \Sigma^{t}}$-PCPP verifer $\pcppdoubletest$,}
$$
where $|w|$ is the input length of $(\Fbb,d,m,t)$-\doubletest{} and $\Sigma=\Fbb^{d+1}$.
\end{theorem}

\subsubsection{Verification of Linear Constraints}\label{sec:svsat_linear}

We then turn to verifying whether $\sigma$ satisfies all linear constraints $E_\linearc$. 
Recall from \Cref{itm:def:svcsp_1} that all linear constraints are set up between $x_i$ and $y_i$, with the constraint that $y_i = M_ix_i$. 
By defining auxiliary variables $z_i := y_i - M_ix_i$, it suffices to check whether $z_i \equiv \vec{0}_t$ for every $i\in [k]$.

Thus, we add the auxiliary proof $\hat{z}\in (\Fbb^t)^{|\Fbb|^m}$, which is supposed to fulfill the following condition
\begin{equation}\label{eq:hatz}
\hat{z}(p) = \hat{y}(p) - \hat{M}(p) \hat{x}(p)
    \quad\text{holds for all $p\in\Fbb^m$,}
\end{equation}
where $\hat M \in (\Fbb^{t\times t})^{|\Fbb|^m}$ is the parallel RM encoding of $(M_1,\ldots,M_k,0^{t\times t},\ldots,0^{t\times t})\in(\Fbb^{t\times t})^{\binom{m+d}d}$.
Here we extend \Cref{def:prm} for matrix values: the value of $\xi_i$ is $M_i$ for $i\le k$ and is $0^{t\times t}$ for $i>k$.
Equivalently, for each matrix coordinate $(i,j)\in[t]\times[t]$, $\hat M[i,j]\in\Fbb^{|\Fbb|^m}$ is the RM encoding of $(M_1[i,j],\ldots,M_k[i,j],0,\ldots,0)$.
We remark that entries of $\hat M$ can be efficiently computed on the fly by demand and are not included as a proof for the PCPP verifier. 

Based on the discussion above, we obtain the following fact.

\begin{fact}\label{fct:checking_z_linear}
Linear constraints $E_\linearc$ in $G$ are satisfied iff $\hat x,\hat y$ are codewords of $\prm$ and the systematic part of $\hat{z}$, defined by \Cref{eq:hatz}, are all $0^t$.
\end{fact}

Recall that $\hat x,\hat y$ being correct codewords are guaranteed by the analysis of parallel constraints above.
Hence we focus on testing the systematic part of $\hat z$, which amount to $\hat z$ parallel satisfying the following circuit.

\begin{definition}[The Circuit $C_\linearc$]\label{def:cl}
The circuit $C_\linearc$ receives as input a word $\tilde z$ of length $|\Fbb|^m$ over alphabet $\Fbb$.
It checks if $\tilde z(\xi_i) = 0$ holds for all $i\in[k]$, where we recall that $\xi_i$ from \Cref{def:prm}.

The size of the circuit $C_\linearc$ is bounded by $(|\Fbb|^m + k)\cdot\poly |\Fbb|$.
\end{definition}

We will use a variant of \doubletest{}, denoted \singletest{}, on $C_\linearc$ and $\hat z$. But before that, we give a degree bound for $\hat z$.

\begin{claim}\label{clm:hatz}
    If $\hat{x}, \hat{y}$ are codewords of $\prm$, then the $\hat{z}$ defined by \Cref{eq:hatz} is a codeword of ${\rm RM}^{\Fbb,m,2d,t}$.
\end{claim}
\begin{proof}
Expanding the matrix multiplication, for each coordinate $i \in [t]$, we have 
$$
\hat z[i](p) = \hat y[i](p) - \sum_{j \in [t]} \hat M[i,j](p)\cdot\hat x[j](p)
\quad\text{for all $p\in\Fbb^m$,}
$$ 
where $\hat M[i,j](p)$ is the $(i,j)$-th entry of $\hat M(p)$.
Since $\hat y[i],\hat M[i,j],\hat x[j]$ are truth tables of degree-$d$ polynomials, $\hat z[j]$ is the truth table of a degree-$2d$ polynomial, which means $\hat z\in\eccim(\textsf{RM}^{\Fbb,m,2d,t})$.
\end{proof}

\paragraph{Single Test Problem (\singletest{})}
At this point, checking \Cref{fct:checking_z_linear} can be safely handled by the following \singletest{} with proper degree conditions.

\begin{definition}[\singletest{}]\label{def:singletest}
Assume $\charF(\Fbb)=2$. $(\Fbb,m,d,t)$-\singletest{} is a pair language over $\Sigma_\pairx=\bin,\Sigma_\pairy=\Fbb^t$ consisting of $w=(C,T_1)$ where
\begin{itemize}
\item $C$ is a Boolean circuit with $2|\Fbb|^m\log|\Fbb|$ input bits and $T_1\in(\Fbb^t)^{|\Fbb|^m}$ is a codeword of $\prm$;
\item if we view $\Fbb$ as $\bin^{\log|\Fbb|}$, $T_1$ parallel satisfies $C$.
\end{itemize}
We define the input length $|w|$ to be the size of $C$ plus $|\Fbb|^m\log|\Fbb|$.
\end{definition}

In comparison, \singletest{} simply removes the second table $T_2$ from \doubletest{}.
Combining \Cref{clm:hatz} and \Cref{fct:checking_z_linear}, we have the following result.

\begin{fact}\label{fct:hatz}
Linear constraints $E_\linearc$ in $G$ are satisfied iff $\hat x,\hat y$ are codewords of $\prm$, $\hat z$ satisfies \Cref{eq:hatz}, and $(C_{\linearc}, \hat{z})\in(\Fbb,m,2d,t)$-\singletest{}.
\end{fact}

Analogous to \Cref{thm:pcpp-doubletest}, we also have an efficient PCPP verifier for \singletest{} as follows, which will also be proved in \Cref{sec:pcpp-probttt}.

\begin{theorem}[PCPP for \singletest{}{}]\label{thm:pcpp-singletest}    
Assume $\charF(\Fbb)=2$ and $|\Fbb|\ge\max\cbra{6md,2^{100}m\log|\Fbb|}$.
For any $\delta \in[0,1]$, $(\Fbb, d, m, t)$-\singletest{} has 
$$
\text{an $\pbra{O(1), \log|w| + O\pbra{\log^{0.1}|w|+\log|\Fbb|}, \delta,\Omega(\delta), \Sigma^{t}}$-PCPP verifer $\pcppsingletest$,}
$$
where $|w|$ is the input length of $(\Fbb,d,m,t)$-\singletest{} and $\Sigma=\Fbb^{d+1}$.
\end{theorem}

Finally, we briefly sketch how to check if $\hat z$ satisfies \Cref{eq:hatz}, as needed in \Cref{fct:hatz}. This will be done by randomly picking an index $p$ and checking whether \Cref{eq:hatz} holds on that $p$. The soundness will be analyzed by Schwartz-Zippel Lemma. This will be formalized in \Cref{sec:svsat_whole}.

\subsubsection{The Whole Construction and Analysis}\label{sec:svsat_whole}

Based on the above discussion, we are now ready to construct the PCPP verifier $\Pcal$ for \svsat{}.
We invoke \Cref{thm:pcpp-doubletest} to obtain PCPP verifiers $\pcppdoubletest$ and $\pcppsingletest$ for $(\Fbb,m,d,t)$-\doubletest{} and $(\Fbb,m,2d,t)$-\singletest{} respectively. 

Recall that the input of $(\Fbb,m,d,t)$-\svsat{} is $(G,(\hat x,\hat y))$.
The auxiliary proof consists of $\hat z$, $\pi_1$, and $\pi_2$, where
\begin{itemize}
    \item $\hat z$ is supposed to be a codeword in $\textsf{RM}^{\Fbb,m,2d,t}$ and $\hat z(p) = \hat y(p)-\hat M(p)\hat x(p)$ for all $p\in\Fbb^m$;
    \item $\pi_1$ is supposed to be the auxiliary proof to convince $\pcppdoubletest$ that $(C_{\parallelc},\hat x \circ \hat y)$ belongs to the pair language $(\Fbb,m,d,t)$-\doubletest{}.
    \item $\pi_2$ is supposed to be the auxiliary proof to convince $\pcppsingletest$ that $(C_{\linearc}, \hat z)$ belongs to the pair language $(\Fbb,m,2d,t)$-\singletest{}.
\end{itemize} 

The verifier $\Pcal$ performs one of the following three tests with equal probability.
\begin{enumerate}[label=\textbf{(T\arabic*)}]
    \item \label{itm:pcpp-1} Feed the pair of words $(C_{\parallelc},\hat x \circ \hat y)$ and the auxiliary proof $\pi_1$ into $\pcppdoubletest$. Reject if $\pcppdoubletest$ rejects.
    \item \label{itm:pcpp-2} Feed the pair of words $(C_{\linearc},\hat z)$ and the auxiliary proof $\pi_2$ into $\pcppsingletest$. Reject if $\pcppsingletest$ rejects.
    \item \label{itm:pcpp-3} Generate a random point $p \in \Fbb^m$, reject if $\hat z(p) \ne \hat y(p)-\hat M(p)\hat x(p)$.
\end{enumerate}

At this point, we are ready to analyze the PCPP verifier $\Pcal$. In particular, \Cref{thm:pcp-svcsp} follows from the combination of the following \Cref{lem:pcpp-para}, \Cref{lem:pcpp-c}, and \Cref{lem:pcpp-s}.

\begin{lemma}[Parameters]\label{lem:pcpp-para}
Assume $\charF(\Fbb)=2$ and $|\Fbb|\ge\max\cbra{12md,2^{101}m\log|\Fbb|}$. 
Then $\Pcal$ tosses 
$$
\log(|E|+|\Fbb|^m) + O\pbra{\log^{0.1}(|E|+|\Fbb|^m) +\log |\Fbb|}.
$$ 
unbiased coins and makes $O(1)$ queries. The alphabet of the auxiliary proof is $\Sigma^{t}$ where $\Sigma=\Fbb^{d+1}$.
\end{lemma}
\begin{proof}
    By \Cref{thm:pcpp-doubletest} and \Cref{thm:pcpp-singletest}, both $\pcppdoubletest$ and $\pcppsingletest$ make constant queries. Also in \Cref{itm:pcpp-3}, $\Pcal$ makes 3 queries on $\hat x, \hat y, \hat z$. Thus $\Pcal$ makes $O(1)$ total queries.

    By \Cref{def:cl,def:cp}, the input lengths of $\pcppdoubletest$ and $\pcppsingletest$ in \Cref{itm:pcpp-1} and \Cref{itm:pcpp-2} are
    $$\left\{\begin{aligned}
        |w_1| & = (|\Fbb|^m + |E|) \cdot \poly |\Fbb| + 2 |\Fbb|^m \log |\Fbb| & \le (|E|+|\Fbb|^m) \poly |\Fbb|, \\
        |w_2| & = (|\Fbb|^m +k)\cdot \poly |\Fbb| + |\Fbb|^m \log |\Fbb| & \le (|E|+|\Fbb|^m) \poly |\Fbb|,
    \end{aligned}\right.$$
    respectively.
    Putting this into \Cref{thm:pcpp-doubletest} and \Cref{thm:pcpp-singletest}, the number of unbiased coins used in $\pcppdoubletest$ and $\pcppsingletest$ is
    $$
    \log(|E|+|\Fbb|^m) + O\pbra{\log^{0.1}(|E|+|\Fbb|^m) +\log |\Fbb|}. 
    $$
    In \Cref{itm:pcpp-3}, $\Pcal$ tosses $m \log |\Fbb|$ coins. Since $\Pcal$ only executes one of the three tests, the randomness is bounded by their maximum. 

    The auxiliary proofs $\pi_1,\pi_2$ have alphabet $\Sigma$ by \Cref{thm:pcpp-doubletest} and \Cref{thm:pcpp-singletest}. The alphabet of $\hat z$ is $\Fbb^t$, which can also be embedded into the larger $\Sigma^t$.
\end{proof}

\begin{lemma}[Completeness]\label{lem:pcpp-c}
Assume $\charF(\Fbb)=2$ and $|\Fbb|\ge\max\cbra{12md,2^{101}m\log|\Fbb|}$. 
Suppose $\hat x= \prm(\sigma_x),\hat y = \prm(\sigma_y)$, and the assignment $\sigma$ given by $\sigma_x$ and $\sigma_y$ (recall \Cref{def:svsat}) is a solution to $G$. 
Then there exist $\hat z,\pi_1,\pi_2$ which $\Pcal$ accepts with probability $1$.
\end{lemma}
\begin{proof}
    By \Cref{def:cp}, $(\hat x,\hat y)$ parallel satisfies the circuit $C_{\parallelc}$. Thus $(C_{\parallelc},\hat x\circ \hat y)\in(\Fbb,m,d,t)$-\doubletest{} by \Cref{fct:parallel-checking}. 
    Hence there exists an auxiliary proof $\pi_1$ which makes $\pcppdoubletest$ accepts with probability 1. \Cref{itm:pcpp-1} therefore always passes.

    For each $p \in \Fbb^m$, define $\hat z(p) = \hat y(p)-\hat M(p)\hat x(p)$. By \Cref{clm:hatz}, $\hat z\in\eccim(\textsf{RM}^{\Fbb,m,2d,t})$. 
    Let $\sigma_x,\sigma_y,\sigma_z$ be the messages $\hat x,\hat y,\hat z$ encodes respectively. 
    Note that $(\sigma_x,\sigma_y)$ satisfies the linear constraints $E_{\mathsf l}$ in $G$, i.e., $\sigma_y(i)-M_i\sigma_x(i)=0^t$ for all $i\in[k]$.
    Hence all $i \in [k]$, we have
    $$
    \sigma_z(i) = \hat z(\xi_i)= \hat y(\xi_i)-\hat M(\xi_i)\hat x(\xi_i)= \sigma_y(i)-M_i \sigma_x(i) = 0^t,
    $$ 
    where $\{\xi_1,\ldots,\xi_{\binom{m+d}{d}}\}$ are the distinct points defining the encoding of parallel RM code (see \Cref{def:prm}). 
    Therefore, $\hat z$ parallel satisfies $C_{\linearc}$. 
    By \Cref{fct:hatz}, there exists an auxiliary proof $\pi_2$, which makes $\pcppsingletest$ accepts with probability 1, and \Cref{itm:pcpp-2} always passes.

    Finally \Cref{itm:pcpp-3} always passes due to the definition of $\hat z$. This completes the proof.
\end{proof}

\begin{lemma}[Soundness]\label{lem:pcpp-s}
Assume $\charF(\Fbb)=2$ and $|\Fbb|\ge\max\cbra{12md,2^{101}m\log|\Fbb|}$. 
Let $\delta\in[0,1]$ be arbitrary.
If $(\hat x,\hat y)$ is $\delta$-far from satisfying, i.e., $\delta$-far from the restriction of $(\Fbb,m,d,t)$-\textsc{SVSat} on $G$, then $\Pcal$ rejects with probability $\Omega(\delta)$.
\end{lemma}
\begin{proof}
Let $\kappa\ge1$ be a large constant, the specific value of which depends on the hidden constants in \Cref{thm:pcpp-doubletest} and \Cref{thm:pcpp-singletest}.
By modifying the hidden constant in $\Omega(\cdot)$ here and noticing that $\delta$-far implies $\delta'$-far for any $\delta'\le\delta$, we safely assume $\delta\le1/\kappa^2$.

Fix arbitrary $(\hat x,\hat y)$ that is $\delta$-far from satisfying. Assume that $\Pcal$ rejects with probability at most $\kappa\cdot\delta$, since otherwise the statement already holds.
Then each of the tests \Cref{itm:pcpp-1}, \Cref{itm:pcpp-2}, and \Cref{itm:pcpp-3} reject with probability at most $3\kappa\cdot\delta$.
By choosing $\kappa$ sufficiently large and according to soundness guarantee of $\pcppdoubletest$ and $\pcppsingletest$ in \Cref{thm:pcpp-doubletest} and \Cref{thm:pcpp-singletest}, we know
\begin{enumerate}
    \item\label{itm:svsat_soundness_1} $(\hat x,\hat y)$ is $\delta$-close to $(\bar x,\bar y)$, which is a pair of codewords of $\prm$ that parallel satisfies $C_{\parallelc}$. 
    
    Since $\hat x,\hat y$ have the same length, this alse implies that $\hat x$ is $2\delta$-close to $\bar x$ and $\hat y$ is $2\delta$-close to $\bar y$. 
    \item\label{itm:svsat_soundness_2} $\hat z$ is $\delta$-close to $\bar z$, which is a codeword of $\textsf{RM}^{\Fbb,m,2d,t}$ that parallel satisfies $C_{\linearc}$.
\end{enumerate} 
We aim to show that $(G,(\bar x,\bar y))\in(\Fbb,m,d,t)$-\svsat{}, which contradicts to the assumption that $(\hat x,\hat y)$ is $\delta$-far from satisfying and completes the proof.

Let $\sigma_{\bar x}$, $\sigma_{\bar y}$, $\sigma_{\bar z}$ be the messages of $\bar x,\bar y, \bar z$ respectively.
It now suffices to prove the assignment $\sigma$ given by 
\begin{equation*}
\sigma(v) = \begin{cases}
    \sigma_{\bar x}(i) & v = x_i\\
    \sigma_{\bar y}(i) & v = y_i
\end{cases}
\end{equation*}
satisfies all constraints in $E_{\parallelc}\dot\cup E_{\linearc}$.

By \Cref{itm:svsat_soundness_1} and \Cref{fct:parallel-checking}, $\sigma$ satisfies all constraints in $E_{\parallelc}$. 
To analyze constraints in $E_\linearc$, we first prove that $\bar z=\bar y-\hat M\bar x$ in accordance with \Cref{eq:hatz}.
Assume this is false for some entry $p\in\Fbb^m$, i.e.,
\begin{equation}\label{eq:barz_soundness}
\bar z(p) \neq \bar y(p)-\hat M(p)\bar x(p).
\end{equation}
Note that $\bar x,\bar y,\hat M$ are all of parallel degree-$d$ and $\bar z$ is of parallel degree-$2d$.
Then by Schwartz-Zippel lemma, \Cref{eq:barz_soundness} actually happens for at least $1-\frac{2d}{|\Fbb|}$ fraction of points $p\in\Fbb^m$.
Now recall the test in \Cref{itm:pcpp-3}, which checks precisely the above for a random $p\sim\Fbb^m$ with $\bar x,\bar y,\bar z$ replaced by $\hat x,\hat y,\hat z$.
By \Cref{itm:svsat_soundness_1,itm:svsat_soundness_2} and a union bound, with probability at least $1-5\delta-\frac{2d}{|\Fbb|}$, on this random $p$ we have $\hat x(p)=\bar x(p),\hat y(p)=\bar y(p),\hat z(p)=\bar z(p)$ and \Cref{eq:barz_soundness} happens, which makes \Cref{itm:pcpp-3} reject.
By our assumption on $|\Fbb|$ and $\delta\le1/\kappa^2$ with $\kappa$ sufficiently large, this rejection probability is at least $0.9>3\kappa\cdot\delta$ and contradicts to our assumption on the rejection probability of \Cref{itm:pcpp-3}.
In short, \Cref{eq:barz_soundness} can never happen.

Finally we are ready to show that constraints in $E_\linearc$ are satisfied by $\sigma$.
By \Cref{itm:svsat_soundness_2} and \Cref{fct:hatz}, $\sigma_{\bar z}(i) = 0^t$ holds for all $i\in[k]$, and thus
$$
\sigma_{\bar y}(i)-M_i \sigma_{\bar x}(i)
= \bar y(\xi_i) - \hat M(\xi_i) \bar x(\xi_i) 
= \bar z(\xi_i)
= \sigma_{\bar z}(i) = 0.
$$
Therefore, all constraints in $E_{\linearc}$ are also satisfied. This completes the whole soundness proof.
\end{proof}

\subsection{Putting Everything Together}

Now, we are ready to prove the main theorems. 

\begin{proof}[Proof of \Cref{thm:main}]
We start with an arbitrary $n$-variable {\sc 4-Regular 3-Coloring} instance $\Gamma$, which prohibits algorithms of runtime $2^{o(n)}$ in the worst case by \Cref{thm:eth_3color} assuming ETH.
By \Cref{thm:reduction_svcsp}, we obtain an \svcsp{} instance $G=(V,E,\Fbb_4^t,\{\Pi_e\}_{e\in E})$ in time $\poly(n,2^{n\log k/k})$ which preserves the satisfiability of $\Gamma$.
In addition, $t=O\pbra{\frac{n\log k}k}$ and $|V|=O(k),|E|=O(k)$.

Let $m$ and $d$ be integers to be chosen later satisfying
\begin{equation}\label{eq:thm:main_1}
k\le\binom{m+d}d.
\end{equation}
Let $\Fbb$ be a field of characteristic two which contains $\Fbb_4$ as a subfield and satisfies
\begin{equation}\label{eq:thm:main_2}
|\Fbb|\ge\max\cbra{12md,2^{101}m\log|\Fbb|}.
\end{equation}
By \Cref{fct:svsat}, $G$ is satisfiable iff there exists $\hat x,\hat y\in(\Fbb^t)^{|\Fbb|^m}$ such that $(G,(\hat x,\hat y))\in(\Fbb,m,d,t)$-\svsat{}.
Then we construct the PCPP verifier $\Pcal$ for $(\Fbb,m,d,t)$-\svsat{} from \Cref{thm:pcp-svcsp} with $\delta=1$ to obtain a PCP verifier (recall \Cref{def:pcpp}) $\Pcal'$ for the satisfiability of $G$.

The query complexity, completeness, alphabet, and randomness of $\Pcal'$ follow from those of $\Pcal$ in \Cref{thm:pcp-svcsp}; and the soundness of $\Pcal'$ is the (unspecified) constant soundness parameter by setting $\delta=1$ in \Cref{thm:pcp-svcsp}.
In particular, 
$$
\text{the alphabet size is }|\Sigma|^t=|\Fbb|^{(d+1)\cdot t}=|\Fbb|^{O(dn\log k/k)}
$$
and
$$
\text{the randomness is }\log(k+|\Fbb|^m)+O\pbra{\log^{0.1}(k+|\Fbb|^m)+\log|\Fbb|}\text{ coins.}
$$
Then we apply \Cref{fct:pcpp} and obtain a 2CSP instance $\Lambda$ preserving the satisfiability of $G$ (and thus $\Gamma$) where
\begin{itemize}
\item the size of the alphabet of $\Lambda$ is 
$$
N=|\Fbb|^{O(dn\log k/k)},
$$
\item the number of variables in $\Lambda$ is at most
$$
K=2|\Fbb|^m+2^{\log(k+|\Fbb|^m)+O\pbra{\log^{0.1}(k+|\Fbb|^m)+\log|\Fbb|}}\cdot O(1)=(k+|\Fbb|^m)\cdot\poly\pbra{|\Fbb|,2^{\log^{0.1}(k+|\Fbb|^m)}},
$$
\item the number of constraints in $\Lambda$ is a constant multiple of the number of variables in $\Lambda$.
\end{itemize}

Finally we optimize the choice of $m,d,|\Fbb|$.
Assume $\log k$ is a perfect square and is sufficiently large and set
$$
|\Fbb|=2^{1000}\log k\cdot 2^{\sqrt{\log k}},
\quad
m=\sqrt{\log k},
\quad
d=2^{30}\sqrt{\log k}\cdot 2^{\sqrt{\log k}}.
$$
Then
$$
\binom{m+d}d\ge\pbra{\frac dm}^m=\pbra{2^{30}\cdot 2^{\sqrt{\log k}}}^{\sqrt{\log k}}\ge k,
$$
which is consistent with \Cref{eq:thm:main_1}.
We also have
\begin{align*}
12md=12\cdot 2^{30}\cdot \log k\cdot 2^{\sqrt{\log k}}\le|\Fbb|
\quad\text{and}
\\2^{101}m\log|\Fbb|\le2^{101}\sqrt{\log k}\pbra{\sqrt{\log k} + \log \log k + 1000} \le|\Fbb|,
\end{align*}
which is consistent with \Cref{eq:thm:main_2}.
Moreover, $\Lambda$ has alphabet size $N$ and the number of variables $K$ as follows.
$$N=\pbra{\log k\cdot 2^{\sqrt{\log k}}}^{O\pbra{\frac{n\cdot \log^{1.5} k\cdot 2^{\sqrt{\log k}}}{k}}} \quad\text{and}\quad K=k\cdot2^{O(\sqrt{\log k}\log\log k)}$$ 
where $N$ can be furthered simplified to
$$
N = \pbra{2^{n/k}}^{2^{O(\sqrt {\log k})}},
$$
Since $\Gamma$ has no $2^{o(n)}$-time algorithm by assumption, $\Lambda$ has the lower bound $f(K)\cdot N^{\lowerboundK}$-time for any computable function $f$ as claimed.

%In the end, we remark that the implicit constant soundness gap can be boosted to any constant by a constant number of randomness-efficient query repetitions~(see, e.g., \cite[Lemma 2.11]{BGH06}).

\end{proof}

The PCP statement \Cref{thm:new-pcp-restate} follows directly from the proof above.

\begin{proof}[Proof of \Cref{thm:new-pcp-restate}]
From any instance of {\sc 3SAT}, there is a linear-size reduction to an instance of {\sc 4-Regular 3-Coloring} \cite{Mar10}.
Therefore we only need to construct the desired PCP for {\sc 4-Regular 3-Coloring}.
This follows directly from the verifier $\Pcal'$ in the proof of \Cref{thm:main}.
In particular, we stick to the parameter choice there and obtain a PCP verifier with alphabet size
$$
2^{n\cdot 2^{O(\sqrt{\log k})}/k}
$$
and randomness
$$
\log k+O\pbra{\sqrt{\log k}\log \log k}.
$$
The implicit constant soundness of $\Pcal'$ can be boosted to $\frac12$ by a constant number of randomness-efficient query repetitions~(see e.g., \cite[Lemma 2.11]{BGH06}).
\end{proof}
\section{From {\sc 4-Regular 3-Coloring} to {\sf SVecCSP}}\label{sec:prepro}

The goal of this section is to reduce {\sc 4-Regular 3-Coloring}, which is known to be ETH-hard \Cref{thm:eth_3color}, to \svcsp{}.

\begin{theorem*}[\Cref{thm:reduction_svcsp} Restated]
    There is a reduction algorithm such that the following holds. Given as input an integer $6\le k\le n$ and an $n$-variable {\sc 4-Regular 3-Coloring} instance $\Gamma$, it produces an \svcsp{} instance $G=(V,E,\Sigma,\{\Pi_e\}_{e \in E})$ where:
    \begin{enumerate}[label=\textbf{(R\arabic*)}]
        \item\textsc{Variables and Constraints.} $|V|=O(k)$ and $|E|=O(k)$.
        \item\textsc{Runtime.} The reduction runs in time $\poly(n,2^{n\log k/k})$. 
        \item\textsc{Alphabet.} $\Sigma = \Fbb_4^{t}$ where $t=O\pbra{\frac{n\log k}k}$.
        \item\textsc{Satisfiability.} $G$ is satisfiable iff $\Gamma$ is satisfiable.
    \end{enumerate}
\end{theorem*}

The reduction starts by grouping vertices into supernodes, which take vector values.
Then the constraints between supernodes correspond to (possibly multiple) constraints in the original instance.
To make sure the new instance has small size, we need a grouping method (\Cref{lem:sparse_csp}) that produces as few supernodes and constraints as possible.
Then we make duplicates of variables and rearrangements of their coordinates to make sure parallel constraints are scattered properly (\Cref{prop:reduction_vcsp}).
Finally we make more duplicates to ensure that linear constraints form a matching and parallel constraints are applied on all coordinates (\Cref{prop:vcsp_to_svcsp}).

\subsection{From {\sc 4-Regular 3-Coloring} to {\sf VecCSP}}

The following \Cref{lem:sparse_csp} serves the purpose of the grouping method and can be seen as a parameterized version of the \emph{sparsification lemma} \cite{IP01, IPZ01} in classical computation complexity.

\begin{lemma}[\cite{Mar10, KMPS23}]\label{lem:sparse_csp}
    There is an algorithm $\mathcal A$ such that the following holds.
    $\Acal$ takes as input a 2CSP instance $G=(V,E,\Sigma,\{\Pi_e\}_{e \in E})$ and an integer $6\le k\le |V|$, outputs a 2CSP instance $G'=(V',E',\Sigma^t,\{\Pi'_e\}_{e \in E'})$ in time $\poly\pbra{|V|,|\Sigma|^t}$ where $|V'|=k$ and $t \le O\pbra{(|V|+|E|) \cdot \frac{\log k}k}$, such that $G$ is satisfiable iff $G'$ is satisfiable.
    In addition, the constraint graph of $G'$ is a 3-regular graph.
    
    In the actual construction, each $x\in V'$ corresponds to a subset $S(x)\subseteq V$ of size $t$ and takes values in $\Sigma^{S(x)}$ as assignments to  variables in $S(x)$. 
    For each $e=\cbra{x,y}\in E'$, the constraint $\Pi_e$ is the conjunction of the following: 
    \begin{enumerate}
        \item\label{itm:lem:sparse_csp_1} equality constraints on common variables of $S(x)$ and $S(y)$;
        \item\label{itm:lem:sparse_csp_2}  constraints across\footnote{The construction in \cite{KMPS23} ensures that each original constraint $(u,v) \in E$ is covered by some new constraint $(x,y) \in E'$ in the sense that $u \in S(x), v \in S(y)$. This means that we only need to check the cross constraints in $G'$ and omit the ones purely inside $S(x)$ or inside $S(y)$.} $S(x)$ and $S(y)$ in $G$. 
    \end{enumerate}
\end{lemma}
\begin{remark}
Note that since the constraint graph of $G'$ is 3-regular, the total number of constraints in $G'$ is linear in $|V'|$. In addition, $t$ only incurs an extra $\log k=k^{o(1)}$ blowup over the information theoretic limit $(|V|+|E|)/k$.
This is crucial for us to get almost tight hardness. Indeed, naive approaches (e.g., the reduction in \cite{GLRSW24}) will incur a polynomial loss.
We also remark that it is an open problem whether the extra $\log k$ can be further removed, which, if true, implies a precise parameterized analog of the sparsification lemma and has many applications. See \cite{Mar10} for detailed discussions.
\end{remark}

Given \Cref{lem:sparse_csp}, we first obtain a vector-valued CSP (\vcsp{}) instance (\Cref{prop:reduction_vcsp}), which we will soon convert into \svcsp{} (\Cref{prop:vcsp_to_svcsp}).

\begin{definition}[\vcsp{}, \cite{GLRSW24}]\label{def:vcsp}
A CSP instance $G = (V,E,\Sigma,\{\Pi_{e}\}_{e\in E})$ is a \vcsp{} if the following properties hold.
\begin{itemize}
    \item $\Sigma=\Fbb^t$ is a $t$-dimensional vector space over a finite field $\Fbb$ with $\charF(\Fbb)=2$.
    \item For each constraint 
    $e=\{u,v\}\in E$ where $u = (u_1,\dots,u_t)$ 
       and $v = (v_1,\ldots,v_t)$ are two variables in $V$, the constraint validity function $\Pi_e$ is classified as one of the following cases:
    \begin{itemize}
        \item \textsc{Linear.} There exists a  matrix $M_e\in \Fbb^{t\times t}$ such that 
        $$
        \Pi_e(u,v) = \indicator_{u = M_ev}.
        $$
        \item \textsc{Parallel.} There exists a \emph{sub-constraint}  $\Pi_e^{sub}: \Fbb\times \Fbb\to \bin$ and a subset of coordinates $Q_e\subseteq[t]$ such that $\Pi_e$ checks $\Pi_e^{sub}$ for every coordinate in $Q_e$, i.e., 
        $$
        \Pi_e(u,v) = \bigwedge_{i\in Q_e} \Pi_e^{sub}(u_i,v_i).
        $$
    \end{itemize}
    \item Each variable is related to at most one parallel constraint.
\end{itemize}
\end{definition}

Note that \svcsp{}, the special case of \vcsp{} we use, additionally enforces linear constraints to be a matching and enforces parallel constraints to operate on all coordinates (i.e., $Q_e=[t]$).

\begin{proposition}[\vcsp{} Intermediate Instance]\label{prop:reduction_vcsp}
    There is a reduction algorithm such that the following holds. Given as input an integer $6\le k\le n$ and an $n$-variable {\sc 4-regular 3-Coloring} instance $\Gamma$, it produces an \vcsp{} instance $G=(V,E,\Sigma,\{\Pi_e\}_{e \in E})$ where:
    \begin{itemize}
        \item\textsc{Variables and Constraints.} $|V|=O(k)$ and $|E|=O(k)$.
        \item\textsc{Runtime.} The reduction runs in time $\poly(n,2^{n\log k/k})$.
        \item\textsc{Alphabet.} $\Sigma = \Fbb_4^{t}$ where $t=O\pbra{\frac{n\log k}k}$.
        \item\textsc{Satisfiability.} $G$ is satisfiable iff $\Gamma$ is satisfiable.
    \end{itemize}
\end{proposition}
\begin{proof}
    We first plug the {\sc 4-Regular 3-Coloring} instance $\Gamma$ from \Cref{thm:eth_3color} into \Cref{lem:sparse_csp} and obtain a 2CSP instance $G'=(V',E',\Fbb_4^t, \{\Pi'_e\}_{e \in E'})$. 
    By the construction in \Cref{lem:sparse_csp}, each $x\in V'$ corresponds to a set $S(x)$ of $t$ variables in $\Gamma$. 
    We fix an arbitrary order in $S(x)$ and use $x[i]$ to denote the $i$-th variable in $S(x)$. 
    From the fact that $\Gamma$ is 4-regular and the construction in \Cref{lem:sparse_csp}, the produced 2CSP instance $G'$ has the following properties:
    \begin{itemize}
        \item for each $\cbra{x,y}\in E'$ and $i\in[t]$, there are at most five sub-constraints between $x[i]$ and $\{y[j] : j \in [t]\}$: at most one equality check (\Cref{itm:lem:sparse_csp_1} of \Cref{lem:sparse_csp}) and at most four coloring checks (\Cref{itm:lem:sparse_csp_2} of \Cref{lem:sparse_csp}); 
        \item the constraint graph of $G'$ is $3$-regular;
        \item $|V'|=k$, $|E'|=O(k)$, $t=O\pbra{\frac{n\log k}k}$, and the runtime is $\poly(n,k,4^t)=\poly(n,2^{n\log k/k})$;
        \item $G'$ is satisfiable iff $\Gamma$ is satisfiable.
    \end{itemize}

    \begin{figure}[ht]
        \centering
        \includegraphics[height=13cm]{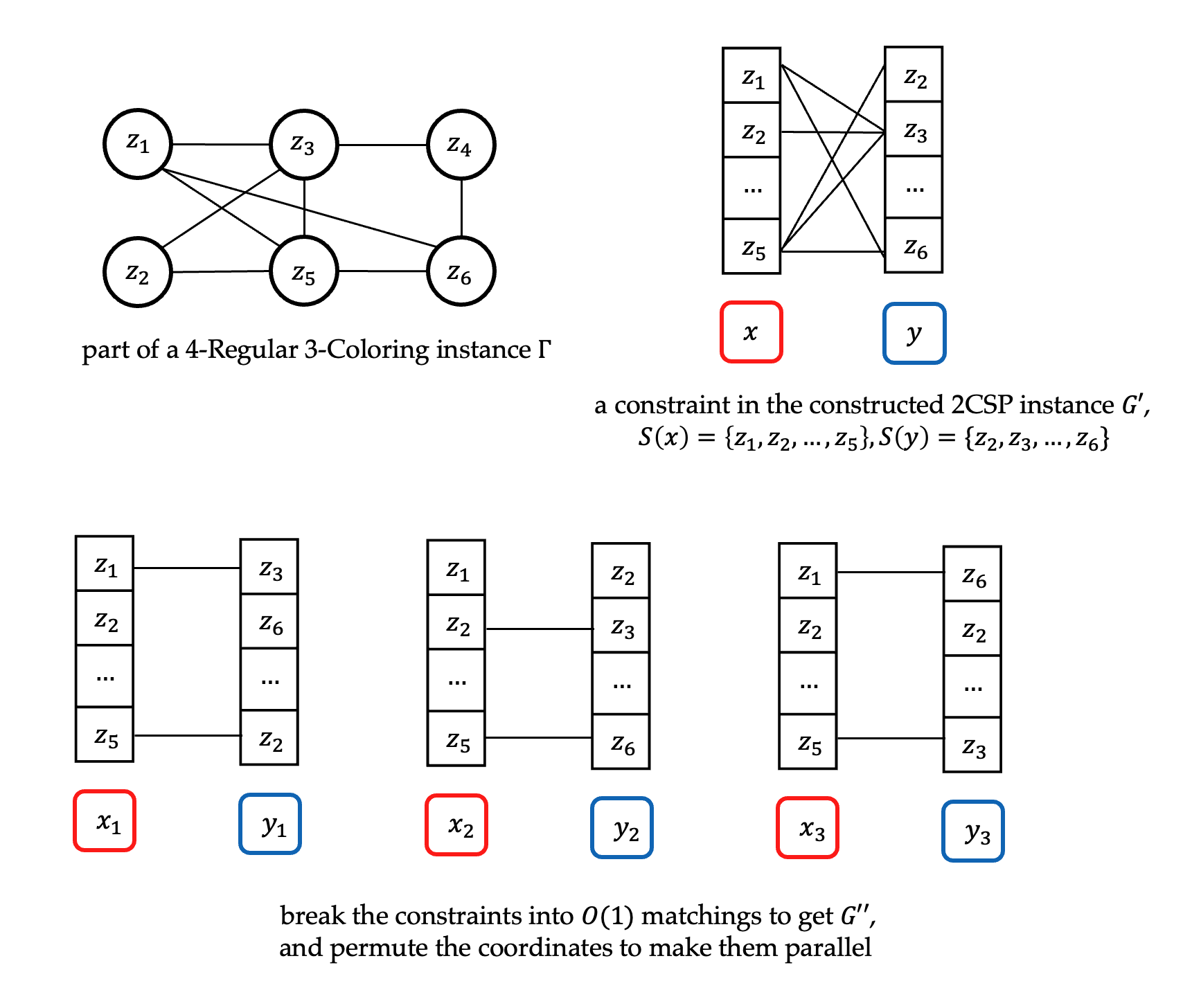}
        \caption{An example of $G'$ and $G''$ and the permutation to parallelize sub-constraints.}
        \label{fig:svcsp-1}
    \end{figure}

    Thus, we duplicate each $x \in V'$ into constant many copies, and distribute the sub-constraints in $G'$ onto different copies.
    This produces another 2CSP instance $G''=(V'',E'',\Fbb_4^t,\{\Pi_e''\}_{e\in E''})$ where
    \begin{itemize}
        \item for each $e=\cbra{u,v}\in E''$, the constraint $\Pi''_e$ has exactly one type of sub-constraint (i.e., equality or coloring), which forms a partial (non-parallel) matching across the coordinates (i.e., $S(u),S(v)$) of $u,v$;

        (Note that the consistency checks among duplicates will be added later.)
        \item each variable in $V''$ is related to exactly one constraint;
        \item $|V''|=O(|V'|)=O(k)$ and $|E''|=|E'|=O(k)$. 
    \end{itemize}
    The above procedure is efficient: we only need to perform matching decompositions for each $\cbra{x,y}\in E'$ separately for equality checks and coloring checks.

    Before adding the consistency checks among duplicates, we first permute coordinates of each variable in $V''$ to parallelize the partial matchings. This is possible since each variable in $V''$ is related to exactly one constraint in $G''$.
    For a fixed $x\in V'$, let $x_1,\ldots,x_m\in V''$ be the duplicates of $x \in V'$. 
    After the permutation, we add linear constraints between $x_i$ and $x_{i+1}$ for $1 \le i < m$ to check whether they are consistent (i.e., the correct permuted copies of each other).
    See \Cref{fig:svcsp-1} for a streamlined presentation.

    The construction of the \vcsp{} instance $G=(V,E,\Sigma=\Fbb_4^t,\{\Pi_e\}_{e\in E})$ is completed after the permutation and adding the consistency checks among duplicated.
    The satisfiability is naturally preserved, $|V|=|V''|=O(k)$, and $|E|\le|E''|+|V''|=O(k)$.
    In terms of \Cref{def:vcsp}, the consistency checks are linear constraints and the constraints in $G''$ after permutation are parallel constraints.
\end{proof}

\subsection{From {\sf VecCSP} to {\sf SVecCSP}}

Given \Cref{prop:reduction_vcsp}, \Cref{thm:reduction_svcsp} follows directly by the following general reduction from \vcsp{} to \svcsp{}.

\begin{proposition}[\vcsp{} to \svcsp{}]\label{prop:vcsp_to_svcsp}
    There is a reduction algorithm such that the following holds. Given as input a \vcsp{} instance $G=(V,E,\Fbb^t,\{\Pi_e\}_{e\in E})$, it produces an \svcsp{} instance $G'=(V',E',\Fbb^t,\{\Pi'_e\}_{e \in E'})$ where:
    \begin{itemize}
        \item\textsc{Variables and Constraints.} $|V'|=O(|V|+|E|)$ and $|E'|=O(|V|+|E|)$.
        \item\textsc{Runtime.} The reduction runs in time $\poly(|V|,|E|,|\Fbb|^t)$.
        \item\textsc{Satisfiability.} $G'$ is satisfiable iff $G$ is satisfiable.
    \end{itemize}
\end{proposition}
\begin{proof}
    Now that we get a \vcsp{} instance $G$, we show how to modify it to to obtain an \svcsp{} instance $G'$ that satisfies properties \Cref{itm:def:svcsp_2} and \Cref{itm:def:svcsp_1}. The construction consists of two steps and see \Cref{fig:svcsp-2} for a streamlined presentation.
    \begin{itemize}
        \item First, from $G$, get another \vcsp{} instance $\hat G$ which satisfies \Cref{itm:def:svcsp_2}.
        \item Next, based on $\hat G$, build a \svcsp{} instance $G'$ which in addition satisfies \Cref{itm:def:svcsp_1}.
    \end{itemize}

    To satisfy \Cref{itm:def:svcsp_2}, we split each variable $x$ in $V$ into a parallel variable $x^\parallelc$ and a linear variable $x^\linearc$ in $\hat G$ for parallel and linear constraints separately.
    Then we construct the constraints $\hat E$ in $\hat G$.
    \begin{itemize}
    \item For each linear constraint $e=\{x,y\}\in E$, we add the same linear constraint on $\{x^\linearc,y^\linearc\}$ in $\hat E$. 
    \item For each parallel constraint $e=\{x,y\}\in E$ with sub-constraint $\Pi_e^{sub}$ and subset of coordinates $Q_e$, we add a parallel constraint on $\{x^\parallelc,y^\parallelc\}$ in $\hat E$, which has $\Pi_e^{sub}$ applied on \emph{all} coordinates.
    \item We need to additionally check in $\hat G$ the partial equality between $x^\parallelc$ and $x^\linearc$ \emph{only} on the subset of coordinates $Q_e$. 
    Since each variable $x$ is related to at most one parallel constraint as guaranteed in \Cref{def:vcsp}, this additional check is well-defined. 
    
    This check can be written as $M_{Q_e} x^\parallelc = M_{Q_e} x^\linearc$, where $M_{Q_e}$ is a matrix which projects on coordinates inside $Q_e$ and zeroing out coordinates outside $Q_e$. 
    To have the matrix only on one side, we introduce an additional variable $x^{\sf a}$ in $\hat G$, and add two linear constraints $x^{\sf a} = M_{Q_e} x^\parallelc $ and $x^{\sf a} = M_{Q_e} x^\linearc$ to $\hat E$.
    \end{itemize}
We first show that the above construction preserves the satisfiability as follows.
    \begin{itemize}
    \item Given a solution $\sigma$ for the original $G$, assign $\sigma(x)$ to $x^\linearc$ and $x^{\sf a}$, which satisfies all linear constraints in $\hat E$.
    Then assign $\sigma(x)$ to $x^\parallelc$ on the subset $Q_e$ of coordinates, which satisfies all the partial equality checks in $\hat E$.
    Finally, assign arbitrary solution\footnote{Technically it is possible that $\Pi_e^{sub}$ is not satisfiable. If $Q_e=\emptyset$, then we can simply replace $\Pi_e^{sub}$ by any satisfiable sub-constraint. If otherwise $Q_e\ne\emptyset$, then the original $G$ is not satisfiable and $\hat G$ is also not satisfiable. Therefore the construction still works.} of $\Pi_e^{sub}$ to $x^\parallelc$ on coordinates outside $Q_e$, which satisfies all the parallel constraints in $\hat E$. 
    \item Given a solution $\sigma'$ of $\hat G$, assign $\sigma'(x^\linearc)$ to every $x$ in $G$, which satisfies all linear constraints in $E$.
    Since $\sigma'(x^\parallelc)$ satisfies the parallel constraint in $G'$ for all coordinates and the partial equality check guarantees consistency between $\sigma'(x^\linearc)$ and $\sigma'(x^\parallelc)$ on the coordinates $Q_e$, the corresponding parallel constraints in $E$ are satisfied as well.
    \end{itemize}
Moreover, the variable set $\hat V\subseteq\bigcup_{x\in V}\cbra{x^\parallelc,x^\linearc,x^{\sf a}}$ has size $|\hat V|=O(|V|)$ and the constraint set $\hat E$ has size $|\hat E|\le|E|+2\cdot|V|=O(|V|+|E|)$.

    Now we construct $G'$ from $\hat G$ to satisfy \Cref{itm:def:svcsp_1}. 
    The final variable set of $G'$ will be $V' = X \dot \cup Y$, which is constructed along with the constraint set $E'=E'_\linearc\dot\cup E'_\parallelc$ as follows.
    \begin{itemize}
    \item Initialize $X,Y$ as disjoint copies of $\hat V$ and initialize $E'_\linearc=E'_\parallelc=\emptyset$. For a variable $u \in \hat V$, we denote as $x_u, y_u$ its $X$-copy and $Y$-copy in $V'$, respectively.
    \item For each $u \in \hat V$, we add an equality, which is a linear constraint with the identity matrix, in $E'_\linearc$ between $x_u$ and $y_u$.

    Note that this is consistent with \Cref{itm:def:svcsp_1}.
    \item Then for each parallel constraint $e=\cbra{x^\parallelc,y^\parallelc}\in\hat E$, we add the same constraint in $E'_\parallelc$ on the $X$-copies of $x^\parallelc$ and $y^\parallelc$.
    \item Finally for each linear constraint\footnote{In particular, $e=\cbra{u,v}$ can be $\cbra{x^\linearc,y^\linearc}$ or $\cbra{x^\linearc,x^{\sf a}}$ or $\cbra{x^\parallelc,x^{\sf a}}$.} $e=\{u,v\}\in\hat E$ that checks $ u = M_e v$, we add new variables $x_e$ to $X$ and $y_e$ to $Y$.
    Then we impose a linear constraint $y_e = M_e x_e$ between them in $E'_\linearc$, which is consistent with \Cref{itm:def:svcsp_1}.
    We further add two equality constraints, which are identified as parallel constraints in $E'_\parallelc$, between $x_e$ and $x_u$, as well as between $y_e$ and $x_v$. 
    \end{itemize}

    The construction of $G'$ preserves the satisfiability of $\hat G$ as all the duplicates in $G'$ of variables in $\hat G$ are connected by identity constraints.
    Moreover, $|X|=|Y|\le|\hat V|+|\hat E|=O(|V|+|E|)$ and $|E'|\le|\hat V|+|\hat E|+3\cdot|\hat E|=O(|V|+|E|)$ as desired. 
\end{proof}

\begin{figure}[ht]
    \centering
    \includegraphics[width = \textwidth]{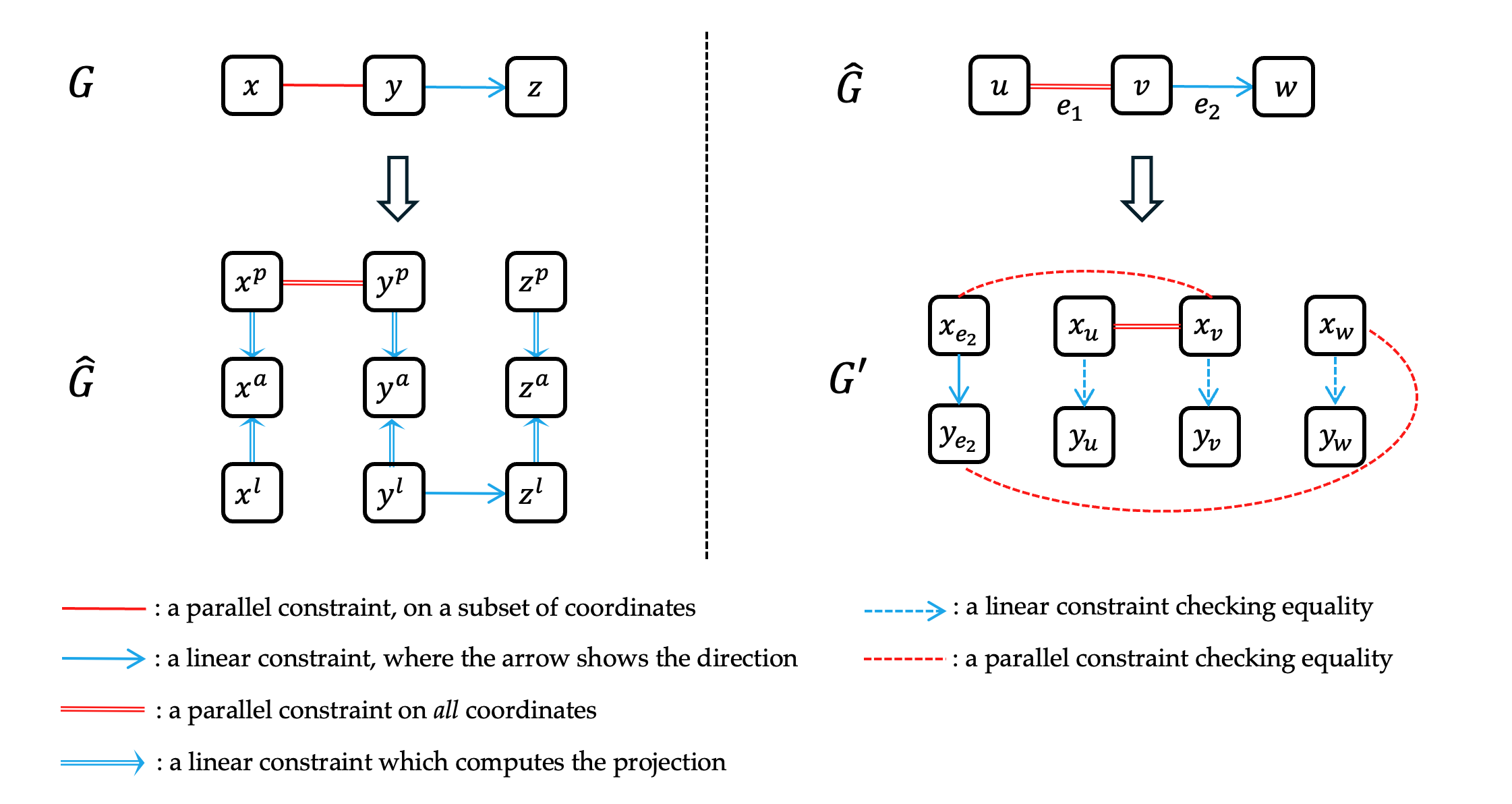}
    \caption{An illustration of the reduction from $G$ to $\hat G$, and from $\hat G$ to $G'$.}
    \label{fig:svcsp-2}
\end{figure}
\section{PCPP for Multi-Test Problem}
\label{sec:pcpp-probttt}

In this section we prove \Cref{thm:pcpp-doubletest} for \doubletest{} and \Cref{thm:pcpp-singletest} for \singletest{}. For convenience, we define the following multi-test problem (\probttt{}) which generalizes \doubletest{} and \singletest{} in a straightforward fashion.

\begin{definition}[\probttt{}]\label{def:probttt}
    Assume $\charF(\Fbb)=2$.
    $(\Fbb, m, d, t,u)$-\probttt{} is a pair language over $\Sigma_\pairx = \bin, \Sigma_\pairy = \Fbb^t$ consisting of all words in the form of $w=(C,T_1\circ \cdots \circ T_u)$, where
    \begin{itemize}
        \item $C$ is a Boolean circuit with $u\cdot|\Fbb|^m \log |\Fbb|$ input bits and
        $T_1,\dots, T_u\in (\Fbb^t)^{|\Fbb|^m}$ are codewords of $\prm$;
        \item if we view $\Fbb$ as $\bin^{\log|\Fbb|}$, $T_1\circ\cdots\circ T_u$ parallel satisfies $C$.
    \end{itemize}
    We define the input length $|w|$ to be the size of $C$ plus $u\cdot|\Fbb|^m \log |\Fbb|$.
\end{definition}

By setting $u=1$ or $u=2$, we immediately obtain $(\Fbb,m,d,t)$-\singletest{} or $(\Fbb,m,d,t)$-\doubletest{}.
Hence \Cref{thm:pcpp-singletest} and \Cref{thm:pcpp-doubletest} follows directly from the following result for \probttt{}.

\begin{theorem}[PCPP for \probttt{}]\label{thm:pcpp-main} 
Assume $\charF(\Fbb)=2$ and $|\Fbb|\ge\max\cbra{6md,2^{100}m\log|\Fbb|}$.
Assume $u\le2^{50}$ is a positive integer.
Then for any $\delta \in[0,1]$, $(\Fbb, d, m, t,u)$-\probttt{} has 
$$
\text{an $\pbra{O(1), \log|w| + O\pbra{\log^{0.1}|w|+\log|\Fbb|}, \delta,\Omega(\delta), \Sigma^{t}}$-PCPP verifer $\pcpptabletest$,}
$$
where $|w|$ is the input length of $(\Fbb,d,m,t,u)$-\probttt{} and $\Sigma=\Fbb^{d+1}$.
\end{theorem}

\paragraph{Proof sketch} Let $(C,T_1\circ \cdots \circ T_u)$ be an input for $(\Fbb,d,m,t,u)$-\probttt{}. Our goal is to check whether $T_1\circ \cdots \circ T_u$ is $\delta$-close to some $T_1^*\circ \cdots \circ T_u^*\in$ \probttt{}$(C)$, i.e., the restriction of the pair language \probttt{} on $C$. 
In other words, the following two conditions hold:
\begin{enumerate}[label=(C\arabic*)]
    \item \label{itm:check:c1} for each $j\in[u]$, $T_j^*\in\Fbb^{|\Fbb|^m}$ is the truth table of a polynomial of parallel degree $d$;
    \item \label{itm:check:c2} $T_1^*\circ \cdots \circ T_u^*$, viewed as a word in $\bin^{u|\Fbb|^m\log|\Fbb|}$, parallel satisfies the Boolean circuit $C$.
\end{enumerate}

To guarantee \Cref{itm:check:c1}, we use the PCPP verifier $\pcppldt$ from the codeword testing of $\prm$ (see \Cref{thm:ldt-maintext}).
Given $T_1^*\circ\cdots\circ T_u^*$ satisfying \Cref{itm:check:c1}, we aim to test that it also satisfies \Cref{itm:check:c2}.
That is, for each fixed $i\in[t]$, $T_1^*[i]\circ\cdots\circ T_u^*[i]$ satisfies the Boolean circuit $C$.

To this end, we will use the PCPP verifier $\pcppckt$ of \pcircuit{} (see \Cref{thm:pcpp-cktsat}).
This alone, however, is not sufficient as $\pcppckt$ cannot rule out the case that changing $o(1)$ fraction of entries in $T_1^*[i]\circ\cdots\circ T_u^*[i]$ satisfying the circuit $C$.
To fix this issue, we have to exploit the fact that each $T_j^*[i]$ is supposed to be the truth table of a degree-$d$ polynomial, which, by Schwartz-Zippel lemma and the fact that $u$ is a constant, forbids such attacks.
As a result, we need to incorporate the codeword testing circuit (see \Cref{thm:pldt-circuit}) to enforce the low degree condition.

Unfortunately, this still does not work due to a subtle alphabet mismatch: the codeword testing works over $\Fbb$ but \Cref{itm:check:c2} needs to flat $\Fbb$ as $\bin^{\log|\Fbb|}$.
Therefore, the distance guaranteed by the low degree condition can be dilated by a worst-case factor of $\log|\Fbb|=\omega(1)$ after converting $\Fbb$ to $\bin^{\log|\Fbb|}$, for which reason the mentioned attack can still be carried out.
To address this issue, we employ a standard approach to lift the conversion of $\Fbb$ via error correcting codes~\cite{BGH06, ALM+98, Din07}.
More formally, after flattening $\Fbb$ as $\bin^{\log|\Fbb|}$, we take it through an error correcting code with constant rate and distance, which produces a codeword in $\bin^{O(\log|\Fbb|)}$ and, more importantly, has constant relative distance against other codewords.

In summary, to handle \Cref{itm:check:c2}, we need to use $\pcppckt$ to check the validity of the combination of (1) the original circuit $C$, (2) the codeword testing procedure, and (3) the error correcting lifting.
This is parallel for each coordinate and is presented in \Cref{sec:single_checking}.
Then in \Cref{sec:tabletest_everything}, we put together the argument for \Cref{itm:check:c1} and prove \Cref{thm:pcpp-main}.

\begin{remark}
One may wonder the necessity of using a separate codeword testing for \Cref{itm:check:c1}, as we anyway need to use it for \Cref{itm:check:c2}.
The difference lies in the proximity: the former uses \Cref{thm:ldt-maintext} which guarantees the parallel proximity (i.e., including all the dimensions $[t]$), whereas the latter uses \Cref{thm:pldt-circuit} which only implies coordinate-wise proximity (i.e., individually for each coordinate $i\in[t]$).

Without the former, we can only get an $\Omega(t\cdot\delta)$ final proximity via a union bound. 
Upgrading the latter needs to generalize the construction of $\pcppckt$ to the parallel setting, which arguably requires more work. Hence we stick to our current presentation for simplicity.
\end{remark}

\subsection{Single-Coordinate Checking Circuit}\label{sec:single_checking}

As sketched above, we construct another Boolean circuit $C'$ for \Cref{itm:check:c2}, which augments the original circuit $C$ with an alphabet lifting via error correcting code and a low-degree check.
Later, this single-coordinate checking circuit will be applied in parallel for every coordinate.

\paragraph{Lifted Flattening of $\Fbb$}
We will use the following standard error correcting code to achieve such an alphabet lifting.

\begin{proposition}[Binary ECC with Constant Rate and Distance~\cite{Justesen1972ClassOC}]
\label{prop:ecc}
    For every $n\ge 1$, there exists an efficiently computable error correcting code with the encoding map ${\rm ECC}_n: \bin^n\to \bin^{7n}$ of relative distance $\delta({\rm ECC_n})\ge 0.01$. 
\end{proposition}

Let $\Enc: \Fbb\to \{0,1\}^{7\log |\Fbb|}$ to be ${\rm ECC}_{\log |\Fbb|}$ composed with the natural flattening of $\Fbb$ into $\bin^{\log |\Fbb|}$.
Let $\Dec$ be the corresponding decoding function.

Later putting back into the construction of $\pcpptabletest$, we will apply $\Enc$ to every entry of $T_1[i]\circ\cdots\circ T_u[i]\in\Fbb^{u|\Fbb|^m}$ for every coordinate $i\in[t]$.
For convenience, we define the parallel encoding map $\Enc^\odot: \Fbb^{|\Fbb|^m}\to \bin^{7|\Fbb|^m\log |\Fbb|}$ as
\begin{equation}\label{eq:lift_ecc_parallel}
\Enc^\odot(z)(k,\cdot)=\Enc(z_k)
\quad\text{for $z\in\Fbb^{|\Fbb|^m}$ and $k\in[|\Fbb|^m]$,}
\end{equation}
where we view $\Enc^\odot(z)\in\bin^{7|\Fbb|\log|\Fbb|}$ as a function $[|\Fbb|^m]\times[7\log|\Fbb|]\to\bin$ and the index $k$ points to the lifted flattening of the $k$-th entry of $z$.
Given the definition of $\Enc^\odot$, the lifted flattening of $T_1[i]\circ\cdots\circ T_u[i]$ can be simply represented as $\Enc^\odot(T_1[i])\circ\cdots\circ\Enc^\odot(T_u[i])$.

\paragraph{The Construction of $C'$}
Given the lifted flattening of alphabet, we now present the construction.
The circuit $C'$ takes as input $y_1\circ\cdots\circ y_u$ where each $y_j\in\bin^{7|\Fbb|^m\log|\Fbb|}$ is supposed to be $\Enc^\odot(\hat y_j)$ for some $\hat y_j\in\Fbb^{|\Fbb|^m}$.
Note that $\hat y_j$ will be $T_j[i]$, rolling over all coordinates $i\in[t]$.
The circuit $C'$ check the following three things in order:
\begin{enumerate}[label=(S\arabic*)]
    \item\label{itm:c'_1} Check if each $y_j$ is a codeword of $\Enc^\odot$ and, if so, compute each $\hat y_j$ and view it as a binary string via the standard flattening of $\Fbb$. 
    
    This requires to decode a total number of $u\cdot|\Fbb|^m$ words in $\bin^{7\log|\Fbb|}$, which can be has circuit size $u|\Fbb|^m\cdot\polylog|\Fbb|\le|\Fbb|^m\poly|\Fbb|$.
    \item\label{itm:c'_2} Check if each $\hat{y_i}$ satisfies $C_{\ldt}$ from \Cref{thm:pldt-circuit}, i.e., $\hat y_i$ is the truth table of a degree-$d$ polynomial.

    This has size $|\Fbb|^m\poly|\Fbb|$ by \Cref{thm:pldt-circuit}.
    \item\label{itm:c'_3} Check whether $\hat{y_1}\circ \cdots \circ \hat{y_u}$ satisfies the original circuit $C$. 

    This has size precisely the size of $C$.
\end{enumerate}

Now we list properties of $C'$.

\begin{fact}[Satisfiability]\label{fct:c'_sat}
If $\Enc^{\odot}(\hat{y_1})\circ \cdots \circ \Enc^{\odot}(\hat{y_u})$ satisfies the circuit $C'$, then $\hat{y_1}\circ \cdots \circ \hat{y_u}$ satisfies the circuit $C$.
\end{fact}

\begin{fact}[Size]\label{fct:c'_size}
Recall from \Cref{def:probttt} that $|w|$ is the input length of $(\Fbb,m,d,t)$-\probttt{}, which equals the size of $C$ plus $u\cdot|\Fbb|^m\log|\Fbb|$.
The size of the circuit $C'$ is at most $|w|\cdot\poly|\Fbb|$ and the input of $C'$ has length $O(|w|)$.
\end{fact}

\begin{claim}[Distance]\label{clm:c'_distance}
If $y_1\circ\cdots\circ y_u$ passes \Cref{itm:c'_1,itm:c'_2} but not \Cref{itm:c'_3}, then it is $2^{-60}$-far from solutions of $C'$.
\end{claim}
\begin{proof}
Let $z_1\circ\cdots\circ z_u$ be an arbitrary solution of $C'$, which means it passes \Cref{itm:c'_1,itm:c'_2,itm:c'_3}.
Let $\hat z_1,\ldots,\hat z_u$ be the decoding outcome from \Cref{itm:c'_1}.
Then there exists $j\in[u]$ such that $\hat y_j\ne\hat z_j$.
Since they both pass \Cref{itm:c'_2}, $\hat y_j$ and $\hat z_j$, viewed as an element from $\Fbb^{|\Fbb|^m}$, correspond to truth tables of distinct degree-$d$ polynomials.
Hence $\Delta(\hat y_j,\hat z_j)\ge1-\frac d{|\Fbb|}>\frac12$ by Schwartz-Zippel lemma and our assumption on $|\Fbb|$.
Recall from \Cref{eq:lift_ecc_parallel} that $\Enc^\odot$ applies $\Enc$ entrywise.
Since $y_j=\Enc^\odot(\hat y_j)$ and $z_j=\Enc^\odot(\hat z_j)$, we now have $\Delta(y_j,z_j)>0.01\cdot\frac12$ by \Cref{prop:ecc}.
Since $u\le2^{50}$, we have
$$
\Delta(y_1\circ\cdots\circ y_u,z_1\circ\cdots\circ z_u)\ge\frac{\Delta(y_j,z_j)}u>2^{-60}
$$
as desired.
\end{proof}

\subsection{Combining Codeword Testing}\label{sec:tabletest_everything}

Now we are ready to combine the codeword testing for \Cref{itm:check:c1} and the single-coordinate checking circuit $C'$ for \Cref{itm:check:c2} to prove \Cref{thm:pcpp-main}.

\begin{proof}[Proof of \Cref{thm:pcpp-main}]
The auxiliary proof for $\pcpptabletest$ consists of two parts. 
    \begin{itemize}
        \item The first part is $\pi_{\ldt,1}, \dots, \pi_{\ldt,u}$. Each $\pi_{\ldt,i}$ has alphabet $\Sigma^t$ and is constructed by \Cref{thm:ldt-maintext}, which is supposed to be the auxillary proof for codeword testing of $\prm$ on $T_i$. 
        
        \item The second part is denoted by $\pi_\ckt$, which has the alphabet $\bin^t$ naturally embedded into $\Sigma^t$. 
        For each coordinate $i\in [t]$, $\pi_\ckt[i]$ is constructed by \Cref{thm:pcpp-cktsat} for $\pcppckt$ to check if $\Enc^{\odot}(T_1[i])\circ \cdots \circ \Enc^{\odot}( T_u[i])$ satisfies the circuit $C'$.
    \end{itemize}

\paragraph{Testing Procedure of $\pcpptabletest$}
Now we describe the testing procedure. $\pcpptabletest$ executes one of the following two tests with equal probability.
    \begin{itemize}
        \item For each $i\in[u]$, $\pcpptabletest$ invokes $\pcppldt$ to run the codeword testing for $\prm$ on $T_i\circ \pi_{\ldt,i}$. This checks whether $T_i\in\eccim(\prm)$.
        \item $\pcpptabletest$ parallel simulates $\pcppckt$ to test if $\Enc^\odot(T_1[i])\circ \cdots \circ \Enc^\odot(T_u[i])\circ \pi_\ckt[i]$ satisfies $C'$ for all coordinates $i\in[t]$.
        
        In detail, for each coordinate $i\in[t]$, $\pcpptabletest$ tosses random coins as $\pcppckt$ does,      
        and probes entries of $\pi_\ckt[i]$ if needed.
        Whenever $\pcppckt$ needs to probe some bit of $\Enc^\odot(T_j[i])$, $\pcpptabletest$ queries the corresponding entry (i.e., the index $k$ in \Cref{eq:lift_ecc_parallel}) of $T_j[i]$, performs the lifted flattening of $\Fbb$ for that entry, and obtains the desired bit.

        We emphasize that the randomness used to simulate $\pcppckt$ is the same for all coordinates. Therefore, the queries by $\pcppckt$ are simulated in parallel for all coordinates $i\in[t]$, as the query locations are uniquely determined by the randomness.
    \end{itemize}

\paragraph{Parameters of $\pcpptabletest$} 
 Since $u$ is a constant and both $\pcppldt,\pcppckt$ have constant queries (see \Cref{thm:ldt-maintext} and \Cref{thm:pcpp-cktsat}), $\pcpptabletest$ makes constant queries.

 Recall that the input length $|w|$ equals the size of $C$ plus $u|\Fbb|^m\log|\Fbb|$.
 By \Cref{thm:ldt-maintext}, the first part tosses 
 $$
 m\log|\Fbb|+O(\log\log|\Fbb|+\log m)\le\log|w|+O(\log m)\le\log|w|+O(\log|\Fbb|)
 $$ 
 coins, where we used the assumption on $|\Fbb|$.
 By \Cref{thm:pcpp-cktsat} and \Cref{fct:c'_size}, the second part tosses 
 $$
 \log(|w|\cdot\poly|\Fbb|) + O\pbra{\log^{0.1}(|w|\cdot\poly|\Fbb|)}\le\log|w|+O\pbra{\log^{0.1}|w|+\log|\Fbb|}
 $$ 
 coins. Since we only execute one of them, the number of random coins is the maximum of the above two as desired.

  \paragraph{Completeness and Soundness} 
  The completeness is straightforward by the completeness of $\pcppldt$ (see \Cref{thm:ldt-maintext}) and $\pcppckt$ (see \Cref{thm:pcpp-cktsat}) and the construction of $C'$ (see \Cref{sec:single_checking}). 
  We focus on the soundness analysis:
  assuming that $T_1\circ\cdots\circ T_u$ is $\delta$-far from \probttt{}$(C)$ (i.e., we do not have \Cref{itm:check:c1,itm:check:c2}), $\pcpptabletest$ rejects with probability $\Omega(\delta)$.
  By modifying the hidden constant in $\Omega(\cdot)$ and noticing that $\delta$-far implies $\delta'$-far for any $\delta'\le\delta$, we additionally assume $\delta\le2^{-100}$.
  
  Assume towards contradiction that the above soundness statement is false. We first show that each $T_j$ is close to being parallel degree-$d$.
  This comes from the following \Cref{fct:tabletest_1}, which can be deduced directly from \Cref{thm:ldt-maintext}. 

    \begin{fact}\label{fct:tabletest_1}
        If for some $j\in[u]$, $T_j$ is $\delta$-far from $\eccim(\prm)$, then $\pcpptabletest$ rejects with probability $\Omega(\delta)$.
    \end{fact}

    Now we assume each $T_j$ is $\delta$-close to some $T_j^*\in\eccim(\prm)$, which corresponds to \Cref{itm:check:c1}.
    Next, we show $T_1^*\circ\cdots\circ T_u^*$ parallel satisfies $C$, which corresponds to \Cref{itm:check:c2}.
    By \Cref{fct:c'_sat}, it suffices to show for each coordinate $i\in[t]$ that $\Enc^\odot(T_1^*[i])\circ\cdots\circ\Enc^\odot(T_u^*[i])$ satisfies $C'$, which is precisely the following \Cref{clm:tabletest_2}.

    \begin{claim}\label{clm:tabletest_2} 
    For each coordinate $i\in[t]$, $\Enc^\odot(T_1^*[i])\circ\cdots\circ\Enc^\odot(T_u^*[i])$ satisfies \Cref{itm:c'_1,itm:c'_2,itm:c'_3}.
    \end{claim}
    
    Given \Cref{clm:tabletest_2}, we arrive at a contradiction and complete the soundness analysis.
\end{proof}

Finally we prove \Cref{clm:tabletest_2}.
\begin{proof}[Proof of \Cref{clm:tabletest_2}]
    Note that \Cref{itm:c'_1} is already satisfied.
    Since each $T_j^*\in\eccim(\prm)$, each $T_j^*[i]$ is the truth table of a degree-$d$ polynomial, which means \Cref{itm:c'_2} is also satisfied.
    If \Cref{itm:c'_3} is false for some $i\in[t]$, then by \Cref{clm:c'_distance}, $\Enc^\odot(T_1^*[i])\circ\cdots\circ\Enc^\odot(T_u^*[i])$ is $2^{-60}$-far from any solution of $C'$.
    Observe that
    \begin{align*}
        \Delta(\Enc^\odot(T_1[i])\circ\cdots\circ T_u[i],\Enc^\odot(T_1^*[i])\circ\cdots\circ T_u^*[i])
        &\le\Delta(\Enc^\odot(T_1)\circ\cdots\circ T_u,\Enc^\odot(T_1^*)\circ\cdots\circ T_u^*)\\
        &\le\max_{j\in[u]}\Delta(\Enc^\odot(T_j),\Enc^\odot(T_j^*))
        \tag{since $\Delta$ is relative distance}\\
        &\le\delta.
        \tag{by the choice of $T_j^*$}
    \end{align*}
    Since $\delta\le2^{-100}$, we know that $\Enc^\odot(T_1[i])\circ\cdots\circ T_u[i]$ is $2^{-70}$-far from solutions of $C'$.
    By \Cref{thm:pcpp-cktsat}, this means $\pcpptabletest$, which executes $\pcppckt$ with half probability, rejects with probability $\frac12\cdot\frac12=\Omega(\delta)$.
    Recall that we assumed that $\pcpptabletest$ does not reject with probability $\Omega(\delta)$, which is a contradiction.
    Thus \Cref{itm:c'_3} should also be satisfied and this completes the proof of \Cref{clm:tabletest_2}.
\end{proof}

\section*{Acknowledgement}
We thank Eli Ben-Sasson for clarifying questions regarding \cite{ben2003randomness} and thank Karthik C.S. for pointing out the application to {\sc Max $k$-Coverage}~(\Cref{hardness:maxkcoverage}). 

\bibliographystyle{alpha} 
\bibliography{ref}

\appendix
\section{Derandomized Parallel Low Degree Test}\label{sec:lowdeg}

In this section, we design a derandomized parallel low degree test to prove \Cref{thm:ldt-maintext} and \Cref{thm:pldt-circuit}. This is obtained by combining the derandomized low degree test \cite[Theorem 4.1]{ben2003randomness} with the parallel low degree test \cite[Lemma 7.7]{LRSW23}, where the latter builds on \cite{friedl1995some}.
While the combination is standard, we decide to expand the proof sketch in \cite{ben2003randomness} to a full proof for completeness.

\begin{theorem*}[\Cref{thm:ldt-maintext} Restated]
    Assume $\charF(\Fbb)=2$ and $|\Fbb|\ge\max\cbra{6md,2^{100}m\log|\Fbb|}$.
    Let $\Sigma=\Fbb^{d+1}$ be the set of univariate degree-$d$ polynomials over $\Fbb$.
    There exists an efficient verifier $\pcppldt$ with the following properties.
    \begin{itemize}
        \item The input of $\pcppldt$ is $T\circ \pi$, where $T\in (\Fbb^t)^{|\Fbb|^m}$ is supposed to be a codeword of $\prm$ and $\pi\in (\Sigma^t)^{|\Fbb|^m\cdot (m\log |\Fbb|)^{O(1)}}$ is the auxiliary proof.
        \item $\pcppldt$ tosses $m\log |\Fbb|+O\pbra{\log \log |\Fbb| +\log m}$ unbiased coins and makes $2$ queries on $T\circ \pi$.
        \item If $T\in \eccim(\prm)$, then there exists some $\pi$ such that $\pcppldt(T\circ \pi)$ always accepts.
        \item If $T$ is $\delta$-far from $\eccim(\prm)$, then $\Pr[\pcppldt(T\circ \pi)\ \text{rejects}]\ge 2^{-40}\delta$ for any $\pi$.
    \end{itemize}
\end{theorem*}

\begin{theorem*}[\Cref{thm:pldt-circuit} Restated]
Assume $\charF(\Fbb)=2$ and $|\Fbb|\ge\max\cbra{6md,2^{100}m\log|\Fbb|}$.
There exists a Boolean circuit $C_\ldt$ of size $\cktsize$ for $T\in (\Fbb^t)^{|\Fbb|^m}$, where we encode $\Fbb$ as $\bin^{\log|\Fbb|}$, such that $T$ is codeword of $\prm$ iff $T$ parallel satisfies $C_\ldt$.    
\end{theorem*}

\subsection{Extra Notation}\label{sec:ldt_extra_notation}

We first set up some necessary notation.

\paragraph{Parallel Low Degree Polynoimal}
We first recall the notion of parallel polynomial from \Cref{sec:prelim}.
Let $\Fbb$ be a finite field.
For a parallel-output function $f\colon\Fbb^m\to\Fbb^t$, we denote $f[1],\ldots,f[t]\colon\Fbb^m\to\Fbb$ as its single-output components, i.e., $f(x)=(f[1](x),\ldots,f[t](x))$.
We aim to test if $f[1],\ldots,f[t]$ are consistent with degree-$d$ polynomials on a large common set of inputs.
Formally, we say $f$ is $\delta$-close to parallel degree-$d$ iff $f$ is $\delta$-close to $\prm$; and we say $f$ is parallel degree-$d$ if $\delta=0$.

A simple union bound shows that if each $f[i]$ is $\delta$-close to degree-$d$ (e.g., from the standard low degree test), then $f=(f[1],\ldots,f[t])$ is $t\delta$-close to parallel degree-$d$. However for our purposes, such a loss is not affordable since $t$ is typically large.
Therefore we need to open-box the specific low degree test and show that it implies consistency on $1-\delta$ common fraction of inputs simultaneously for all $f[i]$.

\paragraph{Parallel Low Degree Test}
For $x,y\in\Fbb^m$, the line crossing $x$ in direction $y$ is the set
$$
\ell_{x,y}:=\cbra{x+t\cdot y\colon t\in\Fbb}.
$$
Note that if $y=0^m$ then $\ell_{x,y}=\cbra{x}$, otherwise $\ell_{x,y}$ has $|\Fbb|$ points.
Let $\Lbb$ be the set of lines in $\Fbb^m$.
For a parallel function $f\colon\Fbb^m\to\Fbb^t$ and any $\ell\in\Lbb$, we use $f|_\ell\colon\ell\to\Fbb^t$ to denote the restriction of $f$ on the line $\ell$.

Let $\Fbb^{(d,t)}$ be the set of univariate parallel degree-$d$ polynomials, i.e.,
$$
\Fbb^{(d,t)}:=\cbra{h\colon\Fbb\to\Fbb^t\colon h_i\text{ is a univariate degree-$d$ polynomial for each }i\in[t]}.
$$
The standard low degree test $\LDTest^{f,g}$ \cite{rubinfeld1996robust}, translated to the parallel setting \cite{LRSW23}, is given oracle access to $f\colon\Fbb^m\to\Fbb^t$ and $g\colon\Lbb\to\Fbb^{(d,t)}$, which later tests their consistency\footnote{The term consistency has been called correlated agreement, e.g., in recent literature of proximity gaps for RS codes and related literature in FRI protocol/SNARKs.}.
For each line $\ell\in\Lbb$, we denote $g(\ell)\in\Fbb^{(d,t)}$ as its parallel line polynomial, which we interpret as a univariate parallel degree-$d$ polynomial mapping elements in $\ell$ to $\Fbb^t$.
Thus for a line $\ell$ and a point $z\in\ell$, $g(\ell)(z)\in\Fbb^t$ is well defined.
We say $f$ agrees with $g(\ell)$ on $z$ if $f(z)=g(\ell)(z)$.

For each $f\colon\Fbb^m\to\Fbb^t$ and $d\in\Nbb$, we use $f_\Lbb\colon\Lbb\to\Fbb^{(d,t)}$ to denote the restriction of $f$ on each line to its closest parallel degree-$d$ polynomial. That is, for each line $\ell$, we set $f_\Lbb(\ell)$ to be the parallel degree-$d$ polynomial closest to $f|_\ell$, where we break tie arbitrarily.
Note that we will always assume the degree to be tested is $d$, and thus we omit $d$ in defining $f_\Lbb$ for simplicity.

In the completeness case (i.e., $f$ is indeed parallel degree-$d$), we can pick $f_\Lbb(\ell)=f|_\ell$ which is also parallel degree-$d$.
The parallel low degree test $\LDTest^{f,g}$ explores the reverse direction:
independently select $x,y\sim\Fbb^m$ and accept iff $f(x)=g(\ell_{x,y})(x)$, i.e., $f$ agrees with the parallel line polynomial $g(\ell_{x,y})$ on point $x$.
The parallel low degree test \cite{LRSW23} shows that if $\LDTest^{f,g}$ accepts with probability $1-\delta$, then $f$ is $O(\delta)$-close to parallel degree-$d$. Note that this bound does not depend on $t$.

\paragraph{Derandomized Parallel Low Degree Test}
Now we introduce the derandomized version of the parallel low degree test.
Following \cite{ben2003randomness}, this simply replaces the uniformly random direction $y$ by a pseudorandom $y\sim S$ for a much smaller set $S\subseteq\Fbb^m$.
Hence the number of randomness drops from $|\Fbb|^{2m}$ to $|S|\cdot|\Fbb|^m$.

\begin{definition}[Derandomized Parallel Low Degree Test]\label{def:derand_pldt}
Let $S\subseteq\Fbb^m$ and $f\colon\Fbb^m\to\Fbb^t,g\colon\Lbb\to\Fbb^{(d,t)}$.\footnote{Technically, $g$ only needs to be defined on lines whose directions are in $S$. We choose to assume $g$ is defined over all lines $\Lbb$ for simplicity. \label{fn:derandom_lines}}
The derandomized parallel low degree test $\LDTest^{f,g}_S$ is executed as follows:
independently select $x\sim\Fbb^m$ and $y\sim S$, then accept iff $f(x)=g(\ell_{x,y})(x)$.
\end{definition}

Later the set $S$ is chosen to be a $\lambda$-biased set as in \cite{ben2003randomness}.

\begin{definition}[$\lambda$-Biased Set]\label{def:eps_biased_set}
$S\subseteq\Fbb^m$ is a $\lambda$-biased set iff
\begin{itemize}
\item $S$ is symmetric, i.e., if $y\in S$ then $-y\in S$;\footnote{In some literature this symmetric assumption is not imposed. The parameter $\lambda$ in that case is comparable with the one here by a multiplicative factor of $2$.}
\item $\abs{\E_{y\sim S}\sbra{\chi(y)}}\le\lambda$ holds for any non-trivial homomorphism\footnote{This homomorphism is usually referred as character. It is trivial if it maps everything to $1$. An illustrative example is when $\Fbb=\Fbb_2$ and $\chi$ is a parity function.} $\chi\colon\Fbb^m\to\mu_p$, where $\mu$ is the multiplicative group of $p$-th unit root and $p$ is the characteristic of $\Fbb$.
\end{itemize}
\end{definition}

In a graph theoretical reformulation, $S$ is $\lambda$-biased iff the graph $G_S$ is an (undirected) expander graph with expansion factor $1-\lambda$, where the vertex set of $G_S$ is $\Fbb^m$ and $x,y\in G_S$ is connected iff $x-y\in S$ (see e.g., \cite{jalan2021near}).
For our purposes, we quote the following results derived directly from the expanding property of biased sets.

\begin{lemma}[{\cite[Lemma 4.3]{ben2003randomness}}]\label{lem:derand_sampling}
Suppose $S\subseteq\Fbb^m$ is $\lambda$-biased. Then for any $B\subseteq\Fbb^m$ of density $\mu=\frac{|B|}{|\Fbb|^m}$ and any $\eps>0$, we have
$$
\Pr_{x\sim\Fbb^m,y\sim S}\sbra{\abs{\frac{\abs{\ell_{x,y}\cap B}}{\abs{\ell_{x,y}}}-\mu}>\eps}\le\pbra{\frac1{|\Fbb|}+\lambda}\cdot\frac\mu{\eps^2}.
$$
\end{lemma}

Finally we remark that $\lambda$-biased sets exist for $|S|=\Omega\pbra{\frac{m\log|\Fbb|}{\lambda^2}}$ \cite{alon1994random} and efficient explicit constructions of sizes $\poly(m,\log|\Fbb|,\frac1\lambda)$ are also obtained.

\begin{fact}[See e.g., \cite{jalan2021near}]\label{fct:biased_construction}
For any finite field $\Fbb$, positive integer $m$, and parameter $\lambda\in(0,1]$, a $\lambda$-biased set of size $\poly(m,\log|\Fbb|,\frac1\lambda)$ can be constructed efficiently in time $\poly(m,|\Fbb|,\frac1\lambda)$. 
\end{fact}

\paragraph{Augmented Derandomized Parallel Low Degree Test}
Unfortunately the derandomized parallel low degree test is not sufficient to guarantee the exact quantitative degree condition, as the directions $S$ have fairly limited possibilities.
More concretely, even if $\LDTest_S^{f,g}$ succeeds with probability $1$, it is not guaranteed that $f$ is parallel degree-$d$.
To compensate the missing directions, we need to augment \Cref{def:derand_pldt} with an additional test that checks the consistency of $f$ and $g$ on a purely random direction from the origin \cite{ben2003randomness}.

\begin{definition}[Augmented Derandomized Parallel Low Degree Test]\label{def:aug_dpldt}
Let $S\subseteq\Fbb^m$ and $f\colon\Fbb^m\to\Fbb^t,g\colon\Lbb\to\Fbb^{(d,t)}$.
The augmented derandomized parallel low degree test $\AugLDTest_S^{f,g}$ is executed as follows: with equal probability we perform one of the following two tests:
\begin{itemize}
\item Independently select $x\sim\Fbb^m$ and $y\sim S$, then accept iff $f(x)=g(\ell_{x,y})(x)$.
\item Select $z\sim\Fbb^m$, then accept iff $f(z)=g(\ell_{0^m,z})(z)$.
\end{itemize}
\end{definition}

The first test in $\AugLDTest_S^{f,g}$ is simply $\LDTest_S^{f,g}$, for which we will later show that it guarantees that $f$ is close to parallel degree-$md$.
Then the second test allows us to further bring the degree down to $d$.

\subsection{Codeword Testing}

\Cref{thm:ldt-maintext} and \Cref{thm:pldt-circuit} follow directly from the following result, combined with the explicit constructions for biased sets.

\begin{theorem}\label{thm:derand_augpldt_detail}
Assume $|\Fbb|\ge6md$, $|\Fbb|\ge2^{100}\cdot m\log|\Fbb|$, and $\lambda\le\frac1{2^{100}\cdot m\log|\Fbb|}$.
Let $S\subseteq\Fbb^m$ be a $\lambda$-biased set.
If
$$
\Pr\sbra{\AugLDTest_S^{f,g}\text{ accepts}}\ge1-\delta,
$$
then $f$ is $2^{40}\delta$-close to parallel degree-$d$.
\end{theorem}

\begin{proof}[Proof of \Cref{thm:ldt-maintext}]
We first note that $|\Fbb|^m>\binom{m+d}d$ and $|\Fbb|>d$ are satisfied assuming the conditions on $|\Fbb|$ in \Cref{thm:ldt-maintext}, thus $\prm$ is well defined.

Then we instantiate \Cref{thm:derand_augpldt_detail} with $\lambda=\frac1{2^{100}\cdot m\log|\Fbb|}$ and the $\lambda$-biased set $S$ by \Cref{fct:biased_construction}.
The construction of $S$ is efficient in time $\poly(|\Fbb|,m)$ and $S$ has size $(m\log|\Fbb|)^{O(1)}$.
Then we define $\pcppldt$ as $\AugLDTest_S^{f,g}$ where $T=f$ and $\pi$ is defined to be the entries of $g$ that can be possibly queried. Recall \Cref{def:aug_dpldt}. Then we have
$$
|\pi|\le|\Fbb|^m\cdot|S|+|\Fbb|^m=|\Fbb|^m\cdot(m\log|\Fbb|)^{O(1)}.
$$
In addition, by merging the randomness of the two tests in $\AugLDTest_S^{f,g}$, $\pcppldt$ tosses 
$$
1+\log\pbra{|\Fbb|^m|S|}=m\log|\Fbb|+O\pbra{\log\log|\Fbb|+\log m}
$$ 
total coins.
The completeness is obvious and the soundness follows from \Cref{thm:derand_augpldt_detail}.
\end{proof}

\begin{proof}[Proof of \Cref{thm:pldt-circuit}]
Now we turn to \Cref{thm:pldt-circuit}.
By \Cref{thm:ldt-maintext}, it suffices to implement $\pcppldt$ purely on $f$.
In addition, since $\pcppldt=\AugLDTest_S^{f,g}$ and by \Cref{def:aug_dpldt}, it performs the same check in parallel for each coordinate $i\in[t]$.
This means that we only need to instantiate $\pcppldt$ for a single coordinate (or equivalently, think of $t=1$) to design the circuit $C_\ldt$.

To get rid of the extra proof $g$, we simply set $g=f_\Lbb$.
Then, whenever we need information about entries in $g$ (i.e., a line polynomial), we can probe the entries along the line in $f$ to compute it.
We remark that this is inefficient in terms of the query complexity, but it is still efficient in terms of the circuit complexity.

Now we describe the construction of $C_\ldt$ for a fixed coordinate $i\in[t]$.
Based on the coin toss of $\pcppldt$, it checks the consistency of an $\Fbb$-valued point (i.e., $f[i](x)$ or $f[i](z)$) with an evaluation point (i.e., $x$ or $z$) of a degree-$d$ line polynomial over $\Fbb$ (i.e., $f_{\ell_{x,y}}[i]$ or $f_{\ell_{0^m,z}}[i]$).
To implement it as a circuit, we take the conjunction of all the sub-circuit outcome from coin toss possibilities, where each sub-circuit performs the following computation.
\begin{itemize}
\item It first interpolates the line polynomial using entries of $f[i]$ along the line, and checks if this line polynomial is degree-$d$.

This can be efficiently done with $\poly|\Fbb|$ gates.
\item Then it evaluates the value of the desired point of the line polynomial, and checks if it is the same as the one directly obtained from $f[i]$.

This also requires $\poly|\Fbb|$ gates only.
\end{itemize}

The correctness of $C_\ldt$ follows directly from \Cref{thm:ldt-maintext} and our choice of $g=f_\Lbb$.
Since the number of coin toss possibilities is $|\Fbb|^m\poly(m,\log|\Fbb|)$, by our assumption on $|\Fbb|$, the size of $C_\ldt$ is $\cktsize$ as claimed.
\end{proof}

To prove the statement about $\AugLDTest$, we need to analyze $\LDTest$ first, which guarantees a weaker degree bound.

\begin{theorem}\label{thm:derand_pldt_detail}
Assume $|\Fbb|\ge3d$, $|\Fbb|\ge2^{100}\cdot m\log|\Fbb|$, and $\lambda\le\frac1{2^{100}\cdot m\log|\Fbb|}$.
Let $S\subseteq\Fbb^m$ be a $\lambda$-biased set.
If
$$
\Pr\sbra{\LDTest_S^{f,g}\text{ accepts}}\ge1-\delta,
$$
then $f$ is $2^{30}\delta$-close to parallel degree-$md$.
\end{theorem}

Assuming \Cref{thm:derand_pldt_detail}, we conclude the proof of \Cref{thm:derand_augpldt_detail}.

\begin{proof}[Proof of \Cref{thm:derand_augpldt_detail}]
First we assume $\delta\le2^{-40}$ since otherwise $2^{40}\delta\ge1$ and the statement trivially holds.
Recall \Cref{def:aug_dpldt} that $\AugLDTest_S^{f,g}$ executes $\LDTest_S^{f,g}$ with probability $1/2$.
Hence $\LDTest_S^{f,g}$ must accept with probability at least $1-2\delta$.
By \Cref{thm:derand_pldt_detail}, this means that $f$ is $2^{31}\delta$-close to a parallel degree-$md$ polynomial $f'$. It suffices to show that $f'$ is acutally parallel degree-$d$.

Assume towards contradiction that $f'$ is not parallel degree-$d$.
Now we consider the second half of $\AugLDTest_S^{f,g}$, which checks if $f(z)=g(\ell_{0^m,z})(z)$ for $z\sim\Fbb^m$.
Then we have
\begin{align}
\Pr_{z\sim\Fbb^m}\sbra{f(z)=g(\ell_{0^m,z})(z)}
&\le\Pr_{z\sim\Fbb^m}\sbra{f(z)\ne f'(z)}+\Pr_{z\sim\Fbb^m}\sbra{f'(z)=g(\ell_{0^m,z})(z)}
\notag\\
&\le2^{31}\delta+\Pr_{z\sim\Fbb^m}\sbra{f'(z)=g(\ell_{0^m,z})(z)}.
\label{eq:thm:derand_augpldt_detail_1}
\end{align}
To analyze \Cref{eq:thm:derand_augpldt_detail_1}, we consider the following quantity:
\begin{equation}\label{eq:thm:derand_augpldt_detail_2}
\Pr_{w\sim\Fbb^m,i\sim\Fbb}\sbra{f'(i\cdot w)=g(\ell_{0^m,w})(i\cdot w)}.
\end{equation}

On the one hand, conditioned on $i\ne0$, we have $\ell_{0^m,w}=\ell_{0^m,i\cdot w}$ and $i\cdot w$ being uniform in $\Fbb^m$.
Therefore we can relate \Cref{eq:thm:derand_augpldt_detail_1} with \Cref{eq:thm:derand_augpldt_detail_2}:
\begin{equation}\label{eq:thm:derand_augpldt_detail_3}
\Cref{eq:thm:derand_augpldt_detail_2}
\ge\pbra{1-\frac1{|\Fbb|}}\cdot\Pr_{z\sim\Fbb^m}\sbra{f'(z)=g(\ell_{0^m,z})(z)}.
\end{equation}
On the other hand, conditioned on $w$, $f'(i\cdot w)$ is a univariate parallel degree-$md$ polynomial in $i$.
Let $d<d'\le md$ be the parallel degree of $f'$.
Then the coefficient of $i^{d'}$ in $f'(i\cdot w)$ is a non-zero parallel degree-$d'$ polynomial.
By Schwartz–Zippel lemma, it vanishes on at most $\frac{d'}{|\Fbb|}$ fraction of choices of $w$.
For each $w$ that the top coefficient of $i^{d'}$ does not vanish, by Schwartz–Zippel lemma, $f'(i\cdot w)$ agrees with $g(\ell_{0^m,w})(i\cdot w)$ on at most $d'$ choices of $i$ since $g(\ell_{0^m,w})$ is parallel degree-$d$ and $d<d'$.
This means
\begin{align*}
\Cref{eq:thm:derand_augpldt_detail_2}
&\le\Pr_{w\sim\Fbb^m}\sbra{\text{coeff of }i^{d'}\text{ in }f'(i\cdot w)\text{ vanishes}}
+\Pr_{w\sim\Fbb^m,i\sim\Fbb}\sbra{f'(i\cdot w)=g(\ell_{0^m,w})(i\cdot w)\mid\text{not vanish}}\\
&\le\frac{d'}{|\Fbb|}+\frac{d'}{|\Fbb|}\le\frac{2md}{|\Fbb|}.
\tag{since $d'\le md$}
\end{align*}
Combining this with \Cref{eq:thm:derand_augpldt_detail_3} and \Cref{eq:thm:derand_augpldt_detail_1}, we have
\begin{align*}
\Pr_{z\sim\Fbb^m}\sbra{f(z)=g(\ell_{0^m,z})(z)}
&\le2^{31}\delta+\frac{2md}{|\Fbb|-1}\le\frac12,
\end{align*}
where we used the fact that $\delta\le2^{-40}$ and $|\Fbb|\ge6md$.
Since this is half of the actual $\AugLDTest_S^{f,g}$, it means
$$
\Pr\sbra{\AugLDTest_S^{f,g}\text{ accepts}}\le\frac12+\frac12\cdot\frac12<1-\delta,
$$
which is a contradiction.

In conclusion, $f'$ must be parallel degree-$d$.
\end{proof}

Our \Cref{thm:derand_pldt_detail} is the parallel version of \cite[Theorem 4.1]{ben2003randomness}.
Its proof follows from iteratively applying the following lemma, which is the parallel version of \cite[Lemma 4.4]{ben2003randomness}.
\begin{lemma}\label{lem:derand_pldt}
Assume $|\Fbb|\ge3d$.
Let $S\subseteq\Fbb^m$ be a $\lambda$-biased set and $T\subseteq S$ of size $|T|\ge|S|/2$.
Let $f\colon\Fbb^m\to\Fbb^t$.
If
$$
\Pr\sbra{\LDTest_T^{f,f_\Lbb}\text{ accepts}}\ge1-\delta,
$$
then for any $2\delta\le\gamma\le2^{-20}$, there exists $f'\colon\Fbb^m\to\Fbb^t$ and $T'\subseteq T$ with the following properties:
\begin{enumerate}
\item\label{itm:lem:derand_pldt_1}
$|T'|\ge\pbra{1-\frac\delta\gamma}|T|\ge\frac{|T|}2$.
\item\label{itm:lem:derand_pldt_2}
$\Delta(f',f)\le4\delta$.
\item\label{itm:lem:derand_pldt_3}
$\Pr\sbra{\LDTest^{f',f_\Lbb}_{T'}\text{ accepts}}\ge1-2^{40}\cdot\gamma\cdot\pbra{\frac1{|\Fbb|}+\lambda}$.
\end{enumerate}
\end{lemma}

Assuming \Cref{lem:derand_pldt}, we first conclude \Cref{thm:derand_pldt_detail}.
\begin{proof}[Proof of \Cref{thm:derand_pldt_detail}]
First we assume $\delta\le2^{-30}$ since otherwise $2^{30}\delta\ge1$ and the statement trivially holds.
Next, we can also assume without loss of generality that $g=f_\Lbb$.
This is because, for each possible\footnote{A line is possible in $\LDTest_S^{f,g}$ if its direction lies in $S$.} line $\ell\in\Lbb$, $\LDTest^{f,g}_S$ conditioned on this line checks a uniformly random point $x\in\ell$ whether $f(x)=g(\ell)(x)$. Since $g(\ell)$ is parallel degree-$d$, the success probability maximized when $g(\ell)=f_\Lbb(\ell)$.

We will repeatedly apply \Cref{lem:derand_pldt} to bring the soundness gap down to $\ll|\Fbb|^{-2m}$, at which point it is actually zero by granularity.
Then we use the following characterization of parallel low degree polynomials similar to \cite{rubinfeld1996robust} to arrive at an actual parallel degree-$md$ polynomial.

\begin{theorem}\label{thm:exact_derand_ldp}
Let $S\subseteq\Fbb^m$ be a $\lambda$-biased set and $T\subseteq S$ of size $|T|>\frac{1+\lambda}2\cdot|S|$.
Then $f\colon\Fbb^m\to\Fbb^t$ is parallel degree-$md$ if $f|_{\ell_{x,y}}$ is parallel degree-$d$ for every $x\in\Fbb^m,y\in T$.
\end{theorem}

\Cref{lem:derand_pldt} will be proved in \Cref{sec:derand_char_ldp}. Now we focus on reducing the soundness gap.

\paragraph{The $\delta\ge\frac1{2^{60}\cdot m\log|\Fbb|}$ Case}
In this case, we first perform a pre-processing round to bring down the soundness gap.
By \Cref{lem:derand_pldt} with $T=S$ and $\gamma=2^{-20}$, we have $S_1\subseteq S$ and $f^{(1)}\colon\Fbb^m\to\Fbb^t$ with the following property:
\begin{enumerate}
\item\label{itm:thm:derand_pldt_detail_1}
$|S_1|\ge\pbra{1-2^{-10}}\cdot|S|$, since $\delta\le2^{-30}$.
\item\label{itm:thm:derand_pldt_detail_2}
$\Delta(f^{(1)},f)\le4\delta$.
\item\label{itm:thm:derand_pldt_detail_3}
$\Pr\sbra{\LDTest^{f^{(1)},f_\Lbb}_{S_1}\text{ accepts}}\ge1-\frac1{2^{60}\cdot m\log|\Fbb|}$, since $|\Fbb|\ge2^{100}\cdot m\log|\Fbb|$ and $\lambda\le\frac1{2^{100}\cdot m\log|\Fbb|}$.
\end{enumerate}
Now define 
$$
\delta_i=\frac{2^{-i}}{2^{50}\cdot m\log|\Fbb|}
\quad\text{and}\quad
\gamma_i=\delta_i\cdot2^{10}\cdot m\log|\Fbb|=\frac{2^{-i}}{2^{40}}.
$$
For each $i=1,\ldots,2\ceilbra{m\log|\Fbb|}$, we apply \Cref{lem:derand_pldt} on $\delta_i,S_i,f^{(i)}$ and obtain $\delta_{i+1},S_{i+1},f^{(i+1)}$.
To show the correctness of this process, we verify by induction on $i$ that the conditions in \Cref{lem:derand_pldt} are satisfied, i.e.,\footnote{In \Cref{lem:derand_pldt} we only require $|S_i|\ge|S|/2$. Here we strengthen it for convenience of \Cref{thm:exact_derand_ldp}.}
\begin{equation}\label{eq:thm:derand_pldt_detail_1}
|S_i|\ge0.9|S|
\quad\text{and}\quad
\Pr\sbra{\LDTest_{S_i}^{f^{(i)},f_\Lbb}}\ge1-\delta_i,
\end{equation}
where we omit $2\delta_i\le\gamma_i\le2^{-20}$ since it holds by the definition of $\delta_i,\gamma_i$.

The base case $i=1$ is valid by \Cref{itm:thm:derand_pldt_detail_1} and \Cref{itm:thm:derand_pldt_detail_3}.
For the inductive cases $i\ge2$, we first observe that the first condition in \Cref{eq:thm:derand_pldt_detail_1} follows from the following calculation:
\begin{align*}
|S_i|
&\ge\pbra{1-\frac{\delta_{i-1}}{\gamma_{i-1}}}|S_{i-1}|
=\pbra{1-\frac1{2^{10}\cdot m\log|\Fbb|}}|S_{i-1}|
\tag{by \Cref{lem:derand_pldt}}\\
&\ge\cdots\ge\pbra{1-\frac1{2^{10}\cdot m\log|\Fbb|}}^{i-1}|S_1|
\tag{by \Cref{lem:derand_pldt} iteratively}\\
&\ge\pbra{1-\frac1{2^{10}\cdot m\log|\Fbb|}}^{i-1}\cdot\pbra{1-2^{-10}}|S|
\tag{by \Cref{itm:thm:derand_pldt_detail_1}}\\
&\ge0.9|S|.
\tag{since $i\le2\ceilbra{m\log|\Fbb|}$}
\end{align*}
The second condition in \Cref{eq:thm:derand_pldt_detail_1} follows from last round of \Cref{lem:derand_pldt}, which establishes that
\begin{align*}
\Pr\sbra{\LDTest_{S_i}^{f^{(i)},f_\Lbb}\text{ accepts}}
&\ge1-2^{40}\cdot\gamma_{i-1}\cdot\pbra{\frac1{|\Fbb|}+\lambda}
=1-2^{-i+1}\cdot\pbra{\frac1{|\Fbb|}+\lambda}\\
&\ge1-2^{-i+1}\cdot2\cdot\frac1{2^{100}\cdot m\log|\Fbb|}
\ge1-\delta_i.
\end{align*}

Let $k=2\ceilbra{m\log|\Fbb|}$. Then the above analysis shows 
$$
|S_k|\ge0.9|S|>\frac{1+\lambda}2\cdot|S|
\quad\text{and}\quad
\Pr\sbra{\LDTest_{S_k}^{f^{(k)},f_\Lbb}\text{ accepts}}\ge1-\frac{2^{-k}}{2^{50}\cdot m\log|\Fbb|}>1-|\Fbb|^{-2m}.
$$
Since $\LDTest^{f^{(k)},f_\Lbb}_{S_k}$ samples $x\sim\Fbb^m$ and $y\sim S_k\subseteq\Fbb^m$ then perform a deterministic check whether $f^{(k)}(x)=f_\Lbb(\ell_{x,y})(x)$, its accepting probability is an integer multiple of $\frac1{|\Fbb|^m\cdot|S_k|}\ge|\Fbb|^{-2m}$.
Therefore, we actually have $\Pr\sbra{\LDTest_{S_k}^{f^{(k)},f_\Lbb}\text{ accepts}}=1$, which means that $f^{(k)}$ is parallel degree-$d$ on $\ell_{x,y}$ for all $x\in\Fbb^m,y\in S_k$. 
By \Cref{thm:exact_derand_ldp}, this means that $f^{(k)}$ is parallel degree-$md$.
In addition, its distance from the original $f$ is
\begin{align*}
\Delta(f^{(k)},f)
&\le\Delta(f^{(1)},f)+\sum_{i=1}^{k-1}\Delta(f^{(i+1)},f^{(i)})
\le4\delta+\sum_{i=1}^{k-1}\Delta(f^{(i+1)},f^{(i)})
\tag{by \Cref{itm:thm:derand_pldt_detail_2}}\\
&\le4\delta+\sum_{i=1}^{k-1}4\cdot\delta_i=4\delta+\sum_{i=1}^{k-1}\frac{4\cdot2^{-i}}{2^{50}\cdot m\log|\Fbb|}
\tag{by \Cref{lem:derand_pldt}}\\
&\le2^{30}\delta.
\tag{since $\delta\ge\frac1{2^{60}\cdot m\log|\Fbb|}$}
\end{align*}

\paragraph{The $\delta<\frac1{2^{60}\cdot m\log|\Fbb|}$ Case}
In this case, the analysis is even simpler as we do not need pre-processing.
Let $S_1=S$ and $f^{(1)}=f$.
Define 
$$
\delta_i=2^{-i+1}\cdot\delta
\quad\text{and}\quad
\gamma_i=\delta_i\cdot2^{10}\cdot m\log|\Fbb|=2^{-i+11}\cdot\delta\cdot m\log|\Fbb|.
$$
We also apply \Cref{lem:derand_pldt} for each $i=1,\ldots,2\ceilbra{m\log|\Fbb|}$ on $\delta_i,S_i,f^{(i)}$ to obtain $\delta_{i+1},S_{i+1},f^{(i+1)}$.
The conditions in \Cref{lem:derand_pldt} along this process can be verified in the similar fashion.
Here we only highlight the difference:
for the second condition, we have
\begin{align*}
\Pr\sbra{\LDTest_{S_i}^{f^{(i)},f_\Lbb}\text{ accepts}}
&\ge1-2^{40}\cdot\gamma_{i-1}\cdot\pbra{\frac1{|\Fbb|}+\lambda}
=1-2^{-i+51}\cdot\delta\cdot m\log|\Fbb|\cdot\pbra{\frac1{|\Fbb|}+\lambda}\\
&\ge1-2^{-i-48}\cdot\delta
\ge1-\delta_i,
\end{align*}
and for the third condition, we have
\begin{align*}
\gamma_i=2^{-i+11}\cdot\delta\cdot m\log|\Fbb|\le\frac{2^{-i+11}}{2^{60}}\le2^{-20}.
\end{align*}
Then similarly, we set $k=2\ceilbra{m\log|\Fbb|}$ and obtain $f^{(k)}$ as a parallel degree-$md$ polynomial. Moreover, we have
\begin{align*}
\Delta(f^{(k)},f)
&\le\sum_{i=1}^{k-1}\Delta(f^{(i+1)},f^{(i)})
\le\sum_{i=1}^{k-1}4\cdot\delta_i
\tag{since $f^{(1)}=f$ and by \Cref{lem:derand_pldt}}\\
&=\sum_{i=1}^{k-1}4\cdot2^{-i+1}\cdot\delta\le2^{30}\delta,
\end{align*}
which completes the proof.
\end{proof}

\subsection{One Round of Correction}\label{sec:one_round_correction_pldt}

This section is devoted to proving \Cref{lem:derand_pldt}, which follows the sketch outlined in \cite{ben2003randomness}.

\begin{lemma*}[\Cref{lem:derand_pldt} Restated]
Assume $|\Fbb|\ge3d$.
Let $S\subseteq\Fbb^m$ be a $\lambda$-biased set and $T\subseteq S$ of size $|T|\ge|S|/2$.
Let $f\colon\Fbb^m\to\Fbb^t$.
If
$$
\Pr\sbra{\LDTest_T^{f,f_\Lbb}\text{ accepts}}\ge1-\delta,
$$
then for any $2\delta\le\gamma\le2^{-20}$, there exists $f'\colon\Fbb^m\to\Fbb^t$ and $T'\subseteq T$ with the following properties:
\begin{enumerate}
\item
$|T'|\ge\pbra{1-\frac\delta\gamma}|T|\ge\frac{|T|}2$.
\item
$\Delta(f',f)\le4\delta$.
\item
$\Pr\sbra{\LDTest^{f',f_\Lbb}_{T'}\text{ accepts}}\ge1-2^{40}\cdot\gamma\cdot\pbra{\frac1{|\Fbb|}+\lambda}$.
\end{enumerate}
\end{lemma*}
\begin{proof}
Let $T'\subseteq T$ be the set of directions $y\in T$ such that for at least $1-\gamma$ fraction of $x\in\Fbb^m$, $f$ agrees with the parallel line polynomial $\ell_{x,y}$ on $x$. That is,
\begin{equation}\label{eq:lem:derand_pldt_1}
T'=\cbra{y\in T\colon|\cbra{x\in\Fbb^m\colon f(x)=f_\Lbb(\ell_{x,y})(x)}|\ge(1-\gamma)\cdot|\Fbb|^m}.
\end{equation}
Then
\begin{align*}
1-\delta
&\le\Pr\sbra{\LDTest_T^{f,f_\Lbb}\text{ accepts}}
\tag{by assumption}\\
&=\Pr_{x\sim\Fbb^m,y\sim T}\sbra{f(x)=f_\Lbb(\ell_{x,y})(x)}
\tag{by \Cref{def:derand_pldt}}\\
&=\frac{|T'|}{|T|}\!\Pr_{x\sim\Fbb^m,y\sim T}\sbra{f(x)=f_\Lbb(\ell_{x,y})(x)\mid y\in T'}+\frac{|T\setminus T'|}{|T|}\!\Pr_{x\sim\Fbb^m,y\sim T}\sbra{f(x)=f_\Lbb(\ell_{x,y})(x)\mid y\notin T'}\\
&\le\frac{|T'|}{|T|}+\pbra{1-\frac{|T'|}{|T|}}\cdot\pbra{1-\gamma},
\tag{by \Cref{eq:lem:derand_pldt_1}}
\end{align*}
which implies \Cref{itm:lem:derand_pldt_1} by rearranging.

To construct $f'\colon\Fbb^m\to\Fbb^t$, for each $x\in\Fbb^m$ we define $f'(x)\in\Fbb^t$ to be the most common value of $f_\Lbb(\ell_{x,y})(x)$ over $y\in T'$ where we break tie arbitrarily.
Now we verify \Cref{itm:lem:derand_pldt_2}.
Let
$$
B=\cbra{x\in\Fbb^m\colon\Pr_{y\sim T'}\sbra{f(x)\neq f_\Lbb(\ell_{x,y})(x)}\ge\frac12}.
$$
By the definition of $f'$, we know that $f'(x)=f(x)$ holds for any $x\notin B$.
Hence \Cref{itm:lem:derand_pldt_2} reduces to showing $\Pr_{x\sim\Fbb^m}[x\in B]\le4\delta$ as follows:
\begin{align*}
\Pr_{x\sim\Fbb^m}[x\in B]
&=\Pr_{x\sim\Fbb^m}\sbra{\Pr_{y\sim T'}\sbra{f(x)\neq f_\Lbb(\ell_{x,y})(x)}\ge\frac12}\\
&\le2\cdot\Pr_{x\sim\Fbb^m,y\sim T'}\sbra{f(x)\neq f_\Lbb(\ell_{x,y})(x)}
\tag{by Markov's inequality}\\
&\le2\cdot\frac{|T|}{|T'|}\cdot\Pr_{x\sim\Fbb^m,y\sim T}\sbra{f(x)\neq f_\Lbb(\ell_{x,y})(x)}
\tag{since $T'\subseteq T$}\\
&=2\cdot\frac{|T|}{|T'|}\cdot\Pr\sbra{\LDTest_T^{f,f_\Lbb}\text{ rejects}}
\le2\cdot\frac{|T|}{|T'|}\cdot\delta
\tag{by assumption}\\
&\le4\delta.
\tag{by \Cref{itm:lem:derand_pldt_1}}
\end{align*}

To prove \Cref{itm:lem:derand_pldt_3}, we first get rid of $f'$.
For a fixed $x\in\Fbb^m$ and each $b\in\Fbb^t$, let $U_b\subseteq T'$ be the set of directions $y$ such that $f_\Lbb(\ell_{x,y})(x)=b$.
Denote $b^*=f'(x)$, which is defined to be the most common value of $f_\Lbb(\ell_{x,y})(x)$ over $y\in T'$. Thus $|U_{b^*}|\ge|U_b|$ for any $b\in\Fbb^t$.
As a result, we have
\begin{align*}
\Pr_{y\sim T'}\sbra{f'(x)=f_\Lbb(\ell_{x,y})(x)}
&=\frac{|U_{b^*}|}{|T'|}=\frac{|U_{b^*}|}{|T'|}\cdot\sum_b\frac{|U_b|}{|T'|}
\tag{since $U_b$'s form a partition of $T'$}\\
&\ge\sum_b\pbra{\frac{|U_b|}{|T'|}}^2
\tag{since $|U_{b^*}|\ge|U_b|$ for all $b$}\\
&=\Pr_{y_1,y_2\sim T'}\sbra{f_\Lbb(\ell_{x,y_1})(x)=f_\Lbb(\ell_{x,y_2})(x)}.
\end{align*}
Now taking the negation and the expectation over random $x$, we have
\begin{align}
\Pr\sbra{\LDTest^{f',f_\Lbb}_{T'}\text{ rejects}}
&=\Pr_{x\sim\Fbb^m,y\sim T'}\sbra{f'(x)\neq f_\Lbb(\ell_{x,y})(x)}
\notag\\
&\le\Pr_{x\sim\Fbb^m,y_1,y_2\sim T'}\sbra{f_\Lbb(\ell_{x,y_1})(x)\neq f_\Lbb(\ell_{x,y_2})(x)}.
\label{eq:lem:derand_pldt_2}
\end{align}
To upper bound \Cref{eq:lem:derand_pldt_2}, we will use bivariate testing theorems \cite{polishchuk1994nearly} similar to \cite{LRSW23,friedl1995some}.
The idea is, as sketched in \cite{ben2003randomness}, to use \Cref{lem:derand_sampling} to show the following claim.

\begin{claim}\label{clm:derand_bivariate}
With probability at least
$$
1-2^{40}\cdot\gamma\cdot\pbra{\frac1{|\Fbb|}+\lambda}
$$
over $x,y_1,y_2$, we have
\begin{enumerate}[label=(\alph*)]
\item\label{itm:bivariate_1}
$\Pr_{i\sim\Fbb}\sbra{f_\Lbb(\ell_{x+i\cdot y_1,y_2})(x+i\cdot y_1)\ne f(x+i\cdot y_1)}\le\frac1{200}$.
\item\label{itm:bivariate_2}
$\Pr_{j\sim\Fbb}\sbra{f_\Lbb(\ell_{x+j\cdot y_2,y_1})(x+j\cdot y_2)\ne f(x+j\cdot y_2)}\le\frac1{200}$.
\item\label{itm:bivariate_3}
$\Pr_{i,j\sim\Fbb}\sbra{f_\Lbb(\ell_{x+i\cdot y_1,y_2})(x+i\cdot y_1+j\cdot y_2)\ne f(x+i\cdot y_1+j\cdot y_2)}\le\frac1{200}$.
\item\label{itm:bivariate_4}
$\Pr_{i,j\sim\Fbb}\sbra{f_\Lbb(\ell_{x+j\cdot y_2,y_1})(x+i\cdot y_1+j\cdot y_2)\ne f(x+i\cdot y_1+j\cdot y_2)}\le\frac1{200}$.
\end{enumerate}
\end{claim}

\Cref{clm:derand_bivariate} intuitively says that $f$ is almost parallel degree-$d$ along the line $\ell_{x,y_1}$ (\Cref{itm:bivariate_1}), the line $\ell_{x,y_2}$ (\Cref{itm:bivariate_2}), and the plane spanned by $\ell_{x,y_1},\ell_{x,y_2}$ (\Cref{itm:bivariate_3,itm:bivariate_4}).
For simplicity, define
\begin{equation}\label{eq:lem:derand_pldt_3}
R(i,j)=f_\Lbb(\ell_{x+i\cdot y_1,y_2})(x+i\cdot y_1+j\cdot y_2)
\quad\text{and}\quad
C(i,j)=f_\Lbb(\ell_{x+j\cdot y_2,y_1})(x+i\cdot y_1+j\cdot y_2).
\end{equation}
Given \Cref{clm:derand_bivariate} and \Cref{eq:lem:derand_pldt_2}, it suffices to show 
\begin{equation}\label{eq:lem:derand_pldt_5}
R(0,0)=C(0,0)
\quad\text{assuming \Cref{itm:bivariate_1,itm:bivariate_2,itm:bivariate_3,itm:bivariate_4}.}
\end{equation}

Observe that $R\colon\Fbb\times\Fbb\to\Fbb^t$ is a parallel bivariate polynomial. Recall that $f_\Lbb$ is the restriction of $f$ on each line of its closest parallel degree-$d$ polynomial. Thus for each fixed $i\in\Fbb$, $R(i,j)$ is parallel degree-$d$ in the variable $j$.
Let $R_1,\ldots,R_t\colon\Fbb\times\Fbb\to\Fbb$ be the single-output components of $R$.
Then for each $r\in[t]$ and in $R_r(i,j)$, the variable $i$ has degree at most $|\Fbb|-1$ and the variable $j$ has degree at most $d$. We say $R$ is of parallel degree $(|\Fbb|-1,d)$ for shorthand.
Similarly, $C(i,j)$ is of parallel degree $(d,|\Fbb|-1)$.

Combining \Cref{itm:bivariate_3,itm:bivariate_4}, we have
\begin{equation}\label{eq:lem:derand_pldt_4}
\Pr_{i,j\sim\Fbb}\sbra{R(i,j)\neq C(i,j)}\le\frac1{100}.
\end{equation}
We will use the following lemma to zero out the inconsistent entries of $R(i,j)$ and $C(i,j)$.

\begin{lemma}[{\cite[Lemma 3]{polishchuk1994nearly}}]\label{lem:PS94_zeros}
Let $Z\subseteq\Fbb\times\Fbb$ be arbitrary. There exists a non-zero bivariate polynomial $E\colon\Fbb\times\Fbb\to\Fbb$ of degree $\pbra{\floorbra{\sqrt{|Z|}},\floorbra{\sqrt{|Z|}}}$ such that $E(i,j)=0$ for all $(i,j)\in Z$.
\end{lemma}

By setting $Z$ to be the set of $(i,j)$ with $R(i,j)\neq C(i,j)$, we have $|Z|\le|\Fbb|^2/100$ by \Cref{eq:lem:derand_pldt_4}.
Thus by \Cref{lem:PS94_zeros}, there is a non-zero bivariate polynomial $E$ of degree at most $\pbra{\frac{|\Fbb|}{10},\frac{|\Fbb|}{10}}$ such that $E(i,j)R(i,j)=E(i,j)C(i,j)$ holds for all $i,j\in\Fbb$. Note that $R$ and $C$ are polynomials with $t$ outputs and $E$ has single output. The product $E(i,j)R(i,j)$ produces a vector of length $t$ with entries of $R(i,j)$ scaled by $E(i,j)$; same for $E(i,j)C(i,j)$.

Then we use the following lemma to show that $R$ and $C$ are close to a parallel bivariate polynomial of parallel degree $(d,d)$.

\begin{lemma}[{\cite[Lemma 8]{polishchuk1994nearly}}]\label{lem:PS94_EPQ}
Let $E,P\colon\Fbb\times\Fbb\to\Fbb$ be bivariate polynomials of degree $(b,a)$ and $(b+d,a+d)$ respectively.
Assume $n>\min\cbra{2b+2d,2a+2d}$.
Assume further there exist distinct $i_1,\ldots,i_n$ such that $E(i_k,\cdot)$ divides $P(i_k,\cdot)$ for all $k\in[n]$, and distinct $j_1,\ldots,j_n$ such that $E(\cdot,j_k)$ divides $P(\cdot,j_k)$ for all $k\in[n]$.
Then $E(\cdot,\cdot)$ divides $P(\cdot,\cdot)$, i.e., there exists a bivariate polynomial $Q\colon\Fbb\times\Fbb\to\Fbb$ of degree $(d,d)$ such that $E(i,j)Q(i,j)=P(i,j)$ holds for all $i,j\in\Fbb$.
\end{lemma}

Fix an arbitrary $r\in[t]$.
For each $i,j\in\Fbb$, define $P_r(i,j)=E(i,j)R_r(i,j)$, which also equals $E(i,j)C_r(i,j)$ since $E(i,j)R(i,j)=E(i,j)C(i,j)$.
Since $E$ is non-zero and has degree $\pbra{\frac{|\Fbb|}{10},\frac{|\Fbb|}{10}}$, there are at least $\frac{9|\Fbb|}{10}$ many distinct $i\in\Fbb$ such that $E(i,\cdot)$ is an univariate polynomial not identically zero, for which $E(i,\cdot)$ divides $P_r(i,\cdot)$; same for $E(\cdot,j)$ and $P_r(\cdot,j)$.
By \Cref{lem:PS94_EPQ} with $a=b=\frac{|\Fbb|}{10}$ and $n=\frac{9|\Fbb|}{10}$, there exists a bivariate polynomial $Q_r\colon\Fbb\times\Fbb\to\Fbb$ with parallel degree $(d,d)$ such that $E(i,j)Q_r(i,j)=P_r(i,j)=E(i,j)R_r(i,j)=E(i,j)C_r(i,j)$ holds for all $i,j\in\Fbb$, where we used the fact that $|\Fbb|\ge3d$ implies $|\Fbb|>\frac{10d}7$.\footnote{If $d\ge1$, then $3d>\frac{10d}7$; otherwise $\frac{10d}7=0<1\le|\Fbb|$.}

Let $Q\colon\Fbb\times\Fbb\to\Fbb^t$ be the parallel-output bivariate polynomial with single-output components $Q_1,\ldots,Q_r$ obtained above. Then $Q(i,\cdot)=R(i,\cdot)$ holds as long as $E(i,\cdot)$ is not identically zero, which has at least $\frac{9|\Fbb|}{10}$ possibilities out of $|\Fbb|$ total choices off $i$.
Recall \Cref{itm:bivariate_1}. For at least $\frac9{10}-\frac1{200}>\frac45$ fraction of $i$'s, we have $f(x+i\cdot y_1)=R(i,0)=Q(i,0)$.
Note that the distance between any two distinct (parallel) degree-$d$ polynomial is at least $1-\frac d{|\Fbb|}$, which is at least $\frac23<2\cdot\frac15$ since $|\Fbb|\ge3d$.
In addition, $Q(\cdot,0)$ is parallel degree-$d$.
Thus $Q(\cdot,0)$ is the closest parallel degree-$d$ polynomial of $f|_{\ell_{x,y_1}}$, i.e., $Q(i,0)=f_\Lbb(\ell_{x,y_1})(x+i\cdot y_1)=C(i,0)$ holds for all $i\in\Fbb$ where we recall \Cref{eq:lem:derand_pldt_3}.
In particular, this means $Q(0,0)=C(0,0)$.

Similarly, using \Cref{itm:bivariate_2}, we can show $Q(0,0)=R(0,0)$, which verifies \Cref{eq:lem:derand_pldt_5} and thus \Cref{itm:lem:derand_pldt_3}, and completes the proof of \Cref{lem:derand_pldt}.
\end{proof}

Finally it remains to prove \Cref{clm:derand_bivariate}.
\begin{proof}[Proof of \Cref{clm:derand_bivariate}]
We show that each item holds with probability at least $1-2^{30}\cdot\gamma\cdot\pbra{\frac1{|\Fbb|}+\lambda}$, and then \Cref{clm:derand_bivariate} follows from a union bound.

We first consider \Cref{itm:bivariate_1} and the analysis for \Cref{itm:bivariate_2} is almost identical.
For each $y\in S$, define
$$
B_y=\cbra{x\in\Fbb^m\colon f_\Lbb(\ell_{x,y})(x)\neq f(x)}.
$$
Then for any fixed $y_2\in T'$, we have $|B_{y_2}|\le\gamma\cdot|\Fbb|^m$ by \Cref{eq:lem:derand_pldt_1}, and, moreover,
\begin{align*}
&\Pr_{x\sim\Fbb^m,y_1\sim T'}\sbra{\text{\Cref{itm:bivariate_1} does not hold}}\\
=&\Pr_{x\sim\Fbb^m,y_1\sim T'}\sbra{\Pr_{i\sim\Fbb}\sbra{f_\Lbb(\ell_{x+i\cdot y_1,y_2})(x+i\cdot y_1)\ne f(x+i\cdot y_1)}>\frac1{200}}\\
=&\Pr_{x\sim\Fbb^m,y_1\sim T'}\sbra{\Pr_{i\sim\Fbb}\sbra{x+i\cdot y_1\in B_{y_2}}>\frac1{200}}
=\Pr_{x\sim\Fbb^m,y_1\sim T'}\sbra{\frac{|\ell_{x,y_1}\cap B_{y_2}|}{|\ell_{x,y_1}|}>\frac1{200}}\\
\le&\Pr_{x\sim\Fbb^m,y_1\sim T'}\sbra{\frac{|\ell_{x,y_1}\cap B_{y_2}|}{|\ell_{x,y_1}|}>\frac{|B_{y_2}|}{|\Fbb|^m}+\frac1{400}}
\tag{since $\frac{|B_{y_2}|}{|\Fbb|^m}\le\gamma$ and $\gamma\le2^{-20}$ by assumption}\\
\le&\frac{|S|}{|T'|}\cdot\Pr_{x\sim\Fbb^m,y_1\sim S}\sbra{\frac{|\ell_{x,y_1}\cap B_{y_2}|}{|\ell_{x,y_1}|}>\frac{|B_{y_2}|}{|\Fbb|^m}+\frac1{400}}
\tag{since $T'\subseteq T\subseteq S$}\\
\le&4\Pr_{x\sim\Fbb^m,y_1\sim S}\sbra{\frac{|\ell_{x,y_1}\cap B_{y_2}|}{|\ell_{x,y_1}|}>\frac{|B_{y_2}|}{|\Fbb|^m}+\frac1{400}}
\tag{by assumption and \Cref{itm:lem:derand_pldt_1}}\\
\le&2^{20}\cdot\pbra{\frac1{|\Fbb|}+\lambda}\cdot\frac{|B_{y_2}|}{|\Fbb|^m}
\tag{by \Cref{lem:derand_sampling}}\\
\le&2^{20}\cdot\gamma\cdot\pbra{\frac1{|\Fbb|}+\lambda}.
\tag{since $\frac{|B_{y_2}|}{|\Fbb|^m}\le\gamma$}
\end{align*}
By taking the expectation over $y_2$, \Cref{itm:bivariate_1} holds with the desired probability.

Now we consider \Cref{itm:bivariate_3} and the analysis for \Cref{itm:bivariate_4} is almost identical.
For each $y\in S$, define
\begin{equation}\label{eq:clm:derand_bivariate_3}
\bar B_y=\cbra{x\in\Fbb^m\colon\Pr_{j\sim\Fbb}\sbra{f_\Lbb(\ell_{x,y})(x+j\cdot y)\neq f(x+j\cdot y)}>2^{-10}}.
\end{equation}
For each $y_2\in S$, we have
\begin{align}
\frac{|\bar B_{y_2}|}{|\Fbb|^m}
&=\Pr_{x\sim\Fbb^m}\sbra{\Pr_{j\sim\Fbb}\sbra{f_\Lbb(\ell_{x,y_2})(x+j\cdot y_2)\neq f(x+j\cdot y_2)}>2^{-10}}
\notag\\
&\le2^{10}\cdot\Pr_{x\sim\Fbb^m,j\sim\Fbb}\sbra{f_\Lbb(\ell_{x,y_2})(x+j\cdot y_2)\neq f(x+j\cdot y_2)}
\tag{by Markov's inequality}\\
&=2^{10}\cdot\Pr_{x\sim\Fbb^m,j\sim\Fbb}\sbra{f_\Lbb(\ell_{x+j\cdot y_2,y_2})(x+j\cdot y_2)\neq f(x+j\cdot y_2)}
\tag{since $\ell_{x,y_2}=\ell_{x+j\cdot y_2,y_2}$}\\
&=2^{10}\cdot\Pr_{z\sim\Fbb^m}\sbra{f_\Lbb(\ell_{z,y_2})(z)\neq f(z)}=2^{10}\cdot\frac{|B_{y_2}|}{|\Fbb|^m}
\notag\\
&\le2^{10}\cdot\gamma,
\label{eq:clm:derand_bivariate_1}
\end{align}
and
\begin{align}
&\Pr_{x\sim\Fbb^m,y_1\sim T'}\sbra{\text{\Cref{itm:bivariate_3} does not hold}}
\notag\\
=&\Pr_{x\sim\Fbb^m,y_1\sim T'}\sbra{\Pr_{i,j\sim\Fbb}\sbra{f_\Lbb(\ell_{x+i\cdot y_1,y_2})(x+i\cdot y_1+j\cdot y_2)\ne f(x+i\cdot y_1+j\cdot y_2)}>\frac1{200}}
\notag\\
\le&\Pr_{x\sim\Fbb^m,y_1\sim T'}\sbra{\Pr_{i\sim\Fbb}\sbra{x+i\cdot y_1\in\bar B_{y_2}}>\frac1{300}},
\label{eq:clm:derand_bivariate_2}
\end{align}
where we use the following reasoning for the last inequality: if the event inside bracket does not happen, then
\begin{align*}
&\Pr_{i,j\sim\Fbb}\sbra{f_\Lbb(\ell_{x+i\cdot y_1,y_2})(x+i\cdot y_1+j\cdot y_2)\ne f(x+i\cdot y_1+j\cdot y_2)}\\
\le&\frac1{300}+\Pr_{i,j\sim\Fbb}\sbra{f_\Lbb(\ell_{x+i\cdot y_1,y_2})(x+i\cdot y_1+j\cdot y_2)\ne f(x+i\cdot y_1+j\cdot y_2)\mid x+i\cdot y_1\notin\bar B_{y_2}}\\
\le&\frac1{300}+2^{-10}<\frac1{200}.
\tag{by \Cref{eq:clm:derand_bivariate_3}}
\end{align*}
Then similar to the analysis for \Cref{itm:bivariate_1}, we continue upper bounding \Cref{eq:clm:derand_bivariate_2} as follows:
\begin{align*}
\text{RHS of \Cref{eq:clm:derand_bivariate_2}}
&=\Pr_{x\sim\Fbb^m,y_1\sim T'}\sbra{\frac{|\ell_{x,y_1}\cap\bar B_{y_2}|}{|\ell_{x,y_1}|}>\frac1{300}}\\
&\le\Pr_{x\sim\Fbb^m,y_1\sim T'}\sbra{\frac{|\ell_{x,y_1}\cap\bar B_{y_2}|}{|\ell_{x,y_1}|}>\frac{|\bar B_{y_2}|}{|\Fbb|^m}+\frac1{400}}
\tag{by \Cref{eq:clm:derand_bivariate_1} and $\gamma\le2^{-20}$}\\
&\le4\Pr_{x\sim\Fbb^m,y_1\sim S}\sbra{\frac{|\ell_{x,y_1}\cap\bar B_{y_2}|}{|\ell_{x,y_1}|}>\frac{|\bar B_{y_2}|}{|\Fbb|^m}+\frac1{400}}
\tag{since $T'\subseteq S$ and $|T'|\ge|S|/4$}\\
&\le2^{20}\cdot\pbra{\frac1{|\Fbb|}+\lambda}\cdot\frac{|\bar B_{y_2}|}{|\Fbb|^m}\le2^{30}\cdot\gamma\cdot\pbra{\frac1{|\Fbb|}+\lambda}.
\tag{by \Cref{lem:derand_sampling} and \Cref{eq:clm:derand_bivariate_1}}
\end{align*}
Therefore \Cref{itm:bivariate_3} holds with the desired probabilit by taking the expectation over $y_1$.
\end{proof}

\subsection{Derandomized Characterizations of Parallel Low Degree Polynomials}\label{sec:derand_char_ldp}

In this part, we extend the characterization of low degree polynomials in \cite{rubinfeld1996robust} to the derandomized setting, where we only consider lines with directions generated by a (large subset of)  small-biased set.

We use superscript $\top$ to denote vector and matrix transpose. For two vectors $u,v\in\Fbb^d$, we use $\abra{u,v}$ to denote their inner product which equals $u^\top v$ (or $v^\top u$).

\begin{theorem*}[\Cref{thm:exact_derand_ldp} Restated]
Let $S\subseteq\Fbb^m$ be a $\lambda$-biased set and $T\subseteq S$ of size $|T|>\frac{1+\lambda}2\cdot|S|$.
Then $f\colon\Fbb^m\to\Fbb^t$ is parallel degree-$md$ if $f|_{\ell_{x,y}}$ is parallel degree-$d$ for every $x\in\Fbb^m,y\in T$.
\end{theorem*}
\begin{proof}
Assume without loss of generality that $t=1$, since we can apply the analysis individually for each single-output component $f[1],\ldots,f[t]$.

We first prove that $T$ has the full rank $m$.
Assume towards contradiction that $T$ has rank at most $m-1$.
Then there exists a non-zero vector $z\in\Fbb^m$ such that $z^\top y=0$ holds for all $y\in T$.
In light of \Cref{def:eps_biased_set}, we construct a non-trivial homomorphism $\chi\colon\Fbb^m\to\mu_p$ to derive a contradiction, where $p$ is the characteristic of $\Fbb$.
Let $\xi\colon\Fbb\to\mu_p$ be a non-trivial (group) homomorphism, where we view $\Fbb$ as an additive group.
We define for each $x\in\Fbb^m$ that 
$$
\chi(x)=\xi(x^\top z),
$$
which is a non-trivial homomorphism since $\xi$ is a non-trivial homomorphism and $z\ne0^m$.
Note that for all $y\in T$, we have
$$
\chi(y)=\xi(y^\top z)=\xi(0)=1.
$$
Thus
\begin{align*}
\abs{\E_{y\sim S}\sbra{\chi(y)}}
&=\abs{\E_{y\sim S}\sbra{\chi(y)\cdot\pbra{\indicator_{y\in T}+\indicator_{y\notin T}}}}
=\abs{\E_{y\sim S}\sbra{\indicator_{y\in T}+\chi(y)\cdot\indicator_{y\notin T}}}\\
&\ge\E_{y\sim S}\sbra{\indicator_{y\in T}}-\E_{y\sim S}\sbra{\indicator_{y\notin T}}
=\frac{|T|}{|S|}-\frac{|S|-|T|}{|S|}\\
&>\lambda,
\tag{since $|T|>\frac{1+\lambda}2\cdot|S|$}
\end{align*}
which contradicts \Cref{def:eps_biased_set}.

Now we fix a set of $m$ linearly independent directions $y_1,\ldots,y_m\in T$.
Then we can interpolate $f$ using degree-$d$ polynomials $f|_{\ell_{x,y_1}},\ldots,f|_{\ell_{x,y_m}}$ for $x\in\Fbb^m$.
For concreteness, we apply the invertible linear transform on $\Fbb^m$ such that $y_1,\ldots,y_m$ map to axis parallel directions $e_1,\ldots,e_m$, where $e_i=(\underbrace{0,\ldots,0}_{i-1},1,\underbrace{0,\ldots,0}_{m-i})$.
Let $f'$ be the polynomial after this transform, which shares the same degree with $f$ since the transform is invertible and linear.
Since $f|_{\ell_{x,y_1}},\ldots,f|_{\ell_{x,y_m}}$ are all degree-$d$, we know that $f'$ is degree-$d$ along the axis parallel lines.
Then by polynomial interpolation (see e.g., \cite[Lemma 28]{rubinfeld1996robust}), $f'$ has degree at most $md$, which in turn means that $f$ has degree at most $md$ as claimed.
\end{proof}

\end{document}